\documentclass[11pt]{article}
\usepackage[utf8]{inputenc}
\usepackage{amsfonts,epsfig}
\usepackage[hyphens]{url}
\RequirePackage{color}
\pdfoutput=1
\usepackage{hyperref}
\hypersetup{
	colorlinks=true,
	linkcolor=black,
	citecolor=black,
	filecolor=black,
	urlcolor=black,
}
\usepackage[linesnumbered,ruled]{algorithm2e}
 \usepackage{tikz}
\usetikzlibrary{shapes,arrows,chains}
\usetikzlibrary[calc]
\usetikzlibrary{bayesnet}
\usepackage{comment}
\graphicspath{{fig/}}
\usepackage{relsize}
\usepackage{amssymb}
 \usepackage{extarrows}
 \usepackage{amsthm}
\usepackage{geometry}
\usepackage{sectsty}
\usepackage{caption, setspace}
\usepackage{wrapfig}
\usepackage{mathtools}
\usepackage{bbm}

\DeclareMathOperator{\E}{\mathrm{E}}

\DeclareMathOperator{\R}{\mathbb{R}}
\DeclareMathOperator{\n}{\mbox{normal}}
\DeclareMathOperator{\X}{\mathcal{X}}
\usepackage[normalem]{ulem}
\sectionfont{\sffamily\bfseries\upshape\large}
\subsectionfont{\sffamily\bfseries\upshape\normalsize} 
\subsubsectionfont{\sffamily\mdseries\upshape\normalsize}
\paragraphfont{\sffamily \mdseries \itshape  \normalsize}
\makeatletter
\renewcommand\@seccntformat[1]{\csname the#1\endcsname.\quad}
 \usepackage[defaultsans, scale=0.98]{opensans} 
\usepackage[T1]{fontenc}
\usepackage[round]{natbib}
\newtheorem{theorem}{Theorem}

\newtheorem{repthm}{Theorem}

\DeclareMathOperator{\w}{\mathbf{w}}

\usepackage{amsmath}
\newcommand\mycom[2]{\genfrac{}{}{0pt}{}{#1}{#2}}

\usepackage{listings}
\definecolor{codegreen}{rgb}{0,0.6,0}
\definecolor{codegray}{rgb}{0.5,0.5,0.5}
\definecolor{codepurple}{rgb}{0.58,0,0.82}
\definecolor{backcolour}{rgb}{0.99,0.99,0.97}

\lstdefinestyle{R2}{
	language=R,
	backgroundcolor=\color{backcolour},   
	commentstyle=\color{codegreen},
	keywordstyle=\color{black},
	numberstyle=\tiny\color{codegray},
	stringstyle=\color{codepurple},
	emphstyle=\color{black},%
	basicstyle={\small \ttfamily},
	breakatwhitespace=false,         
	breaklines=true,                 
	captionpos=b,                    
	keepspaces=true,                 
	numbers=left,                    
	numbersep=5pt,                  
	showspaces=false,                
	showstringspaces=false,
	showtabs=false,                  
	tabsize=2,
	morecomment=[f][0][\color{codegreen}]{\#}
}

\lstdefinestyle{stan}{
	literate={~}{$\sim$}{1},
	otherkeywords={real, simplex,  matrix, data,  parameters, transformed, model\{, vector, generated, quantities},
	backgroundcolor=\color{backcolour},   
	commentstyle=\color{codegreen},
	keywordstyle=\color{magenta},
	numberstyle=\tiny\color{codegray},
	stringstyle=\color{codepurple},
	emph={  
		normal, cauchy, inv_gamma, bernoulli_logit, log, softmax, std_normal, append_row, to_vector, exp,rep_vector
	},
	emphstyle=\color{codepurple},%
	basicstyle={\small \ttfamily},
	breakatwhitespace=false,         
	breaklines=true,                 
	captionpos=b,                    
	keepspaces=true,                 
	numbers=left,                    
	numbersep=5pt,                  
	showspaces=false,                
	showstringspaces=false,
	showtabs=false,                  
	tabsize=2
}
\usepackage{inconsolata}
\usepackage{enumitem}

\makeatletter
\def\@maketitle{%
	\begin{center}%
		\let \footnote \thanks
		{\large \@title \par}%
		{\normalsize
			\begin{tabular}[t]{c}%
				\@author
			\end{tabular}\par}%
		{\small \@date}%
	\end{center}%
}
\makeatother

\title{
	\bf Bayesian hierarchical stacking:  Some models are (somewhere) useful \vspace{.1in}  
}
\author{
	Yuling Yao\footnote{Department of Statistics, Columbia University, New York, USA. {\small\texttt{yy2619@columbia.edu}.}}
	 \and Gregor Pirš\footnote{Faculty of  Computer and Information Science,  University of Ljubljana, Ljubljana, Slovenia.} 
	  \and Aki Vehtari\footnote{Department of  Computer Science,  Aalto University, Espoo, Finland.} 
	 	\and Andrew Gelman\footnote{Department of Statistics and Political Science, Columbia University, New York, USA.}
} 

\date{\vspace{.1in}  20 May 2021  \vspace{-.15in} }

\begin{document}\sloppy
	\maketitle
\thispagestyle{empty}

\begin{abstract}
			Stacking is a widely used model averaging technique that asymptotically yields  optimal predictions among  linear averages.  We show that stacking is most effective when model predictive performance is heterogeneous in inputs, and we can further improve the stacked mixture  with a hierarchical model. 
We generalize  stacking to Bayesian hierarchical stacking. The model weights are varying as a function of data, partially-pooled, and inferred using Bayesian inference.  
%Our Bayesian formulation includes constant-weight (complete-pooling) stacking as a special case.
We further incorporate discrete and continuous inputs, other structured priors, and time series and longitudinal data. 
To verify the performance gain of the proposed method, we 
derive  theory bounds, and demonstrate  on several applied problems.
	
\end{abstract}

 \section{Introduction}
 Statistical inference is conditional on the model, and a general challenge is how to make full use of multiple candidate models. 
 Consider data $\mathcal{D}=(y_i\in \mathcal{Y}, x_i\in \mathcal{X})_{i=1}^n$, and $K$ models $M_1, \dots, M_k$, each having its own parameter vector $\theta_k \in \Theta_k$, likelihood, and prior.  We fit each model and obtain posterior predictive distributions, 
 \begin{equation}\label{eq_select}
 	p(\tilde y| \tilde  x, M_k)= \int_{\Theta_k} p(\tilde y|\tilde  x,  \theta_k, M_k ) p(\theta_k|\{y_i, x_i\}_{i=1}^n, M_k)\, d\theta_k.
 \end{equation}
 The model fit is judged by its expected predictive utility of future (out-of-sample) data $(\tilde y, \tilde x)\in \mathcal{Y}\times \mathcal{X}$, which generally have an unknown \emph{true} joint density $p_t(\tilde y, \tilde x)$. Model selection seeks the best model with the highest utility when averaged over $p_t(\tilde y, \tilde x)$. Model averaging assigns models with weight $w_1, \dots, w_K$ subject to a simplex constraint $\w\in \mathcal{S}_K= \{\w: \sum_{k=1}^K w_k=1;  w_k\in [0,1],  \forall k\}$, and  the future prediction is a linear mixture from individual models:
 \begin{equation}\label{eq_linear_ave}
 	p(\tilde y| \tilde x, \w, \mathrm{model ~ averaging})=\sum_{k=1}^K w_k p(\tilde y|\tilde  x, M_k), ~ \w\in \mathcal{S}_{K}.
 \end{equation}
 Stacking \citep{wolpert1992stacked}, among other ensemble-learners, has been successful for various  prediction  tasks. \cite{yao2018using} apply the stacking idea to combine predictions from separate Bayesian inferences.  The first step is to fit each individual model 
 and evaluate the pointwise leave-one-out  predictive density of each data point $i$ under each model $k$:
 $$p_{k,-i} = \int_{\Theta_k}  p (y_i|\theta_k, x_i, M_k) p\left(\theta_k| M_k, \{(x_{i^\prime}, y_{i^\prime}) : {i^{\prime}\neq i}\}  \right) d\theta_k,$$
 which in a Bayesian context we can approximate using posterior simulations and Pareto-smoothed importance sampling \citep{vehtari2017practical}. Reusing data eliminates the need to model the unknown joint density $p_t(\tilde y, \tilde x)$.
 The next step is to determine the vector\footnote{We use the bold letter $\w$, or $\w(\cdot)$ to reflect that the weight is vector, or a vector of functions.} of weights $\w=(w_1,\dots,w_K)$ that optimize the average log score of    the stacked prediction,
 \begin{equation}\label{stacking_obj}
 	\hat {\w}^{\mathrm{stacking}}=\arg\max_{\w}\sum_{i=1}^n \log \left(\sum_{k=1}^K w_k p_{k,-i}\right), \mbox{ such that } \w\in \mathcal{S}_{K}.
 \end{equation}
 
 However, the linear mixture \eqref{eq_linear_ave} restricts  an identical set of weights for all input $x$. We will later label this solution \eqref{stacking_obj} as \emph{complete-pooling stacking}.  The present paper proposes \emph{hierarchical stacking}, an approach that goes further in three ways:
 \begin{enumerate}[noitemsep,topsep=0pt,  itemsep=4pt]
 	\item Framing the estimation of the stacking weights as a Bayesian inference problem rather than a pure optimization problem. This in itself does not make much difference in the complete-pooling estimate \eqref{stacking_obj} but is helpful in the later development.
 	\item Expanding to a hierarchical model in which the stacking weights can vary over the population.  If the model predictors $x$ take on $J$ different values in the data, we can use Bayesian inference to estimate a $J\! \times\! K$ matrix of weights that partially pools the data both in row and column.
 	\item Further expanding to allow weights to vary as a function of continuous predictors.  This idea generalizes the feature-weighted linear stacking \citep{sill2009feature} with a more flexible form and Bayesian hierarchical shrinkage.
 \end{enumerate}
 
 There are two reasons we would like to consider input-dependent model weights.
 First, the scoring rule measures the expected predictive performance averaged over $\tilde x$ and $\tilde y$, as the objective function in \eqref{stacking_obj} divided by $n$ is a consistent estimate of  $\E_{\tilde x, \tilde y} \log \left(\sum_{k=1}^K w_k p(\tilde y | \tilde x, \mathcal{D}, M_k)\right)$.
 But an overall good model fit does not ensure a good conditional prediction at a given location $\tilde x=\tilde x_0$, or under covariate shift when the distribution of input $x$ in the observations differs from the population of interest. 
 More importantly, different models can be good at explaining different regions in the input-response space, which is why model averaging can be a better solution to model selection. 
 Even if we are only interested in the average  performance, we can further improve model averaging by learning \emph{where} a model is good so as to \emph{locally} inflate its weight.

 In Section \ref{sec_method}, we develop  detailed implementation
 of hierarchical stacking. We explain why it is legitimate to convert an optimization problem into a formal Bayesian  model. With hierarchical shrinkage, we partially pool the stacking weights across data. By varying priors, hierarchical stacking includes classic  stacking and selection as special cases.  We generalize this approach to continuous input variables,  other structured priors, and time-series and longitudinal data. In Section \ref{sec_bound},  we turn heuristics from the previous paragraph into a rigorous learning bound, indicating the benefit from model selection to model averaging, and from complete-pooling model averaging to a local averaging that allows the model weights to vary in the population.
 We outline related work in Section \ref{sec_related}. In Section \ref{sec_examles}, we evaluate the proposed method  in several simulated and  real-data examples,  including a U.S. presidential election forecast.   %We discuss more on the hierarchical shrinkage,  connection to full-Bayesian inference, and future directions in Section \ref{sec_discuss}. 
 
 This paper makes two main contributions:
 \begin{itemize}[noitemsep,topsep=0pt,  itemsep=4pt]
 	\item Hierarchical stacking provides a  Bayesian recipe for model averaging with input-dependent weights and hierarchical regularization. It is beneficial for both improving the overall model fit, and the conditional local fit in small and new areas.
 	\item Our theoretical results characterize how the model list should be locally separated to be useful in model averaging and local model averaging.
 	%\item We use stacking not just as a tool for optimizing predictions but also as a way to understand problems with fitted models.
 	
 	% I move the last point into discussion. I don' think it is exactly true. Stacking weight is not necessarily a proxy for model fit. For example, in the theoretical example in the appendix, when two models become closer and  closer, their weights are more polarized. 
 \end{itemize}

 \section{Hierarchical stacking}\label{sec_method}
 The present paper generalizes the linear model averaging \eqref{eq_linear_ave} to 
 pointwise model averaging. The goal is to construct an  input-dependent model weight function $\w(x)=(w_1(x), \dots, w_K(x)): \mathcal{X} \to \mathcal{S}_K$, and combine the predictive densities pointwisely by 
 \begin{equation}\label{eq_point_stacking}
 	p(\tilde y| \tilde x, \w(\cdot), ~\mathrm{pointwise~ averaging} ) = \sum_{k=1}^K w_k(\tilde x) p(\tilde y| \tilde x,  M_k), \mbox{ such that } \w(\cdot) \in \mathcal{S}_K^{\mathcal{X}}.
 \end{equation}
 If the input is discrete and has finite categories,  one na\"ive estimation of the pointwise optimal  weight is to run complete-pooling stacking \eqref{stacking_obj} separately on each category, which we will label \emph{no-pooling  stacking}. The no-pooling procedure generally has a larger variance and overfits the data.
 
 From a  Bayesian perspective, it is  natural to compromise between unpooled and completely pooled procedures by  a hierarchical model.  Given some hierarchical prior $p^{\mathrm{prior}} \left(\cdot  \right)$, we  define  the posterior distribution of the stacking weights $w\in \mathcal{S}_K^{\mathcal{X}}$ through the usual likelihood-prior protocol:
 \begin{equation}\label{stacking_obj_h}
 	\log p\left(\w(\cdot) | \mathcal{D}\right) =  \sum_{i=1}^n \log \left(\sum_{k=1}^K w_k(x_i) p_{k,-i} \right)+ \log p^{\mathrm{prior}} \left( \w \right) + \mathrm{constant},~~ \w(\cdot) \in \mathcal{S}_K^{\mathcal{X}}.
 \end{equation}
 The final estimate of the pointwise stacking weight  used in  \eqref{eq_point_stacking} is then the posterior mean from this joint density $\E (\w(\cdot) | \mathcal{D})$. We call this approach \emph{hierarchical stacking}.

 \subsection{Complete-pooling and no-pooling stacking}
 For notational consistency, we rewrite the input variables into two groups $(x, z)$, where $x$ are variables on which the model weight $w(x)$ depends   during model averaging \eqref{eq_point_stacking}, and $z$ are all remaining input variables. 
 
 To start, we consider $x$ to be discrete and has $J<\infty$  categories,  $x=1, \dots J$. We will extend to continuous and hybrid $x$ later. The input varying stacking weight function is parameterized by  a $J\! \times\! K$ matrix  $\{w_{jk}\} \in {\mathcal{S}_K^J}$: Each row of the matrix is an element of the length-$K$ simplex.   The $k$-th model in cell $j$ has the weight $w_k(x_i) = w_{jk}, \forall x_i=j$. We fit  each individual model $M_k$ to all observed data $\mathcal{D}=(x_i,z_i, y_i)_{i=1}^n$ and obtain pointwise leave-one-out cross-validated log predictive densities: 
 \begin{equation}\label{eq_pki}
 	p_{k,-i}  \coloneqq  \int_{\Theta_k} p(y_i | \theta_k, x_i, z_i, M_k ) p (\theta_k|  \{(x_{l},y_{l},z_{l}): l\neq i\}, M_k)d \theta_k.  
 \end{equation}
 Same as in complete-pooling stacking, here we avoid refitting each model $n$ times, and instead use the Pareto smoothed importance sampling \citep[PSIS,][]{vehtari2017practical, vehtari2015pareto} to approximate $\{p_{k,-i}\}_{i=1}^{n}$ from one-time-fit posterior draws $p(\theta_k|M_k, \mathcal{D})$. The cost of such approximate leave-one-out cross validation is often negligible compared with individual model fitting. 
 
 To optimize the expected predictive performance  of the pointwisely combined  model averaging, we can maximize the  leave-one-out predictive density
 \begin{equation}\label{eq_no-pooling}
 	\max_{\w(\cdot)} \sum_{i=1}^n \log \left(\sum_{k=1}^K w_{k}(x_i) p_{k,-i}\right).
 \end{equation}
 
 On one extreme, the complete-pooling 
 stacking \eqref{stacking_obj} solves optimization \eqref{eq_no-pooling} subject to a constant constraint 
 $w_{k}(x)=w_{k}(x^{\prime}), \forall k, x, x^{\prime}$.
 %which is equivalent to solving \eqref{stacking_obj} with all data $\mathcal{D}$. 
 On the other extreme,  no-pooling stacking maximizes this objective function \eqref{eq_no-pooling} without extra constraint other than the row-simplex-condition, which amounts to  separately solving complete-pooling stacking \eqref{stacking_obj} on each input cell $\mathcal{D}_j= \{(x_i,z_i, y_i): x_i=j\}$. 
 
 %\begin{equation}\label{eq_no-pooling}
 %w^{\mathrm{stacking}}_\mathrm{no-pooling}=\{ w_{jk}^{\mathrm{no-pooling}}\}_{j=1, k=1}^{J\quad  K}, \quad 
 %w_{j1,\dots,jK}^{\mathrm{no-pooling}} =\arg\max_{w \in \mathcal{S}_K}\sum_{i: x_i=j} \log \sum_{k=1}^K w_k p_{k,-i}.
 %\end{equation}

 If there are a large number of repeated measurements in each cell, 
 $n_j \coloneqq ||\{i: x_i=j\} ||\to \infty$, 
 then $\frac{1}{n_j} \sum_{i: x_i=j} \log \sum_{k=1}^K w_k(j) p_{k,-i}$ becomes a reasonable estimate of the conditional log predictive density 
 $\int_{\mathcal{Y}} p_t(\tilde y | \tilde x=j)  \log    \sum_{k=1}^K w_k(j) p(\tilde y|j,  M_k)  d \tilde y$, with convergence rate $\sqrt{n_j}$,
 and therefore, no-pooling stacking becomes asymptotically optimal among all cell-wise combination weights. 
 For finite sample size, because the cell size is smaller than total sample size, we would expect a larger variance in no-pooling stacking than in complete-pooling stacking. Moreover, the cell sizes are often not balanced, which entails a large noise of no-pooling stacking weight in  small cells.
 
 % the distribution of $x$ in the sample differs from these predictors' distribution in the population of interest. 
 
 \subsection{Bayesian inference for stacking weights}
 Vanilla (optimization-based) stacking \eqref{stacking_obj} is  justified by \emph{Bayesian decision theory}: the expected log predictive density of the  combined model  $\E_{\tilde y} \log \left(\sum_{k=1}^K w_k  p(\tilde y |M_k)\right)$ is estimated by leave-one-out $\frac{1}{n} \sum_{i=1}^n \log \left(\sum_{k=1}^K w_k p_{k,-i}\right)$.    
 The point optimum   asymptotically  maximizes the expected utility \citep{le2017bayes}, hence is an $M_*$-optimal decision in terms of \citet{vehtari2012survey}.

 %Apart from the input-varying yet partially-pooled weight, hierarchical stacking differs from vanilla stacking  for we treat the objective function  as a log likelihood function, as if the  pointwise fits $\{p_{k,-i}\}$ are observations. Is hierarchical stacking  a legitimate \emph{Bayesian inference}?

 To fold stacking into a \emph{Bayesian inference problem}, we want to treat the objective function in \eqref{eq_no-pooling} as a log likelihood with parameter $\w$.  
 %It looks like a multinomial log likelihood, except that $p_{k, j}$ are neither integers nor observations.   Rigorously, this log likelihood is and outcomes $y_{1, \dots, n}$.
 After integrating out individual-model-specific parameters $\theta_k$  such that $p(y| x, M_k)$ is given,   the   outcomes $y_i$ at input location  $x_i$ in the combined model  have  densities 
 $p(y_i| x_i,    w_k(x_i)) = \sum_{k=1}^K      w_k(x_i)  p(y_i| x_i,   M_k)$, which  implies a  joint log likelihood:  $\sum_{i=1}^n \log \left(\sum_{k=1}^K      w_k(x_i)  p(y_i| x_i,   M_k)\right)$. 
 But this procedure has used data twice---in other practices, data are often used twice to pick the prior, whereas here data are used twice to pick the likelihood. 
 
 We  use a two-stage estimation procedure to  avoid reusing data. 
 Assuming a hypothetically provided holdout dataset
 $  \mathcal{D}^{\prime}$ of the same size and identical distribution as observations $ \mathcal{D}= \{y_i,x_i\}_{i=1}^n$, we can use $ \mathcal{D}^{\prime}$ to fit the individual model first and compute  $\tilde p(y_{i} | x_i, M_k, \mathcal{D}^{\prime})= \int\! p(y_{i} | x_i, M_k, \theta_k) p(\theta_k |  M_k,   \mathcal{D}^{\prime}) d\theta_k$. 
 In the second stage we plug in the observed $y_i$, $x_i$, and obtain the pointwise full likelihood $p(y_i|\w, \mathcal{D}^\prime, x_i) = \sum_{k=1}^K     w_k(x_i)   \tilde   p(y_{i} | x_i, M_k, \mathcal{D}^{\prime})$.
 
 Now in lack of holdout data $\mathcal{D}^\prime$,   the  leave-$i$-th-observation-out predictive density $p_{k,-i}$ is a consistent estimate of the 
 pointwise out-of-sample predictive density $\E_{\mathcal{D}^\prime}  \left(\tilde p(y_{i} | x_i, M_k, \mathcal{D}^{\prime})\right)$. By plugging it into the two-stage log likelihood and integrating out the unobserved holdout data $\mathcal{D}^\prime$,  we get a profile likelihood 
 $$ p(y_i |\w, x_i)\coloneqq \E_{\mathcal{D}^\prime} \left(p(y_i|\w, x_i, \mathcal{D}^\prime)\right)
 = \sum_{k=1}^K     w_k(x_i) \E_{\mathcal{D}^\prime} \left( p(y_{i} | x_i, M_k, \mathcal{D}^{\prime})\right) \approx  \sum_{k=1}^K     w_k(x_i)    p_{k,-i}.$$
 Summing over $y_i$ arrives at $\log \left(p(\mathcal{D}|\w)\right)\approx \sum_{i=1}^n \log \left(\sum_{k=1}^K w_{k}(x_i) p_{k,-i}\right)$. This log likelihood coincides with the no-pooling  optimization objective function \eqref{eq_no-pooling}. 
 %Therefore, no-pooling stacking can also be interpreted as its maximum likelihood estimate.
 
 %Such coincidence only applies for evaluating the model predictive performance by logarithmic scoring rule.  
 Integrating out the hypothetical data $\mathcal{D}^\prime$ is related to the idea of marginal data augmentation \citep{meng1999seeking}. \citet{polson2011data} took a similar approach to convert the optimization-based support vector machine into a Bayesian inference.

 \subsection{Hierarchical stacking: discrete inputs}\label{sec_discrete_method}
 The log posterior density of hierarchical stacking model  \eqref{stacking_obj_h} contains the log likelihood defined above $\sum_{i=1}^n \log \left(\sum_{k=1}^K w_k(x_i) p_{k,-i}\right)$, and a prior distribution on the weight matrix $\w= \{w_{jk}\} \in \mathcal{S}_K^{J}$, which we specify in the following.
 
 %we will reach the earlier defined  posterior  density of stacking weights $\log p\left(w \mid \mathcal{D}\right) =  \sum_{i=1}^n \log \left(\sum_{k=1}^K w_k(x_i) p_{k,-i} \right)+ p^{\mathrm{prior}} \left( w \right) + \mathrm{Constant},~~ w \in \mathcal{S}_K^J.$
 %The normalization constant therein ensures a probability density, but is ignoble in Bayesian computation.

 We first take a softmax transformation that bijectively converts the simplex matrix space $\mathcal{S}_K^{J}$ to unconstrained space  $\R^{J(K-1) }$:
 \begin{equation}\label{eq_soft_max}
 	w_{jk}= \frac{\exp(\alpha_{jk})}{  \sum_{k=1}^K \exp( {\alpha_{jk} } )},  ~ 1\leq k\leq K-1,~ 1\leq j \leq J;  \qquad \alpha_{jK}=0,   ~1\leq j\leq J.
 \end{equation}
 $\alpha_{jk}\in \R$ is interpreted as the log odds ratio of model $k$ with reference to $M_K$ in cell $j$. %It is a unconstrained  parameter except for the last row.
 
 We propose a normal hierarchical prior on the unconstrained model weights $(\alpha_{jk})_{k=1 }^{K-1}$ conditional on hyperparameters $\mu\in\R^{K-1}$ and $\sigma\in\R_+^{K-1}$,
 \begin{equation}\label{eq_prior}
 	\mathrm{prior:} \quad \alpha_{jk} \mid \mu_k, \sigma_k \sim \n (\mu_k, \sigma_k),  ~ k=1, \dots, K-1, ~ j=1, \dots, J.
 \end{equation}
 The prior partially pools unconstrained weights  toward the shared  mean  $(\mu_1, \dots, \mu_{K-1})$. The shrinkage effect depends on both the cell sample size $n_j$ (how strong the likelihood is in cell $j$), and the model-specific $\sigma_{k}$ (how much across-cell discrepancy is allowed in model $k$). If $\mu$ and $\sigma$ are given constants, and if the posterior distribution is summarized by its mode, then hierarchical stacking contains two special cases: 
 \begin{itemize}[noitemsep,topsep=0pt,  itemsep=2pt]
 	\item no-pooling stacking by a flat prior  $\sigma_k \to \infty, ~k=1, \dots, K-1$.
 	\item complete-pooling stacking by a concentration prior  $\sigma_k \to 0, ~k=1, \dots, K-1$.
 \end{itemize}
 It is possible to derive other structured priors. For example, a sparse prior   \citep[e.g.,][]{heiner2019structured} on simplex $(w_{j1}, \dots, w_{jK})$ will enforce a cell-wise selection.  % We  place priors on  $\alpha$ as this allows more flexibility than the often-default Dirichlet prior on $w_{j\cdot}$.  
 
 %The parameter $\sigma_k$ can be interpreted as the level of regularization across cells. 
 Instead  of  choosing fixed values, we view  $\mu$ and $\sigma$ as hyperparameters and aim for a full Bayesian solution: to describe the uncertainty of all parameters by their joint posterior distribution $p(\alpha, \mu, \sigma|\mathcal{D})$, letting the data to tell how much regularization is desired. 
 
 To accomplish this Bayesian inference,  we assign a hyperprior to $(\mu, \sigma)$:
 \begin{equation}\label{eq_hprior}
 	\mathrm{hyperprior:} \quad    \mu_k\sim \n(\mu_0, \tau_\mu),  \quad \sigma_k \sim \n^+(0, \tau_\sigma), \quad  k=1, \dots, K-1,
 \end{equation}
 where $\n^+(0, \tau_\sigma)$ stands for the half-normal distribution supported on $[0,\infty)$ with scale parameter $\tau_\sigma$. 
 
 Putting the pieces \eqref{stacking_obj_h}, \eqref{eq_prior},  \eqref{eq_hprior} together, up to a normalization constant that has been omitted,  we attain a joint posterior density of all free parameters $\alpha\in \R^{J\times K}, \mu\in \R^{K-1}, \sigma\in \R_{+}^{K-1}$:
 \begin{multline}\label{stacking_obj_h2}
 	\log p(\alpha, \mu, \sigma | \mathcal{D}) = \sum_{i=1}^n \log  \left( \sum_{k=1}^K w_k(x_i) p_{k,-i}\right) +\\   \sum_{k=1}^{K-1}\sum_{j=1}^J \log p^{\mathrm{prior}} \left( \alpha_{jk}|   \mu_k, \sigma_k \right)  \sum_{k=1}^{K-1} \log p^ {\mycom{\text{hyper}}{\text{prior}}} \left( \mu_k, \sigma_k \right).
 \end{multline}
 Unlike complete and no-pooling stacking, which are typically solved by optimization, the maximum a posteriori (MAP) estimate of \eqref{stacking_obj_h2} is not meaningful: the  mode is attained at the complete-pooling subspace $\alpha_{jk}=\mu_{k}, \sigma_k=0, \forall j, k,$ on which the  joint density is positive infinity.  Instead, we sample $(\alpha, \mu, \sigma)$ from this joint density \eqref{stacking_obj_h2} using Markov chain Monte Carlo  (MCMC) methods
 %Afterward, we  transform back model weights $w_{jk}$  using softmax \eqref{eq_soft_max} and 
 and compute the Monte Carlo  mean of posterior draws $\overline w_{jk}$, which we will call \emph{hierarchical stacking} weights. 
 
 The final posterior predictive density of outcome $\tilde y$ at any input location $(\tilde x, \tilde z)$ is  
 \begin{equation}\label{eq_point_stacking_final}
 	\mathrm{final~predictions: } ~   p(\tilde y| \tilde x, \tilde z, \mathcal{D} ) = \sum_{k=1}^K     \overline w_k(\tilde x)  \int_{\Theta_k} p(\tilde y|\tilde  x,  \tilde z, \theta_k,  M_k) p (\theta_k| M_k, \mathcal{D} )d\theta_k.
 \end{equation}
 Using a  point estimate $\overline w_{jk}$  is not a waste  of the joint simulation draws. Because equation \eqref{eq_point_stacking_final} is a linear expression on $w_k$, and because of the linearity of expectation, using $\overline \w$ is as good as using all simulation draws.
 %---Bayesian methods can be compatible with point estimates too. 
 Nonetheless, for the purpose of post-processing, approximate cross validation, and extra model check and comparison,  we will use all posterior simulation draws; see discussion in Section \ref{sec_stacking_post}.

 \subsection{Hierarchical stacking: continuous and hybrid inputs} 
 The next step is to include more structure in the weights, which could correspond to regression for continuous predictors, nonexchangeable models for nested or crossed grouping factors, nonparametric prior,  or combinations of these.
 
 \subsubsection*{Additive model}\label{sec_additive}
 Hierarchical stacking is not limited to discrete cell-divider $x$. 
 When the input $x$ is continuous or hybrid,  
 one extension is to model the unconstrained weights additively:  
 \begin{align}\label{eq_continous_x}
 	w_{1:K}(x)&= \mathrm{softmax}(w^*_{1:K}(x)),\nonumber\\ w^*_k(x)&= \mu_k + \sum_{m=1}^M \alpha_{mk} f_{m}(x), ~k\leq K-1,  ~w^*_{K}(x)=0, 
 \end{align}
 where $\{f_{m}: \mathcal{X}\to \R\}$ are $M$ distinct features. Here we have already extracted  the prior mean $\mu_k$, representing the ``average'' weight of model $k$ in the unconstrained space.
 The discrete model  \eqref{eq_prior} is now equivalent to letting $f_{m}(x)= \mathbbm{1}(x=m)$ for $m=1, \dots, J$.  We may still use the basic prior \eqref{eq_prior} and hyperprior \eqref{eq_hprior}:
 \begin{equation}\label{eq_continous_x_prior}
 	\alpha_{mk} \mid   \sigma_k \sim \n (0, {\sigma}_k), \quad \mu_k\sim \n(\mu_0, \tau_\mu),  \quad \sigma_{k} \sim \n^+(0, \tau_\sigma).
 \end{equation}
 %to achieve partial pooling and hierarchical regularization, and sample $\alpha, \mu, \sigma| \mathcal{D}$  from the joint posterior density. 
 We provide \texttt{Stan} \citep{stan2020} code for this additive model and discuss practical hyperparameter choice in Appendix \ref{sec_stan}.
 
 Because the main motivation of our paper is to convert the one-fit-all model-averaging algorithm into open-ended Bayesian modeling,  the basic shrinkage prior above  should be viewed as a starting point for model building and improvement.  Without trying to exhaust all possible variants, we list a few useful prior structures:
 \begin{itemize}
 	\item \emph{Grouped hierarchical  prior}. 
 	The basic model  \eqref{eq_continous_x_prior} is limited to have a same regularization  $\sigma_k$ for all $\alpha_{mk}$.  When the features  $f_{m}(x)$ are grouped (e.g., $f_{m}$ are dummy variables from two discrete inputs; states are grouped in regions), we achieve group specific shrinkage by replacing \eqref{eq_continous_x_prior}  by 
 	$$\alpha_{mk}\mid \sigma_{gk} \sim \n(0,\sigma_{g[m]k}), ~ \mu_{k}\sim \n(\mu_0, \tau_\mu), ~ \sigma_{gk}\sim \n^+(0, \tau_\sigma),
 	$$ where $g[m]= 1, \dots G$ is the group index of feature $m$.
 	\item \emph{Feature-model   decomposition}. Alternatively we can learn  feature-dependent regularization by 
 	$$ \alpha_{mk} \mid \mu_k, \sigma_k,  \lambda_m\sim \n (0, \sigma_k \lambda_m),~\lambda_m\sim \mathrm{InvGamma}(a, b),  ~\sigma_{k}\sim \n^+(0, \tau_\sigma).$$
 	\item \emph{Prior correlation}. For discrete cells,  we would like to incorporate prior knowledge of the group-correlation. For example in  election forecast (Section \ref{sec_election}), we have a rough sense of some states being demographically close, and would expect a similar model weights therein. To this end,  we calculate a prior correlation matrix $\Omega_{J\times J}$ from  various sources of state level historical data, and  replace the independent prior \eqref{eq_prior} by a multivariate normal (MVN) distribution,
 	\begin{equation}\label{eq_corr}
 		(\alpha_{1k}, \dots,  \alpha_{jk}) \mid \sigma, \Omega, \mu  \sim \mathrm{MVN} \left((\mu_{k}, \dots, \mu_{k}), \mathrm{diag}(\sigma_{k}^2) \times \Omega~\right). 
 	\end{equation}
 	The prior correlation is especially useful to stabilize stacking weights in small cells.
 	\item \emph{Crude approximation of input  density.} When applying the basic model \eqref{eq_continous_x} to  continuous inputs $x=(x^1, \dots, x^D)\in \R^D$,  instead of a direct linear regression $f_d(x)=x^d$,   we recommend a  coordinate-wise ReLU-typed transformation: 
 	\begin{equation}\label{eq_relu}
 		\{f: f_{2d-1}(x)=(x^d- \mathrm{med}(x^d) )_{+},~~ f_{2d}(x)= (\mathrm{med}(x^d) -x^d )_{-}, ~d\leq D  \},
 	\end{equation}
 	where $\mathrm{med}(x^d)$ is the sample median of $x^d$. The pointwise model predictive performance typically relies on the  training density $P_{X}^{\mathrm{train}}(\tilde x)$: The more training data seen nearby, the better  predictions. %Different models have  different sensitivity on training density. 
 	The feature \eqref{eq_relu} is designed to be a crude approximation of log marginal input densities.
 \end{itemize}
 
 \subsubsection*{Choice of features and exploratory data analysis} 
 Choosing the division of $(x,z)$ in discrete inputs is now a more general problem on how to construct features $f_m(x)$, or a variable selection problem in a regression \eqref{eq_continous_x}.
 %Variable selection and  feature extraction are keys of the statistical workflow and can hardly be summarised by one paragraph.   
 In ordinary statistical modeling, we often start variable selection by  exploratory data analysis. Here we cannot directly associate  model weights $w_{ki}$ with observable quantities.  Nevertheless, we can use the paired pointwise log predictive density difference $\Delta_{ki}= (\log p_{k,-i} - \log p_{K,-i})$ as an exploratory approximation to the trend of $\alpha_{k}(x_i)$.  A scatter plot   of  $\Delta_{ki}$  against   $x$  may suggest which margin of $x$ is likely important. For example,  the dependence of   $\Delta_{ki}$ on whether $x_i$ is in the bulk or tail is an  evidence for our previous recommendation of the rectified features.
 
 As more variables $x$ are allowed to vary in the stacking model, model averaging is more prone to over-fitting.   Pointwise stacking typically has a large noise-to-signal ratio not only due to  model similarity, but also a high variance of pointwise model evaluation: the approximate leave-one-out cross validation possesses Monte Carlo errors; even if we run exact leave-one-out, or use an independent validation set in lieu of leave-one-out, we only observe one $y_i$ for one $x_i$ (if $x$ is continuous) such that $\log p_{k,-i} $ is at best an one-sample-estimate of $\E_{\tilde y| x_i} (\log p(\tilde y | x_i, M_k))$ with non-vanishing variance. 
 If  $f_m(\cdot)$ is flexible enough, then the sample optimum of  no-pooling stacking \eqref{eq_no-pooling}    always degenerates to pointwise model selection that pointwisely picks the model that ``best'' fits current realization of $y_i$: $w_{\arg\max_{k} p_{k,-i}} (x_i)=1$, which is purely over-fitting.  
 %Such pointwise uncertainty is particularly a concern if  $y$ is discrete:  Consider a Bernoulli model at a point with $p\approx 0.5$, and assuming model 1 always predicts 1 and model 2 always predicts 0, then the pointwise selection scheme above picks a random model by a coin flip.   
 %This uncertainty  highlights the necessity of   hierarchical regularization. 
 
 Even in companion with hierarchical priors,  we do not expect to include too many features on which stacking weights depend on. In our experiments,  an additive model with discrete variables and rectified continuous variables without interaction is often adequate. After standardizing all features such that Var$(f_m(x))=1$, we typically use a generic informative prior setting $\tau_\mu = \tau_\sigma=1$ in experiments. 
 With a moderate or large  number of features/cells, $M$, it is sensible to scale the  hyperprior  $\tau_\sigma = \mathcal{O}(\sqrt {1/M}),$  or adopt other feature-wise shrinkage priors such as  horseshoe for better regularization.
 
 \subsubsection*{Gaussian process prior}
 An alternative way to generalize both the discrete prior in Section \ref{sec_discrete_method} and the prior correlation \eqref{eq_corr} is  Gaussian process priors. To this end we need $K-1$ covariance kernels $\mathcal{K}_1, \dots, \mathcal{K}_{K-1}$, and place priors on the unconstrained weight $\alpha_k(x)$, viewed as an $\X \to \R$ function:  
 $\alpha_k(x) \sim \mathcal{GP}(\mu_k, \mathcal{K}_k(x))$. The discrete prior is a special case of a Gaussian process via a zero-one kernel $\mathcal{K}_k(x_i, x_j)= \sigma_{k} \mathbbm{1}(x_i=x_j).$  Due to the previously discussed measurement error and the preference on stronger regularization for continuous $x$,  we recommend simple exponentiated quadratic kernels $\mathcal{K}_k(x_i, x_j)=  a_k \exp(-  \left((x_i - x_j) / \rho_k \right)^2 )$ with an informative hyperprior that avoids too small or too big length-scale $\rho_k$, and too big $a_k$. We present an example in Section  \ref{sec_gp}.

 \subsection{Time series and  longitudinal  data} \label{sec:timeseries} 
 Hierarchical stacking can easily extend to  time series and longitudinal data. Consider a time series dataset where outcomes $y_i$ come sequentially in time $0 \leq t_i \leq T$.  
 The joint likelihood is not exchangeable,  but  still factorizable via $p(y_{1:n}|\theta)= \prod_{i=1}^n p(y_{i}| \theta,  y_{1:(i-1)}).$ Therefore, assuming some stationary condition,  we can approximate the  expected log predictive densities of the next-unit unseen outcome by historical average of one-unit-ahead log predictive densities, defined by 
 $$p_{k,-i} \coloneqq 
 %p(y_{i}|  x_{i}, \{y_{1:(i-1)}, x_{1:(i-1)}\}, M_k)=
 \int_{\Theta_k} p(y_{i}|  x_{i},y_{1:(i-1)}, x_{1:(i-1)}, \theta_k, M_k)p(\theta_k|y_{1:(n-1)}, x_{1:(n-1)}) d\theta_k.$$  
 In hierarchical stacking, we only need to replace the regular leave-one-out  predictive density \eqref{eq_pki} by this redefined $p_{k,-i}$,  and run hierarchical stacking
 \eqref{stacking_obj_h2} as usual.   Using importance sampling based approximation \citep{burkner2019approximate}, we also make efficient computation  without the need to fit each model $n$ times. 
 
 If we worry about time series being non-stationary, we can reweight  the likelihood  in \eqref{stacking_obj_h2}  by a non-decreasing sequence $\pi_i$:
 $n \sum_{i=1}^n \left( \pi_i \log  \left( \sum_{k=1}^K w_k(x_i) p_{k,-i}\right)\right) /\sum_{i=1}^n \pi_i$, so as to  emphasize more recent dates.
 For example,  $\pi_i = 1+ \gamma - (1-t_i/T)^2$, where a fixed parameter  $\gamma>0$ determines how much influence early 
 data has. 
 By appending $x\coloneqq(x, t)$,  the stacking weight can vary across the time variable,  too.

 In Section \ref{sec_election}, we present an election example with  longitudinal polling data  (40 weeks $\times$ 50 states).   For the $i$-th poll (already ordered by date), we encode state index into input $x_i= 1, \dots, 50$, all other poll-specific variables $z_i$, data $t_i$,  and poll outcome $y_i$. 
 We compute the one-week-ahead predictive density   $p_{k,-i} \coloneqq  \int\!p(y_{i}| x_{i},  z_{i},  
 \mathcal{D}_{-i}, M_k) p(\theta_k | \mathcal{D}_{-i}, M_k )d\theta_k $
 where  the dataset $\mathcal{D}_{-i} =  \{(y_{l}, x_{l}, z_{l}):  t_l \leq t_i-7\}$ contains polls from all states up to one week before date $t_i$.

 \section{Why model averaging works and why hierarchical stacking can work better}\label{sec_bound} 
 
 %proved that   under some regularity  conditions,  for any given set of weights $w_1 \dots w_K$, as sample size $n \to \infty$, the following asymptotic limit holds:
 %	$	\frac{1}{n}\sum_{i=1}^n \log \left( \sum_{k=1}^K  w_k \hat p_{k,  -i} (y_i) \right) -  \mathrm{E}_{\tilde y| y} \log \left(  \sum_{k=1}^K w_k p(\tilde y| y , M_k)  \right)          \xlongrightarrow{\text{$L_2$}}   0.	$

 The consistency of leave-one-out cross validation ensures that complete-pooling  stacking  \eqref{stacking_obj} is asymptotically no worse than model selection in  predictions  \citep{clarke2003comparing, le2017bayes}, hence justified by Bayesian decision theory.  
 The theorems we establish in Section \ref{sec_bound_2} go a step further, providing lower bounds on the utility gain of stacking and pointwise stacking. In short, model averaging is more pronounced when the model predictive performances are locally separable, but in the same situation, we can improve the linear mixture model by learning locally which model is better, so that the stacking is a step toward model improvement rather than an end to itself.   We illustrate with a theoretical example in Appendix~\ref{theory_example} and provide proofs in Appendix~\ref{sec_proof}.
 
 %\subsection{Rethinking  \texorpdfstring{$\mathcal{M}$}{M}-views:  The truer, the better?}
 \subsection{All models are wrong, but some are  somewhere useful}
 With an $\mathcal{M}$-closed view \citep{bernardo1994bayesian}, one of the candidate models is the true data generating process, whereas in the more realistic $\mathcal{M}$-open scenario, none of the candidate models is completely correct, hence  models are evaluated to the extent that they interpret the data.
 
 The expectation of a strictly proper scoring rule, such as the expected log predictive density (elpd), is maximized at the  correct data generating process. However, the extent to which a model is ``true'' is contingent on the input information we have collected. Consider an  input-outcome pair $(x,y)$ generated by
 $$ x \in  [0,1], ~y \in \{0, 1\}, ~~~ x\sim \mathrm{uniform}(0,1),     ~\Pr(y=1|x)= x.$$
 If the input $x$ is not observed or is omitted in the analysis,  then
 $M_1: y\sim \mathrm{Bernoulli} (0.5)$
 is the only correct model and is optimal among all probabilistic predictions of $y$ unconditioning on $x$. But this marginally {\em true} model is strictly worse than a {\em misspecified} conditional prediction,
 $M_2: \Pr(y=1|x)= \sqrt{x}, $
 since the expected log predictive densities are $\log(0.5)=-0.69$ and $-\frac{7}{12}=-0.58$ respectively after averaged over $x$ and $y$. The former model is true purely because it ignores some predictors. 
 
 This wronger-model-does-better example does not contradict the log score being strictly proper, as we are changing the decision space from measures on $y$ to conditional measures on $y|x$. 
 %\subsection{All models are wrong, yet some are  somewhere useful}
 But this example does underline two properties of model evaluation and averaging. First, we have little interest in a binary model check.  
 %The correct model $M_1$ passes the hypothesis-testing based model validation conditioning on $y$ but is  irrelevant to model fit or prediction accuracy. %In general,
 The hypothesis testing based model-being-true-or-false depends on what variables to condition on and is not necessarily related to model fit or prediction accuracy. 
 In a non-quantum scheme,  a really ``everywhere true" model that has exhausted all potentially unobserved inputs contains no aleatory uncertainty.
 %thereby making deterministic and exact predictions, and having elpd being infinite (continuous outcomes) or zero (discrete), which is never feasible  practice.   
 %But this example does reflect the tension between what  \citeauthor{breiman2001statistical} called  ``the two cultures of statistical modeling.'' 
 Second, the model fits typically vary across the input space.  In the Bernoulli example, despite its larger overall error, $M_1$  is more desired near $x\approx .5$, and is optimal at  $x=.5$.
 %Instead of using marginal likelihood or  scoring rules  as an one-number-summary of the overall  model fit, a finer grained model inspection checks the   $x-$conditional model performance, which we will elaborate in the present paper. 

 For theoretical interest, we define the conditional (on $\tilde x$) expected (on $\tilde y|\tilde x$) log predictive density in the $k$-th model,   
 $\mathrm{celpd}_{k} (\tilde x) \coloneqq \int_{\mathcal{Y}} p_t(\tilde y | \tilde x)  \log     p(\tilde y|\tilde x,  M_k)  d \tilde y.$
 If $\{\mathrm{celpd}_{k}\}_{k=1}^K$ are known,  we can divide the input space $\X$ into $K$ disjoint sets based on which model has the \emph{locally} best fit (When there is a tie,  the point is assigned the smallest index, and $\mathcal{I}$ stands for ``input''):
 \begin{equation}\label{eq_ik}
 	\mathcal{I}_k \coloneqq\{\tilde x \in \X:  \mathrm{celpd}_{k} (\tilde  x) >\mathrm{celpd}_{k^\prime} (\tilde  x), \forall k^{\prime}\neq k \}, ~ k=1, \dots, K.  
 \end{equation}
 In this Bernoulli example,  $\mathcal{I}_1= [0.25,0.67]$.
 
 \subsection{The gain from stacking,  and what can be gained more}\label{sec_bound_2} 
 In this subsection, we focus on the oracle expressiveness power of model selection and averaging,  and their input-dependent version.
 $\w^{\mathrm{stacking, cp}}$ refers to the complete-pooling stacking weight in the  population:
 \begin{align}\label{eq_stacking_population}
 	\w^{\mathrm{stacking, cp}}&\coloneqq \arg\max_{\w\in \mathcal{S_K}} \mathrm{elpd} (\w),\nonumber\\ 
 	\mathrm{elpd} (\w)&=\int_{\mathcal{X}  \times \mathcal{Y}} \log\left( \sum_{k=1}^K  w_k p(\tilde y| M_k, \tilde x)\right) p_t(\tilde y,  \tilde x) d\tilde y d\tilde x. 
 \end{align}
 
 Apart from the heuristic that  model averaging is likely to be more useful when candidate models are more ``dissimilar"  or ``distinct" \citep{breiman1996stacked, clarke2003comparing},  we are not aware of rigorous theories that characterize this ``diversity'' regarding the effectiveness of stacking.   
 It seems tempting  to use some divergence measure between posterior predictions from each model as a metric of how close these models are, but this is  irrelevant to the true data generating process. 
 %In an theoretical counterexample in Appendix \ref{theory_example},  stacking weights become more polarized as models get closer. 
 
 We define a more relevant metric on  how individual predictive distributions can be \emph{pointwisely} separated.
 The description of a forecast being good is probabilistic on both $\tilde x$ and $\tilde y$: an overall bad forecast may be lucky at an one-time  realization of  outcome $\tilde y$ and covariate $\tilde x$.  We consider the input-output product space $\mathcal{X}\times \mathcal{Y}$ and divide it into $K$ disjoints subsets ($\mathcal{J}$ stands for ``joint''):
 $$\mathcal{J}_k \coloneqq\{(\tilde x, \tilde y)\in \mathcal{X}\times \mathcal{Y}: p(\tilde y|M_k, \tilde x)> p(\tilde y|M_{k^\prime}, \tilde x),  \forall k^\prime\neq k\}, ~ k=1, \dots, K.$$  In this framework,  we call a family of predictive densities $\{p(\tilde y| M_k, \tilde x)\}_{k=1}^K$ to be locally separable with a constant pair  $L>0$ and  $0\leq \epsilon< 1$, if 
 \begin{equation}\label{eq_sep_y}
 	\sum_{k=1}^K \int_{(\tilde x, \tilde y)\in \mathcal{J}_k  } \mathbbm{1}\Big(  \log p(\tilde y|M_k, \tilde x)<   \log p(\tilde y|M_{k^{\prime}}, \tilde x)+ L,  \; \forall    k^{\prime} \neq k \Big)  p_t(\tilde y,  \tilde x) d\tilde y d\tilde x \leq \epsilon.
 \end{equation}
 
 Stacking is sometimes criticized for being a black box.  The next two theorems link stacking weight to a probabilistic explanation. Unlike Bayesian model averaging \citep{hoeting1999bayesian} that computes the probability of a model being ``true'',   stacking is more related to $\Pr(\mathcal{J}_k)$: the  probability of a model being the locally ``best'' fit, with respect to the true joint measure  $p_t(\tilde y,  \tilde x)$.   
 
 \begin{theorem}\label{thm_stacking_seperate_y} 
 	When the separation condition \eqref{eq_sep_y} holds, the complete pooling stacking weight  is  approximately  the probability of the model being the locally best fit:  
 	\begin{equation}\label{eq_opt_approx}
 		w^{\mathrm{stacking, cp}}_k  \approx  w^{\mathrm{approx}}_k \coloneqq \Pr(\mathcal{J}_k)= \int_{\mathcal{J}_k}    p_t(\tilde y,  \tilde x)d\tilde y d\tilde x,
 	\end{equation}
 	in the sense  that the objective function is nearly optimal: 
 	\begin{equation}\label{eq_opt_bound}
 		|\,\mathrm{elpd}(\w^{\mathrm{approx}}) - \mathrm{elpd}(\w^{\mathrm{stacking, cp}})\,| \leq   \mathcal{O}(\epsilon + \exp(-L)).
 	\end{equation} 
 \end{theorem}
 Further, a model is only ignored by stacking if its winning probability is low.
 \begin{theorem}\label{thm_winning_prob}
 	When the separation condition \eqref{eq_sep_y} holds,  and if the $k$-th model has zero weight in stacking, $w_k^\mathrm{stacking, cp}=0$,
 	then the probability of its winning region is bounded by:
 	\begin{equation}\label{eq_opt_bound_2}\Pr(\mathcal{J}_k) \leq  \left(1+ (\exp(L)-1) (1-\epsilon) +  \epsilon\right)^{-1}.
 	\end{equation}
 	The right-hand side can be further upper-bounded by  $\exp(-L)+  \epsilon$.
 \end{theorem}
 
 The  separation condition \eqref{eq_sep_y} trivially holds for $\epsilon=1$ and an arbitrary $L$,  or for $L = 0$ and an arbitrary $\epsilon$, though in those cases the bounds \eqref{eq_opt_bound} and \eqref{eq_opt_bound_2} are too loose. To be clear, we only use the closed form approximation \eqref{eq_opt_approx} for  theoretical assessment. %We use exact optimization and sampling in practical algorithms.
 
 The next theorem  bounds the  utility gain from shifting  model selection to stacking:
 \begin{theorem}\label{theo_lower_stacking}
 	Under the  separation condition \eqref{eq_sep_y}, let  $\rho=\sup_{k} \Pr(\mathcal{J}_k)$, and a deterministic function $g(L,K, \rho, \epsilon)= L(1-\rho)(1-\epsilon)- \log K$, then the utility gain of stacking is  lower-bounded by
 	$$ \mathrm{elpd}_{\mathrm{stacking,cp}} - \sup_k \mathrm {elpd}_k \geq   \max \left(g(L,K, \rho)+ \mathcal{O}(\exp(-L)+ \epsilon), 0\right).$$
 \end{theorem}
 
 Evaluating $\mathcal{J}_k$ requires access to $\tilde y | \tilde x$ and $\tilde x$. Though both terms are unknown, %$\tilde y | \tilde x$  is  the ultimate target in  regressions, whereas $p(\tilde x)$ is often replaced by empirical distribution of training $x$. On the other hand, 
 the roles of $\tilde x$ and $\tilde y$ are not symmetric: we could bespoke the model in preparation for a future prediction at a given $\tilde x$, but cannot be tailored for a realization of $\tilde y$. To be more tractable, we consider the case when the variation on $\tilde x$ predominates the uncertainty of model comparison, such that 
 $\mathcal{J}_k \approx \mathcal{I}_k \times \mathcal{Y}$, where $\mathcal{I}_k$ is defined in \eqref{eq_ik}. 
 More precisely, we define a strong local separation condition with a distance-probability pair ($L$, $\epsilon$):
 \begin{equation}\label{eq_sep}
 	\sum_{k=1}^K \int_{\tilde x \in \mathcal{I}_k} \int_{\mathcal{Y}}  \mathbbm{1}\Big(\log p(\tilde y|M_{k}, \tilde x )< \log p(\tilde y|M_{k^{\prime}}, \tilde x)+ L,  \; \forall    k^{\prime} \neq k  \Big) p_t(\tilde y,  \tilde x) d\tilde y   d \tilde x \leq \epsilon. 
 \end{equation}
 We define $\rho_{\mathcal{X}}=\sup_{k} \Pr(\mathcal{I}_k)$.  Under condition \eqref{eq_sep}, $\rho_{\mathcal{X}}$ and $\rho$ will be close. 
 %$\mathcal{J}_k \approx \mathcal{I}_k \times \mathcal{Y}$ such that $\rho_{\mathcal{X}} \approx \rho$.
 %and  the optimal stacking weights can be further approximated by   $w^{\mathrm{approx}}_k \coloneqq \Pr(\mathcal{I}_k)\coloneqq \int_{\mathcal{I}_k}   p(\tilde x)d\tilde x$ with the same objective bound \eqref{eq_opt_bound}.
 If we know  the input space division $\{\mathcal{I}_k\}$, we can select model $M_k$ for and only for $x\in \mathcal{I}_k$, which we call  pointwise selection. The predictive density is  \begin{equation}\label{eq_point_selection}p(\tilde y| \tilde x, \mathcal{I}, \mathrm{pointwise~selection} ) = \sum_{k=1}^K\mathbbm{1} (\tilde x\in \mathcal{I}_k) p(\tilde y|\tilde x,  M_k). \end{equation}
 
 As per Theorem \ref{theo_lower_stacking}, for  a given pair of $L$ and $\epsilon$, the smaller is $\rho$, the higher improvement ($K(1-\epsilon)(1-\rho)$)  can
 stacking achieve against model selection: the situation in which no model always predominates. Thus, the effectiveness of stacking  can indicate  heterogeneity of model fitting.   
 Next, we show that the heterogeneity of model fitting provides an additional utility gain if we shift from stacking to  pointwise selection:
 
 \begin{theorem}\label{theo_lower_selection}
 	Under the strong separation condition \eqref{eq_sep}, and if the divisions $\{\mathcal{I}_k\}$ are known exactly, then the  extra utility gain of pointwise selection has a lower bound,
 	$$ \mathrm{elpd}_{\mathrm{pointwise~ selection}}    - \mathrm{elpd}_{\mathrm{stacking,cp}}    \geq  -\log \rho_\mathcal{X} + \mathcal{O}(\exp(-L)+ \epsilon).$$
 \end{theorem}

 For a given input location $x_0 \in \X$, the pointwise no-pooling  optimum ${\w}(x_0)\in \mathcal{S}_{K}$ in the population is same as the complete-pooling solution  restricted to the slice  $\{x_0\} \times  \mathcal{Y}$. Hence, applying Theorem \ref{theo_lower_stacking} to each slice will bound the advantage of pointwise averaging  \eqref{eq_point_stacking} against  pointwise selection \eqref{eq_point_selection}.  
 \begin{figure}
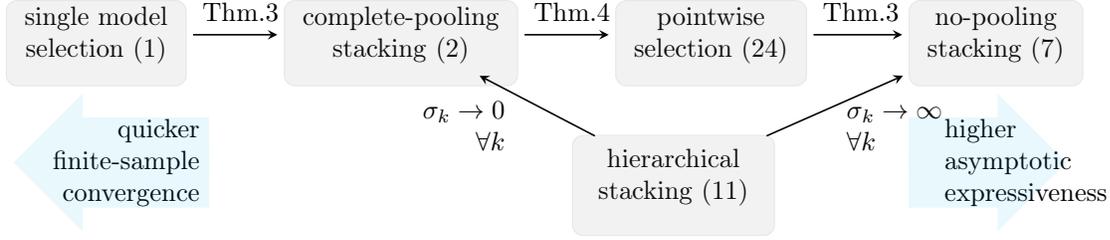

 	\centering
 	
 	\tikz{
 		\small
 		% nodes
 		\hspace{-0.9cm}
 		\node (p1) {\begin{tabular}{c} 	single model  \\ 	selection \eqref{eq_select}\\
 		\end{tabular} 	};
 		\node[ right =of p1 ,xshift= 0.12cm] (p2)  {\begin{tabular}{c}
 				complete-pooling  \\ 	stacking  \eqref{eq_linear_ave}\\ 	\end{tabular} 	};
 		\node[  right =of p2 ,xshift=  0.12cm] (p3)  {\begin{tabular}{c}
 				pointwise  \\ 	selection \eqref{eq_point_selection}\\ 
 		\end{tabular} 	};
 		\node[  right =of p3 ,xshift=  0.18cm] (p4)   {\begin{tabular}{c}
 				no-pooling  \\ 	stacking \eqref{eq_no-pooling}\\
 		\end{tabular} 	};
 		\node[  below =of p3 ,xshift=  -0.5 cm, yshift=0.2cm] (p5)   {\begin{tabular}{c} 	hierarchical  \\ stacking \eqref{stacking_obj_h2}\\
 		\end{tabular} 	};
 		\node  [below=of p1,yshift=0.6cm, xshift=0.4cm ] (m1){\begin{tabular}{r}  quicker \\	finite-sample  \\ convergence\\
 		\end{tabular} 	};
 		\node  [right=of m1,  xshift=0cm ] (m2){};
 		\node  [below=of p4,yshift=0.6cm, xshift=0.4cm ] (m5){\begin{tabular}{l} 	higher \\ asymptotic  \\ expressiveness\\
 		\end{tabular} 	};
 		\node  [left=of p5,yshift= 0.65cm, xshift=0.43cm] (m6){
 			\begin{tabular}{r} 
 				$\sigma_k\to 0$\\ $\forall k$	\end{tabular} 	};
 		\node  [right=of p5,yshift= 0.65cm,  xshift=-0.6cm] (m6){
 			\begin{tabular}{l} 
 				~~$ \sigma_k\to \infty$\\
 				$ ~~\forall k$	\end{tabular}	};
 		\node  [left=of m5,  xshift=0cm ] (m4){};
 		\plate [inner sep=-.1cm,  fill=gray,  opacity=0.1] {plate5}{(p1)} {} ; 
 		\plate [inner sep=-.1cm,  fill=gray,  opacity=0.1] {plate5}{(p2)} {} ; 
 		\plate [inner sep=-.1cm,  fill=gray,  opacity=0.1] {plate5}{(p3)} {} ; 
 		\plate [inner sep=-.1cm,  fill=gray,  opacity=0.1] {plate5}{(p4)} {} ; 
 		\plate [inner sep=0cm,  fill=gray,  opacity=0.1] {plate5}{(p5)} {} ; 
 		\edge  [-stealth,semithick] {p1} {p2};  
 		\node[  right =of p1 ,xshift= -1cm, yshift=0.3cm] {Thm.\ref{theo_lower_stacking}};
 		\edge[-stealth,semithick]  {p2} {p3}; 
 		\node[  right =of p2 ,xshift= -1cm, yshift=0.3cm] {Thm.\ref{theo_lower_selection}};
 		\node[  right =of p3 ,xshift= -1cm, yshift=0.3cm] {Thm.\ref{theo_lower_stacking} };
 		\edge[-stealth,semithick]  {p3} {p4}; 
 		\edge [-stealth,semithick]  {p5} {p2}; 
 		\edge [-stealth,semithick] {p5} {p4}; 
 		\begin{scope}[transparency group, opacity = 0.09]
 			\draw[
 			-triangle 90,
 			line width=3mm,
 			draw=cyan,
 			postaction={draw, line width=1.2cm,   shorten >=0.7cm, -}
 			] (10.8,-1.7) -- (12.9,-1.7);
 			\draw[
 			-triangle 90,
 			line width=3mm,
 			draw=cyan,
 			postaction={draw, line width=1.2cm,   shorten >=0.7cm, -}
 			] (1.5,-1.7) -- (-1.1,-1.7);
 		\end{scope}
 	}\caption{\em  Evolution of methods. First row from left to right: the methods have a higher degree of freedom to ensure a higher asymptotic predictive accuracy, the gain of which is bounded by the labeled theorems. Meanwhile, complex methods come with a slower convergence rate.  The hierarchical stacking is a generalization of all remaining methods by assigning various structured priors, and  adapts to the complexity-expressiveness tradeoff by  hierarchical modeling.}\label{fig:heuristic}
 \end{figure}
 
 The potential utility gain from Theorems~\ref{theo_lower_stacking} and \ref{theo_lower_selection}  is the motivation behind the input-varying model averaging.
 Despite this asymptotic expressiveness, the finite sample estimate remains challenging.  (a) We do not know $\mathcal{I}_k$ or  $\mathcal{J}_k$. We may use leave-one-out cross validation to estimate the overall model fit $\mathrm{elpd}_k$, but in the pointwise version, we want to assess conditional model performance. Further, the more data coming in, the more input locations need to assess.  (b) The asymptotic expressiveness comes with increasing complexity. The free parameters in single model selection, complete-pooling stacking, pointwise selection, and no-pooling stacking are a single model index, a length-$K$ simplex, a  vector of pointwise model selection index $\{1,2,\dots, K\}^\mathcal{X}$, and a matrix of pointwise weight  $(\mathcal{S}_K)^\mathcal{X}$. To handle this complexity-expressiveness tradeoff,  
 it is  natural  to apply the hierarchical shrinkage prior.

 %Although the approximation in Lemma \ref{thm_stacking_seperate} and the lower bound  in Theorem \ref{theo_lower_stacking} do not apply here for conditional models, this generalization is justified  by two motivations: 
 %\begin{itemize}
 %	\item    The  probabilistic weights $0\leq w_k(x) \leq 1$ can reflect the estimation uncertainty of the binary selection $\mathbbm{1} (x\in \mathcal{I}_k)$ from a finite-sample (or one-sample, in case of continuous $x$). 
 %	\item   Conditioning on one static $x_0$,  the lower bound derived from Theorem  \ref{theo_lower_stacking} becomes zero. Nonetheless,  stacking can still improve  predictions. Think about the extreme case where the pointwise true data generating process $p(y|x_0)$ is bimodal and the two modes are captured by two models:  here, stacking \eqref{eq_point_stacking} can recover the true data generating process while selection \eqref{eq_point_selection} cannot. \citet{yao2020stacking} discuss this perspective in details with further theory guarantees. 
 %  \end{itemize}

 \subsection{Immunity  to covariate shift}\label{sec_covariate_shift}
 So far we have adopted an IID view: the training and out-of-sample data are from the same distribution. Yet another appealing property of  hierarchical stacking is its immunity to \emph{covariate shift} \citep{shimodaira2000improving}, a ubiquitous problem in non-representative  sample survey, data-dependent collection, causal inference, and many other areas. 
 
 If the distribution of inputs $x$ in the training sample,   $p^{\mathrm{train}}_{X}(\cdot)$,  
 differs from these predictors' distribution in the population of interest, $p^{\mathrm{pop}}_{X}(\cdot)$ ($p^{\mathrm{pop}}_{X}$ is absolutely continuous with respect to $p^{\mathrm{train}}_{X}$), 
 and if $p(z|x)$ and $p(y|x, z)$ remain invariant, then we  do \emph{not} need to adjust weight estimate from \eqref{stacking_obj_h2}, because it has already aimed at pointwise fit. %though we also impose prior for regularization. 
 
 By contrast,  complete-pooling stacking targets the average risk. Under covariate shift, the sample mean of  leave-one-out score in the $k$-th model, $\frac{1}{n}\sum_{i=1}^{n}\log p(\tilde y | \tilde x, M_k)$, is no longer a consistent estimate of population elpd. To adjust, we can run importance sampling \citep{sugiyama2005input, sugiyama2007covariate, 
 	yao2018using} and reweight the $i$-th term in the objective \eqref{stacking_obj} proportional to the inverse probability ratio $p^{\mathrm{pop}}_{X}(x_i)/p^{\mathrm{train}}_{X}(x_i)$. Even  in the ideal situation when both $p^{\mathrm{pop}}_X$ and $p^{\mathrm{train}}_X$ are known, the importance weighted sum has in general larger or even infinite variance \citep{vehtari2015pareto},  thereby decreasing the effective sample size and  convergence rate in complete-pooling stacking (toward its optimum \eqref{eq_stacking_population}).  When  $p^{\mathrm{train}}_{X}$  is unknown, the covariate reweighting is more complex while hierarchical stacking  circumvents the need of explicit modeling of $p^{\mathrm{train}}_{X}$. 
 
 When we are interested only at one fixed input location $p^{\mathrm{pop}}_{X}(x)= \delta(x=x_0)$, hierarchical stacking  is ready for \emph{conditional} predictions, whereas no-pooling stacking and reweighted-complete-pooling stacking effectively discard all $x_i\neq x_0$ training data in their objectives, especially a drawback when $x_0$ is rarely observed in the sample.

 %That said, we can combine hierarchical stacking and poststratification to estimate the overall average risk. Suppose observations are from an online survey and we identify that the survey has selection bias on certain demographic variables $x$, such as age, gender and education, we can then fit multiple models and add hierarchical stacking to combine models. Finally we use poststratification weight $p^{\mathrm{pop}}_{x=j}$ to compute the average log predictive density in the population  $ \sum_{j=1}^J p^{\mathrm{pop}}_{X}(x=j) (\sum_{i: x_i=j}  \log \sum_{k} (\hat w(x_i) p(y_i|D_{-i}, M_k))/\sum_{i: x_i=j} 1)$. Although it is plausible to further reweigh the hierarchical-stacking likelihood in order to achieve double robustness, here as per the protocol of multilevel modeling \citep{gelman1997poststratification}, we recommend poststratification only  after model inference and  hierarchical  stacking. 
 
 % maybe we can say something about how to use propensity score as a feature  when dealing with stacking+ observational studies.
 
 \section{Related literature}\label{sec_related}
 Stacking \citep{wolpert1992stacked,  breiman1996stacked, leblanc1996combining}, or what we call \emph{complete-pooling stacking} in this paper  has long been a popular method to combine learning algorithms, and has been advocated for averaging Bayesian models \citep{clarke2003comparing,clyde2013bayesian,  le2017bayes, yao2018using}. Stacking is applied in various areas such as  recommendation systems, epidemiology \citep{bhatt2017improved}, 
 network modeling \citep{ghasemian2020stacking}, and post-processing in Monte Carlo computation \citep{tracey2016reducing, yao2020stacking}. Stacking  can be equipped with any scoring rules, while the present paper focuses on the logarithm score by default.

 Our  theory investigation in Section \ref{sec_bound_2} is inspired by the discussion of how to choose candidate models by  \citet{clarke2003comparing} and \citet{le2017bayes}. In $L^2$ loss  stacking, they recommended  ``independent''  models in terms of posterior point predictions $(\E (\tilde y| \tilde x, M_1), \dots, \E(\tilde y| \tilde x, M_K))$ being independent. When combining Bayesian predictive distributions, the correlations of  the posterior predictive mean is not enough to summarize the relation between predictive distributions \citep{pirs2019bayesian}, hence we consider the local separation condition instead.

 Allowing a heterogeneous stacking model weight that  changes with input $x$ is not a new idea.  Feature-weighted linear stacking \citep{sill2009feature} constructs  data-varying model weights of the $k$-th model  by $w_k(x)= \sum_{m=1}^{M}\alpha_{km}f_m(x)$, and $\alpha_{km}$ optimizes the $L^2$  loss of the point predictions of the  weighted model. This is similar to the likelihood term of our additive model specification in Section \ref{sec_additive}, except  we model the unconstrained weights. The direct least-squares optimization solution from feature-weighted linear stacking is what we label \emph{no-pooling stacking}.  
 %Examining pointwise results from cross validation relates to the work of \citet{kamary2014testing} and \citet{pirs2019bayesian} on the use of mixtures to compare multiple models where, unlike in the usual model averaging framework, different models can apply to different data points.   

 It is also not a new idea to  add regularization  and optimize the penalized loss function.  For $L^2$ loss stacking, \citet{breiman1996stacked} advocated non-negative constraints. In the context of combining Bayesian predictive densities, a simplex constraint is necessary.   
 \citet{reid2009regularized} investigated to add  $L^1$ or $L^2$ penalty, $-\lambda ||w||_1$ or $-\lambda ||w||_2$,  into complete-pooling stacking objective \eqref{stacking_obj}. \citet{yao2020stacking}  assigned a Dirichlet$(\lambda), \lambda>1$ prior distribution to  the  complete-pooling stacking  weight vector $w$ to  ensure strict concavity of the objective function.   
 \citet{sill2009feature} mentioned the use of $L^2$ penalization in feature-weighted linear stacking, which is equivalent to setting a  fixed prior for all free parameters $\alpha_{km}\sim \n (0, \tau), \forall k, m$,  whose solution path connects between uniform weighing and no-pooling stacking by tuning $\tau$. 
 All of these schemes are shown to reduce over-fitting with an appropriate amount of regularization,  while the tuning is computation intensive. In particular, each stacking run is built upon one layer of cross validation to compute the expected pointwise score in each model $p_{k,-i}$, and this extra tuning would require to fit each model $n(n-1)$ times for each tuning parameter value evaluation if both done in exact  leave-one-out way.
 \citet{fushiki2020selection} approximated this double cross validation for $L^2$ loss complete-pooling stacking with $L^2$ penalty on $\w$,  beyond which there was no general efficient approximation.  
 
 Hierarchical stacking treats $\{\mu_k\}$ and $\{\sigma_k\}$ as parameters and sample them from the joint density. 
 %We do need a hyperprior for them, but arguably the posterior is less sensitive to hyperprior than to priors.  Moreover,  hierarchical stacking is flexible in allowing $\sigma_k$ and $\mu_k$ to vary in model index $k$,  and easily extends to other structured priors. 
 Such hierarchy could be approximated by using $L^2$ penalized point estimate with a different tuning parameter in each model, and tune all  parameters ($\{\sigma_k\}_{k=1}^{K-1}$ for the basic model, or $\{\sigma_{mk}\}_{m=1, k=1}^{M, ~~K-1}$ for the product model). But then this intensive tuning is the same as finding the Type-II MAP of  hierarchical stacking  in an inefficient grid search (in contrast to gradient-based MCMC).  
 %Figure \ref{fig_ars_l2} compares hierarchical prior and the generic fixed-$L^2$ shrinkage in feature weighted stacking in an example. It is clear that  the whole trajectory path from fixed  $L^2$  prior with varying  $\tau$ is  inferior to hierarchical stacking. 
 
 %Our theory investigation in Section \ref{sec_bound} is inspired by the discussion of model list choice by \citet{clarke2003comparing}. 

 Another popular family of regularization in stacking enforces sparse weights \citep[e.g.,][]{zhang2011sparse, csen2013linear, yang2014minimax}, which include sparse and grouped sparse priors on the unconstrained weights, and sparse Dirichlet prior on simplex weights. The goal is that only a limited number of models are expressed.  From our discussion in Section \ref{sec_bound}, all models are somewhere useful, hence we are not aimed for model sparsity---The concavity of log scoring rules implicitly resists sparsity; The posterior mean of  hierarchical stacking  weights $w_{jk}$ is, in general, never sparse.  Nevertheless, when sparsity is of concern for memory saving or interpretability,  we can run  hierarchical stacking first and then apply projection predictive  variable selection \citep{piironen2017comparison} afterwards to the posterior draws from  the stacking model \eqref{stacking_obj_h2} and pick  a sparse (or cell-wise sparse) solution.

 In contrast to fitting individual models in parallel before model averaging, an alternative approach is to  fit all models jointly in a bigger mixture model. \citet{kamary2014testing} proposed  a  Bayesian hypothesis testing by fitting an encompassing model $p(y| \w,\theta) = \sum_{k=1}^K w_k p(y|\theta_k, M_k)$. The mixture model  requires to simultaneously fit  model parameters and model weights  $p(w_{1, \dots, K}, \theta_{1,\dots, K}|y)$, of which the computation burden is a concern when $K$ is big. \citet{yao2018using} illustrated that (complete-pooling) stacking is often more stable than full-mixture, especially when the sample size is small and some models are similar. Nevertheless, our formulation of hierarchical stacking agrees with \citet{kamary2014testing} in terms of sampling from   the  posterior marginal distribution of $p(\w|y)$ in a full-Bayesian  model.
 A jointly-inferred model $p(y|x,  \w(x),\theta) = \sum_{k=1}^K w_k(x) p(y|x, \theta_k, M_k)$  is related to the ``mixture of experts''  \citep{jacobs1991adaptive,waterhouse1996bayesian} and  ``hierarchical mixture of experts'' \citep{jordan1994hierarchical,  svenswn2003bayesian}, 
 where $w_k(\cdot)$ and $p(\cdot|x, M_k)$ are parameterized by neural networks and
 trained jointly in the bigger mixture model.
 Hierarchical stacking differs from mixture modeling in two aspects. First, its separate inference of individual models $p(\theta_1|y, M_1), \dots, p(\theta_K|y, M_k)$ and weights greatly reduces computation burden, making exact Bayes affordable. Second, the built-in  leave-one-out likelihood helps prevent overfitting. 
 Both the mixture model and stacking have limitations.  If the true data generating process is truly a mixture model, then fitting  individual  component  $p(\theta_k|y)$ separately is wrong and stacking cannot remedy it. On the other hand, stacking and hierarchical stacking are more suitable when each model has already been developed to fit the data on their own. 
 Put it in another way, rather than  to compete with a mixture-of-experts on combining weak learners,  hierarchical stacking is more recommended to combine a mixture-of-experts  with other sophisticated models. Lastly, our  full-Bayesian formulation  makes  hierarchical stacking directly applicable to complex priors and complex data structures, such as time series or panel data, while these extensions are not straightforward in the mixture of experts.

 \section{Examples}\label{sec_examles}
 We present three examples. The well-switching example demonstrates an automated hierarchical stacking implementation with both continuous and categorical inputs. The Gaussian process example highlights the benefit of hierarchical stacking when individual models are already highly expressive. The election forecast illustrates a real-world classification task with a complex data structure. We evaluate the proposed method on several metrics, including the mean log  predictive density on holdout data, conditional log  predictive densities, and the calibration error.
 
 \subsection{Well-switching in Bangladesh}\label{wells}
 We work with a dataset used by \cite{vehtari2017practical} to demonstrate cross validation. A survey with  a size of $n=3020$   was conducted on residents from a small area in Bangladesh that was affected by arsenic in drinking water. Households with elevated arsenic levels in their wells were asked whether or not they were interested in switching to a neighbor’s well, denoted by $y$. Well-switching behavior can be predicted by a set of household-level variables $x$, including the detected  arsenic concentration value in the well, the distance to  the closest known safe well, the education level of the head of household, and whether  any  household members are in community organizations.  The first two inputs are continuous and the remaining two are categorical variables. 
 
 We fit a series of logistic regressions, starting with an additive model including all covariates $x$ in model 1.  In model 2, we replace one  input---well arsenic level---by its  logarithm. In models 3 and 4, we add cubic spline basis functions with ten knots of  well arsenic level and distance,  respectively in input variables. In model 5 we replace the categorical education variable with a continuous measure of years of schooling.
 
 Using the additive model specification \eqref{eq_continous_x} and default prior \eqref{eq_continous_x_prior},  we model the unconstrained  weight  $\alpha_k(x)$ by a linear regression of 
 all categorical inputs and all rectified continuous inputs \eqref{eq_relu}.    In this example the categorical input has eight distinct levels  based on the product of education (four  levels) and community participation (binary).  
 
 %z_k[j]\sim \n(\mu_k, \sigma_k)$, and a standard normal prior on all $\beta_k,  \mu_k, \sigma_k$. 

 %\begin{equation}\label{eq_hybird}
 %\alpha_k(x) = \sum_j \beta_{jk} x_{\mathrm{con}, j} + z_k[x_{\mathrm{cat}}], ~k=1, \dots, 4;  \quad \alpha_5(x)=0.  
 %\end{equation}
 %As per previous discussion \eqref{eq_relu},  we convert all continuous  inputs $x_{\mathrm{con}}$ into two parts  $x_{\mathrm{con}, j}\coloneqq (x_{\mathrm{con},j}-  \mathrm{median}( x_{\mathrm{con}, j}))_{+}$  and $ (x_{\mathrm{con}, j}- \mathrm{median}( x_{\mathrm{con}, j}))_{-}$. 
 
 For comparison, we  consider three alternative approaches: (a) complete-pooling stacking (b) no-pooling stacking: the maximum likelihood estimate of \eqref{eq_continous_x}, and (c)  model selection that picks  model with the highest leave-one-out log predictive densities. 
 We split  data into a training set $(n_{\mathrm{train}}=2000)$ and an independent holdout test set. %We use the training data to fit individual models, obtain point wise leave-one-out log predictive densities, and train the hierarchical stacking model.

 \begin{figure}
 	\centering
 	\includegraphics[width=\textwidth]{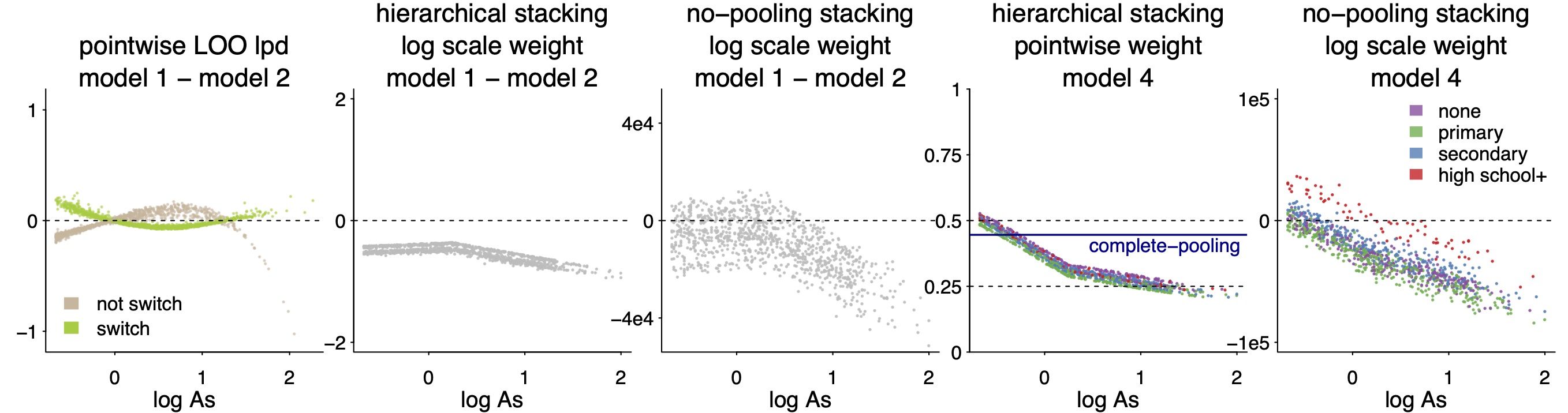}
 	\caption{\em   (1) Pointwise difference of leave-one-out log scores between models 1 and 2, plotted against log arsenic. Model 1 poorly fits  points with high arsenic. (2) Posterior mean of  pointwise unconstrained weight difference between models 1 and 2, $\alpha_{2}(x)-  \alpha_{1}(x)$ in  hierarchical stacking. (3) Pointwise  log weight difference between models 1 and 2 in no-pooling stacking. 
 		(4)  Posterior mean of $w_{4}(x)$, the  weight assigned to model 4,  in  hierarchical stacking, displayed  against  log arsenic  and education levels.  There are few samples with high school education and above, whose effect on model weights is pooled toward the shared mean.   The blue line is the complete-pooling stacking. (5) The unconstrained weight of model 4, $\alpha_4(x)$, in no-pooling stacking.  The ``high school'' effect stands out and the resulting  model weights $\w$ are nearly all zeroes and  ones. } \label{fig:well_pattern}
 \end{figure}   
 The leftmost panel in   Figure \ref{fig:well_pattern} displays the  pointwise difference of leave-one-out log scores for models 1 and 2 against log arsenic values in training data. Intuitively, model 1 fits poorly for data with high arsenic. In line with this evidence,  hierarchical stacking assigns model 1 an overall low weight, and especially low for the right end of the arsenic levels. The second panel shows the  pointwise posterior mean of  unconstrained weight difference between model 1 and 2, $\alpha_{2}(x)-  \alpha_{1}(x)$, against  the  arsenic values in training data.  The no-pooling stacking reveals a similar direction that model 1's weight should be lower with a higher arsenic value, but for lack of hierarchical prior regularization,   the fitted 
 $\alpha_{2}(x)-  \alpha_{1}(x)$ is orders of magnitude larger (the third panel). As a result, the realized pointwise   weights $\w$ are nearly  either zero or one.
 
 The rightmost two columns in  Figure \ref{fig:well_pattern} display the  fitted pointwise weights of model 4 against  log arsenic values and education level in test data.  Because only a small proportion  (7\%) of respondents had high school education and above, the no-pooling stacking weight for this category is largely determined by small sample variation. 
 Hierarchical stacking partially pools this ``high school'' effect toward the shared posterior mean of all educational levels, and the realized  hierarchical  stacking model  weights do not clearly depend on education levels.   
 
 \begin{figure}
 	\centering
 	\includegraphics[width=\textwidth]{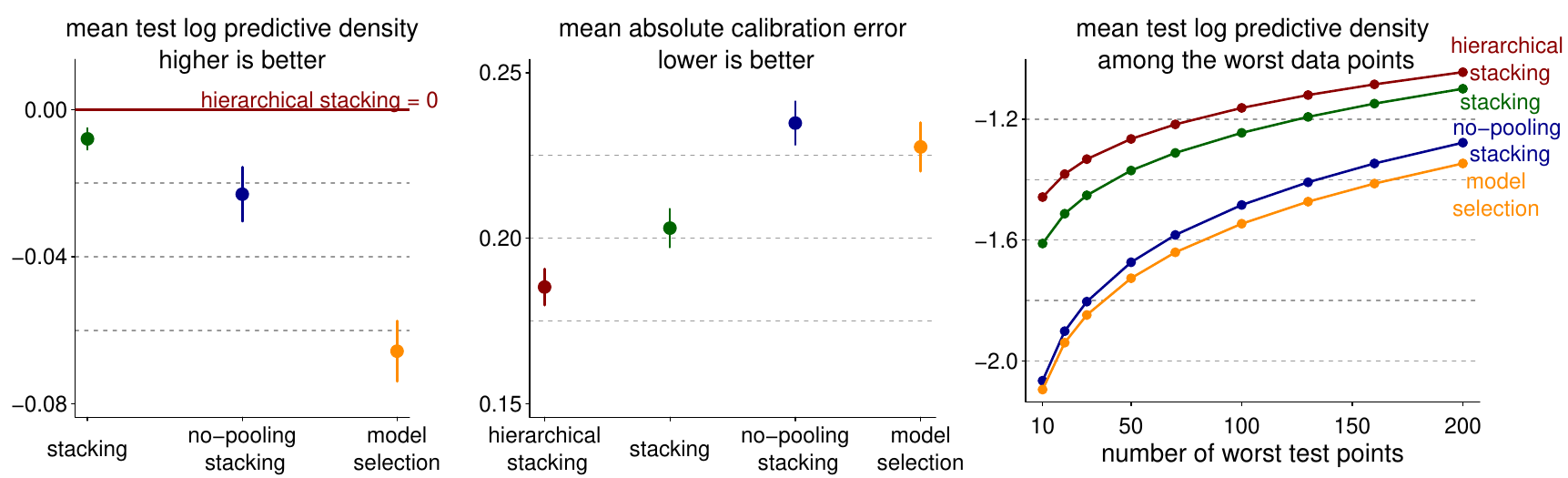}
 	\caption{\em  We evaluate hierarchical, complete-pooling and no-pooling stacking, and model selection on three metrics:   (a)  average log predictive densities  on test data, where we set the hierarchical stacking as benchmark 0,  (b)  calibration error:    discrepancy between the  predicted positive probability and  realized proportion of positives in test data, averaged over 20 equally spaced bins,  and (c)  average log predictive densities among the   $10\leq n_0 \leq 200$ worst test data points. We  repeat 50 random training-test splits with training size 2000 and test  size 1020.	} \label{fig:well}
 \end{figure}
 
 We evaluate  model fit on the following three   metrics. To reduce  randomness,   we evaluate all these metrics averaging over  50 random training-test splits.
 \begin{enumerate}[label=(\alph*), topsep=0pt,itemsep=0ex,partopsep=0ex,parsep=1ex]
 	\item The log predictive densities averaged over test data.  In the first panel of Figure \ref{fig:well}, we set hierarchical stacking as a baseline and all other methods attain lower predictive densities. 
 	\item The $L_1$ calibration error. We set 20 equally spaced bins between 0 and 1. For each bin and each learning algorithm, we collect test data points whose model-predicted positive probability falling in that bin, and  compute the absolute discrepancy between the realized proportion of positives in test data and the model-predicted probabilities. The middle panel in Figure \ref{fig:well} displays the resulted calibration error averaged over 20 bins. The proposed hierarchical stacking has the lowest error. No-pooling stacking has the highest calibration error despite its higher overall log predictive densities than model selection, suggesting prediction overconfidence.
 	\item We compute the average log predictive densities of four methods among the   $n_{\mathrm{worst}}$  most shocking test data points (the ones with  lowest predictive densities conditioning on a given method) for    $n_{\mathrm{worst}}$ varying from 10 to 200 and the total test data has size 1020.   As exhibited in the  last panel in Figure \ref{fig:well}, the proposed  hierarchical stacking consistently outperforms all other approaches for all $n_{\mathrm{worst}}$: a robust performance in the worst-case scenario. 
 \end{enumerate}

 \begin{figure}[!ht]
 	\centering
 	\includegraphics[width=\textwidth]{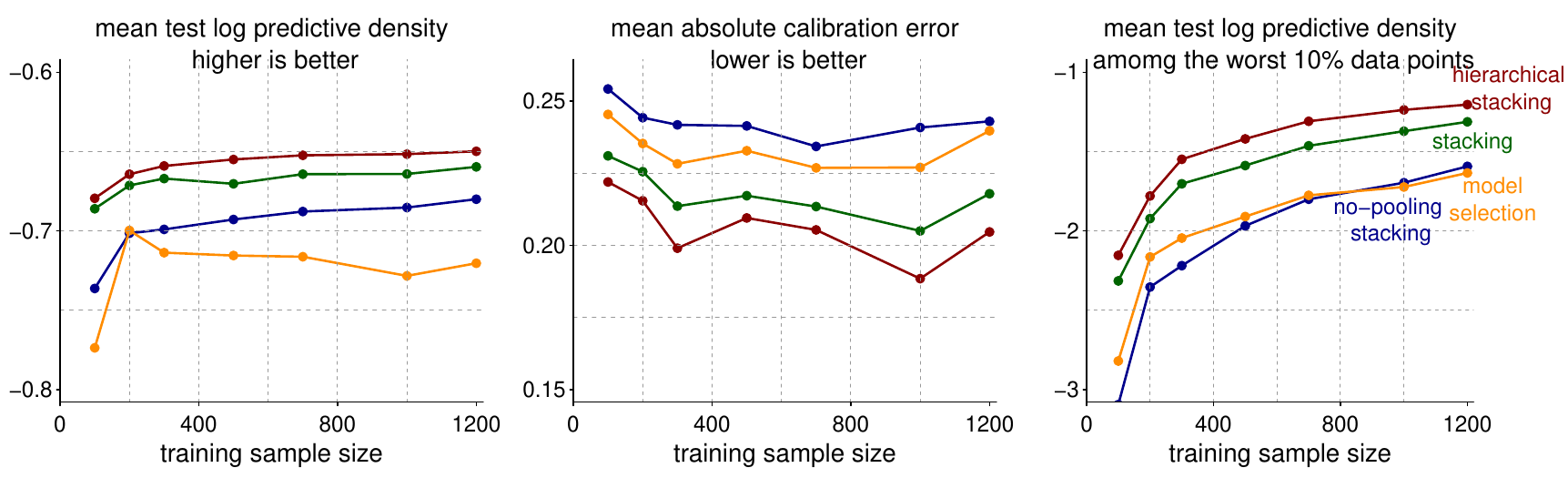}
 	\caption{\em Same comparisons as  Figure \ref{fig:well},  with training sample size varying from 100 to 1200. } \label{fig:well_vary_n}
 \end{figure}
 
 Figure \ref{fig:well_vary_n} presents the same comparisons of four methods while the training sample size $n_{\mathrm{train}}$ varies from 100 to 1200 (averaged over 50 random training-test  splits). In agreement with the heuristic in Figure \ref{fig:heuristic},  the most complex method---no-pooling stacking---performs especially poorly with a small sample size. By contrast, the simplest  method, model selection, reaches its peak elpd quickly with a moderate sample size  but cannot keep improving as training data size grows. The proposed hierarchical stacking performs the best in this setting  under all metrics.

 \subsection{Gaussian process regression weighted by another Gaussian process}\label{sec_gp}
 The local model averaging \eqref{eq_point_stacking_final} tangles a $x$-dependent weight $\w(x)$ and $x$-dependent individual prediction $p(y|x, M_k)$. If  the individual model $y|x,  M_k$ is already big enough to have exhausted ``all'' variability in input $x$, is there still a room for improvement by modeling local model weights $\w(x)$? The next  example  suggests a positive answer.

 Consider a regression problem with observations $\{y_i\}_{i=1}^n$ at one-dimensional  input locations $\{x_i\}_{i=1}^n$. To the data we fit a Gaussian process regression on the latent function $f$  with zero mean and squared exponential covariance, and independent noise $\epsilon$:
 \begin{equation}\label{eq_gp_regression}
 	y_i = f(x_i)+\epsilon _i, ~  \epsilon _i \sim  \mbox{normal}(0, \sigma),  ~ f(x) \sim  \mathcal{GP} \left( 0, a^2 \exp\left( -\frac{(x-x')^2}{\rho^2} \right) \right).
 \end{equation}
 \begin{figure}
 	\centering
 	\includegraphics[width=\textwidth]{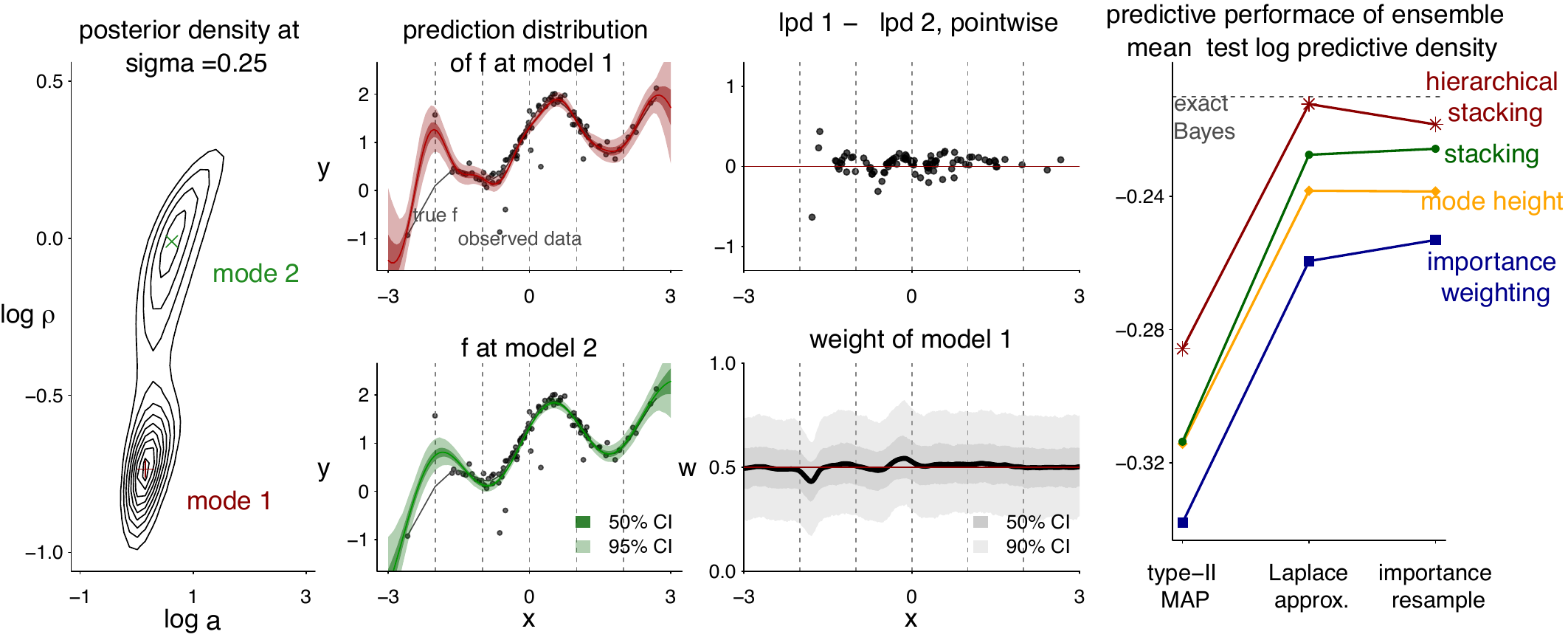}
 	\caption{\em From left to right, Column 1: posterior density at $\sigma=0.25$. At least two modes exist.  Column 2:  predictive distribution of $y$ from two modes. Column 3: the pointwise companion of log predictive density of the Laplace approximations at two modes, and the hierarchical stacking weight of mode 1. Column 4: the test data mean predictive densities of the weighted model, where individual components in the final model consists of either the  MAP, Laplace approximation, or importance sampling around the two modes, and the weighting methods include hierarchical stacking, complete-pooling stacking, mode heights and importance weighing. } \label{gp_point_hier_stacking}
 \end{figure}
 %  $K(x, x' | \rho, \alpha)  = \alpha^2 \exp\left( -\frac{(x-x')^2}{\rho^2}\right)$. $\alpha$ and $\rho$ are the signal-scale and length-scale.
 
 We adopt  training data  from \citet{neal1998regression}. They were generated such that the posterior distribution of  hyperparameters $ \theta=(a, \rho, \sigma)$ contains at least two  isolated modes (see the first panel in Figure \ref{gp_point_hier_stacking}). We  consider three mode-based  approximate inference of  $\theta | y$:   (a) Type-II MAP, where we   pick the local modes of  hyperparameters that maximizes the marginal density $\hat \theta = \arg\max p(\theta|y)$, and further draw local variables  $f| \hat \theta, y$,  (b) Laplace approximation of  $\theta|y$ around the mode, and (c) importance resampling  where we draw uniform samples near the mode and keep sample with probability  proportional to  $p\left(\theta | y\right)$. In the existence of two local modes $\hat \theta_1, \hat \theta_2$, we either obtain two MAPs or two  nearly-nonoverlapped draws, further leading to two predictive distributions.  \citet{yao2020stacking} suggests using complete-pooling stacking to combine two predictions, which shows advantages over other ad-hoc weighting strategies such as mode heights or importance weighting. 
 
 Visually,  mode 1 has smaller length scale,  more wiggling and attracted by training data.  Because of a better overall fit, it receives higher complete-stacking weights.   However, the wiggling tail makes its extrapolation less robust. We now run hierarchical stacking with $x$-dependent  weight $w_k(x)$ for mode $k=1, 2$ by placing another Gaussian process prior on unconstrained weight logit$(w_1(x))$ with squared exponential covariance, $$w_1(x) = \text{invlogit}\left(\alpha(x)\right),   ~~  \alpha (x)\sim \mathcal{GP}(0, K(x)).$$ 
 Despite using the same GP prior, this is not related to the training regression model \eqref{eq_gp_regression}. To evaluate how good the weighted ensemble is, we generate independent holdout test data $(\tilde x_i, \tilde y_i)$.  
 Both training and test inputs,  $x$ and  $\tilde x$,  are distributed from normal$(0, 1)$. 
 As presented in the rightmost panel in Figure \ref{gp_point_hier_stacking}, for all three approximate inferences,   hierarchical stacking  always has a higher  mean test log predictive density than complete pooling stacking and other weighting schemes.

 \begin{figure}
 	\centering
 	\includegraphics[width=0.55\linewidth]{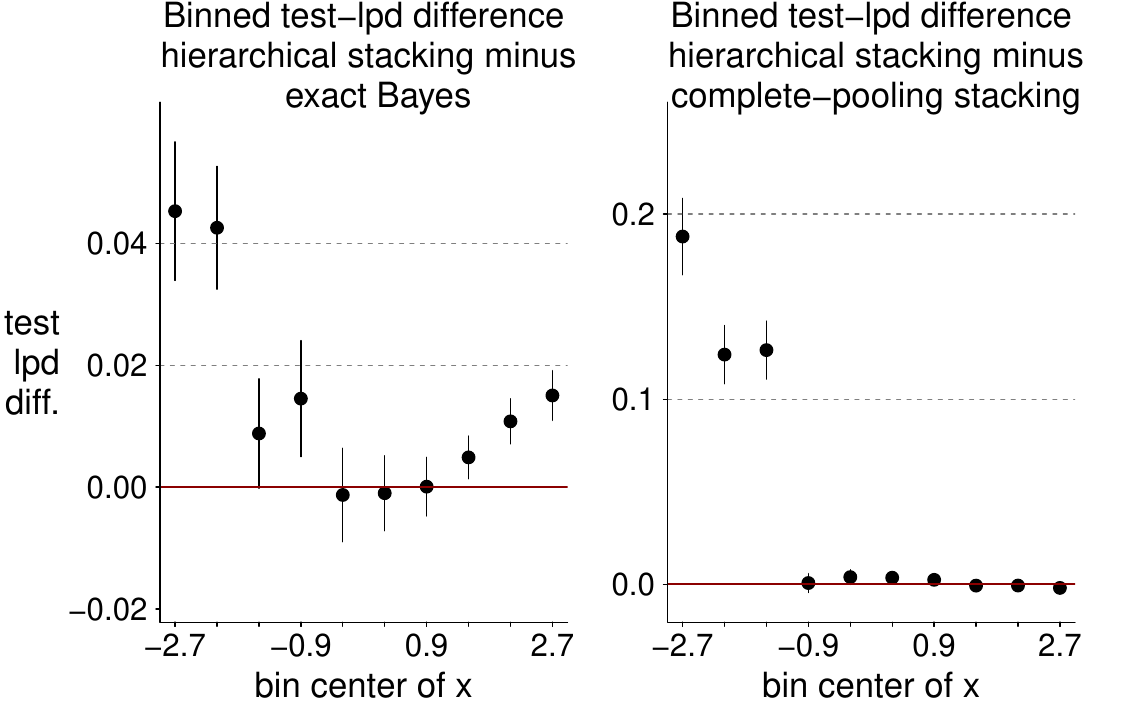}
 	\caption{\em  Compare hierarchical stacking  with  (left) stacking  of two Laplace approximations or (right) a long-chain exact Bayes from the true model. We compare the binned test log predictive densities over 10 equally spaced bins on $(-3,3)$.  A  positive value means  hierarchical stacking has a  better fit  than the counterpart.} \label{gp_covariate_shift}
 \end{figure} 
 
 In this dataset, exact MCMC is able to explore both posterior  modes in model  \eqref{eq_gp_regression} after a long enough sampling. Gaussian process regression equipped with  exact Bayesian inference can be regarded as the ``always true" model here.  Hierarchical stacking  achieves a similar average test data fit by combining two Laplace approximations. Furthermore,  hierarchical stacking has better predictive performance under covariate shift. To examine local model fit, we  generate another independent holdout test data, with results shown in Figure \ref{gp_covariate_shift}. This time the test inputs $\tilde x$ are from uniform$(-3,3)$. We divide the test data into 10 equally spaced bins and compute the mean test data log predictive density inside each bin. Compared with  exact inference,  hierarchical stacking has comparable performance in the bulk region of $x$, while it yields higher predictive densities in the tail,  suggesting a more reliable extrapolation.

 \subsection{U.S. presidential election forecast}\label{sec_election}
 We explore the use of hierarchical stacking on a practical example of forecasting  polls for the 2016 United States presidential election. Since the polling data are naturally divided into states, it  provides a suitable platform for hierarchical stacking in which model weights vary on states.   
 
 To create a pool of candidate models, we first  concisely describe
 the model of \citet{heidemanns2020updated},   an updated dynamic Bayesian forecasting model  \citep{linzer2013dynamic} for the  presidential election, and then follow up with different variations of it.
 Let $i$ be the index of an individual poll, $y_i$ the number of respondents that support the Democratic candidate, and $n_i$ the number of  respondents who support either the Democratic or the Republican candidate in the poll. Let $s[i]$  and $t[i]$ denote the state and time of poll $i$ respectively.  The model is expressed by
 \begin{align}
 	y_i &\sim \text{Binomial}(\theta_i, n_i), \nonumber \\
 	\theta_i &= 
 	\begin{cases}
 		\text{logit}^{-1}(\mu_{s[i],t[i]}^b + \alpha_i + \zeta_i^{\text{state}} + \xi_{s[i]}), & i \text{ is a state poll,} \\
 		\text{logit}^{-1}(\sum_{s=1}^S u_s\mu_{s,t[i]}^b + \alpha_i + \zeta_i^{\text{national}} + \sum_{s=1}^S u_s \xi_{s}), & i \text{ is a national poll,}
 	\end{cases} \label{eq:state_prob}
 \end{align}
 where superscripts denote parameter names, and subscripts their indexes. The term $\mu^b$ is the underlying support for the Democratic candidate, and $\alpha_i$, $\zeta$, and $\xi$ represent different bias terms. $\alpha_i$ is  further decomposed into
 \begin{equation}\label{eq:alpha}
 	\alpha_i = \mu_{p[i]}^c + \mu_{r[i]}^r + \mu_{m[i]}^m + z \epsilon_{t[i]},
 \end{equation}
 where $\mu^c$ is the house effect, $\mu^r$ polling population effect, $\mu^m$ polling mode effect, and $\epsilon$ an adjustment term for non-response bias. Furthermore, an autoregressive (AR(1)) prior is given to the $\mu^b$:
 $\mu_t^b | \mu_{t-1}^b \sim \text{MVN}(\mu_{t-1}^b, \Sigma^b),
 $
 where $\Sigma^b$ is the estimated state-covariance matrix and $\mu_T^b$ is the estimate from the fundamentals.
 
 Although we believe this model reasonably fits data, there is always room for improvement. 
 Our pool of candidates consists of eight models.
 \textbf{$M_1$}:  The  fundamentals-based model  of \citet{Abramowitz2008}.
 \textbf{$M_2$}: The  model of \citet{heidemanns2020updated}.
 \textbf{$M_3$}: $M_2$ without the fundamentals prior, $\mu_T^b = 0$.
 \textbf{$M_4$}: $M_2$ with an AR(2) structure, $\mu_t^b | \mu_{t-1}^b, \mu_{t-2}^b \sim \text{MVN}(0.5\mu_{t-1}^b \mu_{t-2}^b, \Sigma^b)$.
 \textbf{$M_5$}: simplify $M_2$ without polling population effect, polling mode effect, and the adjustment trend for non-response bias,  $\alpha_i = \mu_{p[i]}^c$.
 \textbf{$M_6$}: $M_2$ where we added  an extra regression term $\beta_{\mathrm{stock}} \mathrm{stock}_{t[i]}$ into model \eqref{eq:state_prob}   using the S\&P 500  index at the time of poll $i$.
 \textbf{$M_7$}: $M_2$ without the entire shared bias term, $\alpha_i = 0$.
 \textbf{$M_8$}: $M_2$ without hierarchical structure on states.

 We equip hierarchical stacking with either the  basic independent prior \eqref{eq_prior} or the state-correlated prior \eqref{eq_corr}. The prior correlation  $\Omega$ is estimated using a pool of state-level macro variables (election results in the past,  racial makeup,  educational attainment,  etc.), and has already been used in some of the individual models to partially pool state-level polling. We plug this pre-estimated prior correlation in the correlated stacking prior \eqref{eq_corr} and refer to it as ``hierarchical stacking with correlation" in later comparisons.

 Since the data are longitudinal, we evaluate different pooling approaches using a one-week-ahead forecast with an expanding window for each conducted poll. 
 We extract the fitted one-week-ahead predictions from each individual model, and train hierarchical stacking,  complete-pooling, and no-pooling stacking, and evaluate the combined models by computing their mean log predictive densities on the unseen data next week. To account for the non-stationarity discussed in Section \ref{sec:timeseries}, we only use the last four weeks prior to prediction day for training model averaging. In the end  we obtain a trajectory of  this back-testing performance of hierarchical stacking,  complete-pooling stacking, no-pooling stacking, and single model selection.

 \begin{figure}
 	\centering
 	\includegraphics[width=\textwidth]{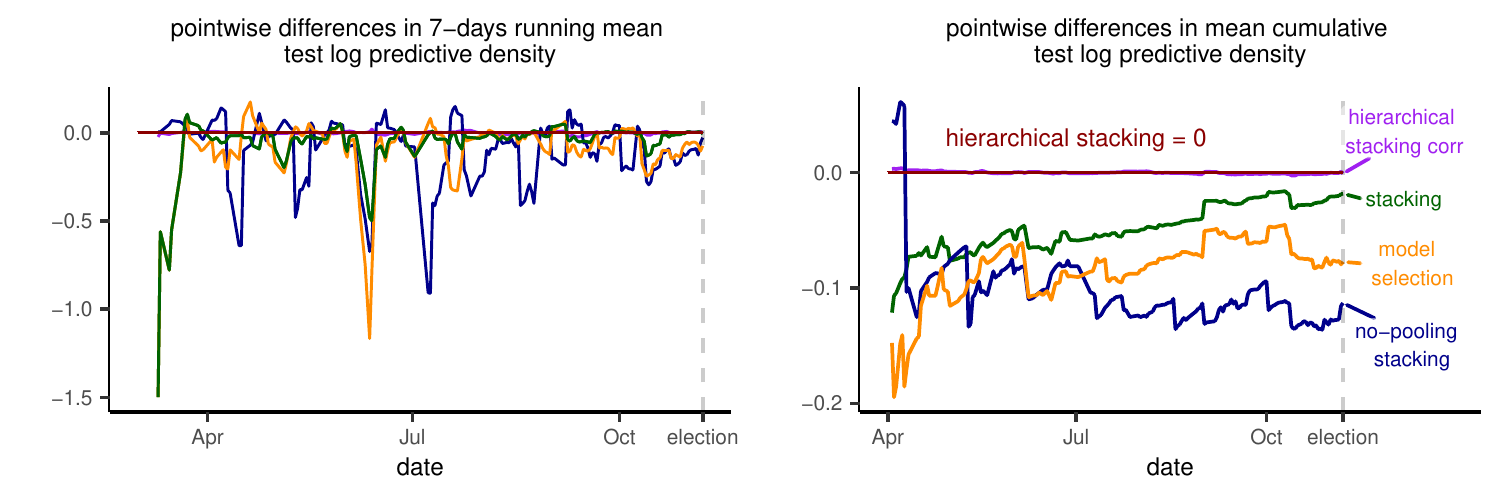}
 	\vspace{-2em}
 	\caption{\em Left: pointwise differences in 7-day running mean  log predictive  densities on one-week-ahead  test data, where we set the hierarchical stacking as benchmark 0. Right: pointwise differences in cumulative  average predictive log density by date. The advantage of hierarchical stacking is most noticeable toward the beginning, where there are fewer polls available.} \label{fig:polling_results}
 \end{figure}
 
 % Detailed evaluation in appendix
 
 \begin{figure}
 	\centering
 	\includegraphics[width=\textwidth]{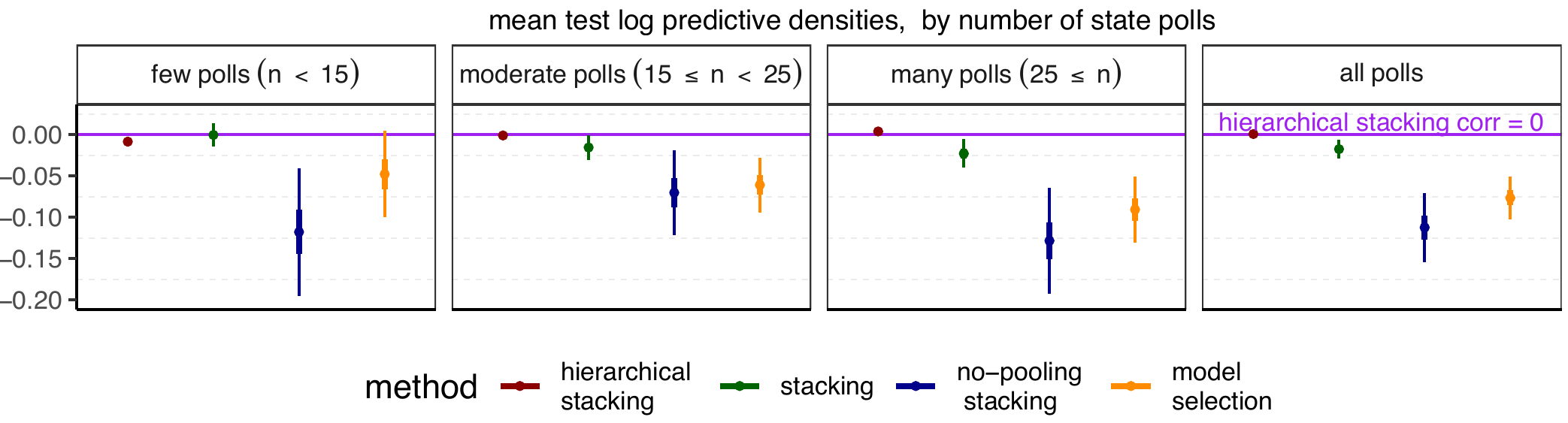} 
 	\caption{\em Mean test log predictive densities with 50\% and 95\% confidence intervals, among subsets of states with few, moderate, and many number of state polls, and among all states. 
 		Correlated hierarchical stacking is set as reference 0. It is better than independent hierarchical stacking when data are scarce. Complete-pooling stacking is close to hierarchical stacking in small states but worse in bigger states. 
 	} \label{fig:polling_mean_by_npolls}
 \end{figure}
 
 The left-hand side of Figure \ref{fig:polling_results} shows the seven-day running average of the one-week-ahead back-test log predictive density from models combined with various approaches. 
 The right-hand side of Figure \ref{fig:polling_results} shows the overall cumulative one-week-ahead back-test log predictive density.   We set the uncorrelated hierarchical stacking to be a constant zero for reference.
 Hierarchical stacking performs the best, followed by stacking, no-pooling stacking, and model selection respectively. 
 The advantage of hierarchical stacking is highest at the beginning and slowly decreases the closer we get to election day. As we move closer to the election, more polls become available, so the candidate models become better and also more similar since some models only differ in priors. As a result,  all combination methods eventually become more similar. No-pooling stacking has high variance and hence performs the worst out of all combination methods. Hierarchical stacking with correlated prior performs similarly to the independent approach, with a minor advantage at the beginning of the year, where the prior correlation stabilizes the state weights where the data are scarce, and later we see this advantage  more discernible in individual states.

 \begin{figure}
 	\centering
 	\includegraphics[width=\textwidth]{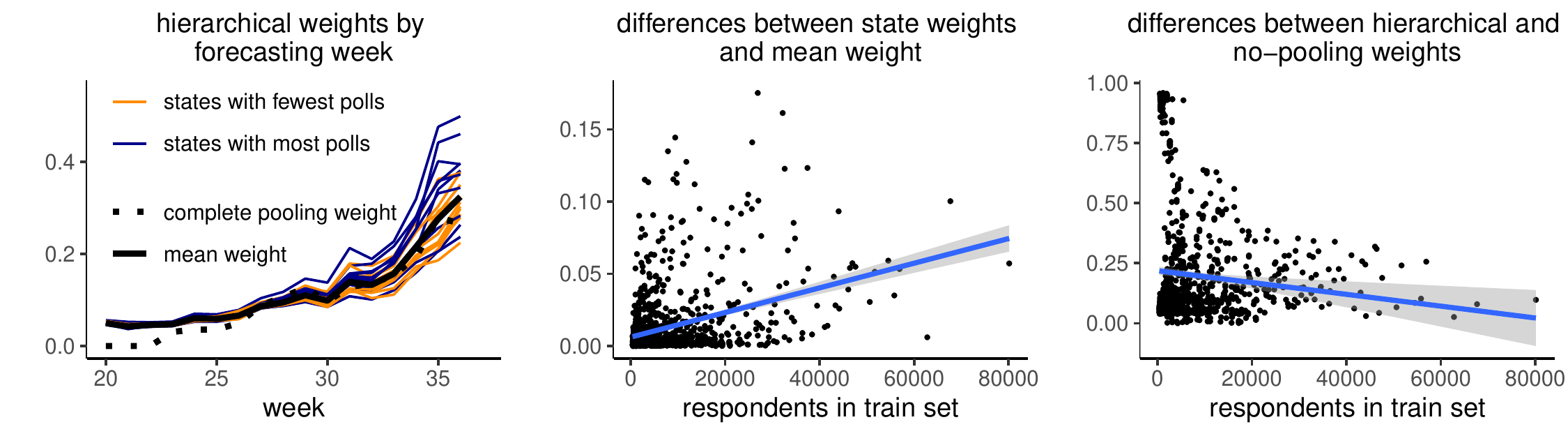}
 	\vspace{-1em}
 	\caption{\em Hierarchical stacking weights for $M_1$ in the polling example. Left: weights for $M_1$ of the 10 states with fewest polls and with most polls over time. Dotted line shows the complete-pooling stacking weight and the solid black line is  the nationwide mean weight.  States with fewer polls are shrunken more toward the mean. Middle: absolute differences between state-wise hierarchical stacking weights and the nationwide mean, against number of respondents. The blue line is the linear trend reference. States with smaller sample sizes are more pooled to the mean. Right: absolute differences between hierarchical stacking and no-pooling stacking weights, generally decreasing with bigger sample sizes.} \label{fig:polling_weights}
 \end{figure}
 
 To examine small area estimates,   we divided   states into three categories based on how many state polls were conducted. Figure \ref{fig:polling_mean_by_npolls} shows the overall mean pointwise differences in test log predictive densities divided by these categories, along with a fourth panel over all states. No-pooling stacking performs the worst in all panels. An explanation for that could be that we are using a four-week moving window to tackle non-stationarity, which might not contain enough data for the no-pooling method. The variance of the no-pooling is amended by the hierarchical approach, which performs on par with stacking with scarcer data and outperforms it otherwise. 
 %We can also observe the advantage of using prior correlations in hierarchical stacking in the first panel.  
 Figure \ref{fig:polling_cumulative_by_state}  in Appendix \ref{app:polling}   shows the state-level cumulative log predictive  density  by time.
 With a large number of state polls available, for example, close to election day in Florida and North Carolina, no-pooling stacking performs well.   In states with fewer polls,  no-pooling stacking is unstable.  Hierarchical stacking  alleviates this instability while retaining enough flexibility for a good performance when large data come in.

 Figure \ref{fig:polling_weights} illustrates how cell size affects  the  pooling effect. The first panel shows the  hierarchical stacking state-wise weights for the first candidate model $w_{1j}$ as a function of date. For either early-date forecasts or states with  few polls,    hierarchical
 stacking weights are more pooled toward the shared nationwide mean.  The middle and right panels compare the difference between state-wise hierarchical stacking weights and the nationwide mean, or with no-pooling weights, against the total number of respondents for each state and prediction date. The cells with more observed data are less pooled and closer to their no-pooling optimums, and vice versa.

 \section{Discussion}\label{sec_discuss}
 \subsection{Robustness in small areas}\label{dis_shrink}
 The input-varying model averaging is designed to improve both the overall averaged prediction $\E_{\tilde y, \tilde x} (\log p_{\mathrm{}}(\tilde y| \tilde x))$ and conditional prediction $\E_{\tilde y | \tilde x= x_0} (\log p_{\mathrm{}}(\tilde y| \tilde x))$, whereas these two  tasks are subject to a trade-off in complete-pooling stacking. 
 
 In addition,  the partial pooling prior \eqref{eq_prior}  borrows information from  other cells, which stabilizes model weights in small cells where there are not enough data for no-pooling stacking.   For a crude mean-field approximation,   the likelihood in the discrete model \eqref{stacking_obj_h2} is approximately $\prod_{j, k} \n(\alpha_{jk}^{\mathrm{mode}}, \lambda_{jk})
 $, where $\alpha^{\mathrm{mode}}= \arg\max_\alpha \sum_{i=1}^n \log  \left( \sum_{k=1}^K w_k(x_i) p_{k,-i}\right)$ is the unconstrained no-pooling stacking weight, and $- \lambda_{jk}^{-2} =  \frac{\partial^2}{\partial \alpha^2_{jk}}|_{\mathrm{mode}} \sum_{i=1}^n \log  \left( \sum_{k=1}^K w_k(x_i) p_{k,-i}\right)$ is the diagonal element of the Hessian.  Because $\alpha_{jk}$  appears in $n_j$ terms of the summation, $\lambda_{jk} = \mathcal{O}(n_j^{-{1}/{2}})$ for a given $k$. Combined with the prior  $\alpha_{jk} \sim  \n(\mu_k, \sigma_k)$, the conditional posterior mean of the  $k$-th model weight in the $j$-th cell is the usual precision-weighed average of the no-pooling optimum and the shared mean: $\alpha_{jk}^{\mathrm{post}} \coloneqq \E( \alpha_{jk} | \lambda_{jk}, \sigma_k, \mu_k, \mathcal{D}) \approx
 ({{\lambda^{-2}_{jk}} \alpha_{jk}^{\mathrm{mode}} +  {\sigma_{k}^{-2}} \mu_k  } )({\lambda^{-2}_{jk}}+ {\sigma_{k}^{-2} })^{-1}.$
 Hence for a given model $k$,   $|\alpha_{jk}^{\mathrm{mode}}- \alpha_{jk}^{\mathrm{post}} |  = \mathcal{O}({n_j^{-1}})$. Larger pooling usually occurs in  smaller cells. This pooling factor is in line with Figure \ref{fig:polling_weights} and general ideas in hierarchical modeling \citep{gelman2006bayesian}.
 Our full-Bayesian solution  also integrates out $\mu_k$ and $\sigma_k$, which further partially pools across models.
 
 The possibility of partial pooling across cells encourages  open-ended data gathering. In the election polling example, even if a pollster is only interested in the forecast of one state, they could gather polling  data from everywhere else, fit multiple models, evaluate models on each state, and use hierarchical stacking to construct model averaging, which is especially applicable when the  state of interest does not  have enough polls to conduct a meaningful  model evaluation individually. In this context  swing states naturally have more state polls, so that  the small-area estimation may not be crucial, but in general, we conjecture that the hierarchical techniques can be useful for model evaluation and averaging in a more general domain adaptation setting. Without going into extra details,   hierarchical models are as useful for making inferences from a subset of data (small-area estimation) as to generalize data to a new area (extrapolation). When the latter task is the focus,  hierarchical stacking only needs to redefine the leave-one-data-out predictive density \eqref{eq_pki} by leave-one-cell-out  
 $p_{k, -i} \coloneqq  \int_{\Theta_k}  p (y_i|\theta_k, x_i, z_i,  M_k) p\left(\theta_k| M_k, \{(x_{i^\prime}, z_{i^\prime}, y_{i^\prime}) : {x_{i^\prime}} \neq x_{i} \}  \right) d\theta_k.$

 \subsection{Using hierarchical stacking to understand local model fit}
 We use hierarchical stacking  not only as a tool for optimizing predictions but also as a way to understand problems with fitted models.
 The fact that hierarchical stacking is being used is already an implicit recognition that we have different models that perform better or worse in different subsets of data, and it can  valuable to explore the conditions under which different models are fitting poorly, reveal potential problems in the data or data processing, and point to directions for individual-model improvement. 
 
 \citet{vehtari2017practical} and \citet{gelman2020bayesian} suggested to examine the pointwise cross-validated log score $\log p_{k,-i}$
 %\int_{\Theta_k}  p (y_i|\theta_k, x_i, M_k) p\left(\theta_k| M_k, \{(x_{i^\prime}, y_{i^\prime}) : {i^{\prime}\neq i}\}  \right) d\theta_k$ 
 as a function of $x_i$, and  see if there is a pattern or explanation for why some observations are harder to fit than others. 
 For example, the first panel of Figure \ref{fig:well_pattern} seems to indicate that Model 1 is incapable of fitting the rightmost 10--15  non-switchers. However, $\log p_{k,-i}$ contains a non-vanishing variance since $y_i$ is a single realization from $p_t(y|x_i)$. Despite its merit in exploratory data analysis, it is hard to  tell from  the raw cross validation scores whether Model 1 is incapable of fitting high arsenic or is merely unlucky for these few points. The hierarchical stacking weight $\w(x)$ provides a smoothed summary of how each model fits locally in $x$ and comes with built-in Bayesian uncertainty estimation. For example, in Figure \ref{gp_point_hier_stacking}, $\log p_{1,-i}- \log p_{2,-i}$ has a slightly inflated right tail, but this small bump is smoothed by stacking,  and the local weight therein is close to $(0.5, 0.5)$.

 \subsection{Retrieving a formal likelihood from an optimization objective}\label{sec_stacking_post}
 The implication of hierarchical stacking \eqref{stacking_obj_h2} being a formal Bayesian model  is that we can evaluate  its posterior distribution as with a regular Bayesian model. For example, we can run (approximate) leave-one-out cross validation of the the stacking posterior $p(\w| \mathcal D_{-i}) \propto p(\w| \mathcal D) /  p(y_i| x_i, \w)= p(\w| \mathcal D) /   \left(\sum_{k=1}^K w_k(x_i)  p_{k,-i}\right)$. In practice, we only need to fit the stacking model \eqref{stacking_obj_h2} once, collect a size-$S$  MCMC sample of stacking parameters from the full posterior $p(\w| \mathcal D)$, denoted by  $\{(w_{k1}(x_i), \dots, w_{kS}(x_i))\}_{i,k}$, compute the PSIS-stabilized  importance ratio of each draw $r_{is} \approx \left(\sum_{k=1}^K w_{ks}(x_i)  p_{k,-i}\right)^{-1}$,  
 and then  compute the mean  leave-one-out cross validated log predictive density  to evaluate the overall out-of-sample fit of  the final stacked model:
 \begin{align} \label{eq_st_loo}
 	\mathrm{elpd}^\mathrm{loo}_\mathrm{stacking}&=
 	\sum_{i=1}^n \log \int_{\mathcal{S}_K} p(y_i|x_i, \w(x_i)) p(\w(x_i)| \mathcal D_{-i}) d (\w(x_i) ) \nonumber\\ &\approx \sum_{i=1}^n \log  \frac{\sum_{s=1}^S   \left(r_{is}  \sum_{k=1}^K w_{ks}(x_i)  p_{k,-i}\right)
 	} {\sum_{s=1}^S r_{is} }. 
 \end{align}
 As discussed in Section \ref{sec_related}, the same task of  out-of-sample prediction evaluation in an  optimization-based stacking requires double cross validation (refit the model $n(n-1)$ times if using leave-one-out), but now becomes almost computationally free by   post-processing posterior draws of stacking.
 
 The Bayesian justification above  applies to log-score stacking.
 In general, we cannot convert an arbitrary objective function into a log density---its exponential is not necessarily integrable, and, even if it is, the resulted density does not necessarily correspond to a relevant model. 
 Take  linear regression for example,  the ordinary least square estimate  $\arg\min_{\beta} \sum_{i=1}^n (y_i- x_i^T \beta)^2$ is identical to  the  maximum likelihood estimate of $\beta$ from  a probabilistic model $y_i|x_i, \beta, \sigma \sim \n( x_i^T \beta, \sigma)$ with flat priors. But the directly adapted ``log posterior density" from the  negative $L^2$ loss, $ \log p( \beta | y) =   -\sum_{i=1}^n (y_i- x_i^T \beta)^2+ C$,  differs from the Bayesian inference of the latter probabilistic model unless $\sigma \equiv 1$. The hierarchical stacking framework may still apply to other scoring rules, while we leave their Bayesian 
 calibration for future research.

 %\subsection{Stacking as an inference paradigm}
 %At a higher level, stacking is not only a model averaging tool, but also an inference paradigm, sitting in parallel to Bayes and empirical Bayes.  Consider one \emph{single} model with data $y$ and parameter $\theta$. There are three related paradigms to draw inference:
 %\begin{itemize}
 %   \item In a full-Bayes view, the posterior inference has to be $$\log p(\theta|y) = \sum_{i}  \log p(y_{i}| \theta) + \log p_{\mathrm{prior}} (\theta) + C.  $$
 %  \item  From the (optimization-based) stacking point of view, the goal of inference is better fit the data with respect to some predictive metric.  Equipped with leave-one-out log score, stacking  estimates the best parameter by 
 %    $$\hat  \theta=\arg\max_{\theta}  \sum_{i}  \log p(y_{i}| y_{-i}, \theta)+ \log p_{\mathrm{prior}} (\theta). $$
 %  \citet{yao2020stacking} discussed this stacking-as-single-model-training idea in the context of multimodal posterior, in which full-Bayes is often overconfident.
 %  This  approach is also  related to empirical Bayes.   
 %\item  This present paper bridges these two ends. 
 %The Bayesian version of stacking can be viewed as an single-model-inference defined by the modified log posterior density:  
 % $$\log p(\theta|y)  = \sum_{i}  \log p(y_{i}| y_{-i}, \theta) + \log p_{\mathrm{prior}} (\theta) + C. $$
 %\end{itemize}

 \subsection{Statistical workflow for  black box algorithms}\label{sec_two_culture} 
 Unlike our previous work \citep{yao2018using} that merely applied stacking to Bayesian models,  the present paper converts  optimization-based  stacking itself  into a formal Bayesian model, analogous to reformulating a least-squares estimate into a  normal-error regression.  \citet{breiman2001statistical} distinguished between two cultures in statistics: the \emph{generative modeling}  culture assumes that  data come from a given stochastic  model, whereas the \emph{algorithmic modeling} treats the data mechanism unknown and advocates black box learning  for the goal of predictive accuracy.  As a method that Breiman himself introduced \citep[along with][]{wolpert1992stacked},  stacking is arguably closer to the  algorithmic end of the spectrum, while our  hierarchical Bayesian  formulation pulls it toward the generative modeling end. 
 
 Such a full-Bayesian formulation  is appealing for two reasons.  First, the  generative modeling language facilitates flexible data inclusion during model averaging. For example, the election forecast model contains various outcomes on state polls and national polls from several pollsters,  and pollster-, state- and national-level fundamental predictors, and prior state-level correlations.  It is not  clear how methods like bagging or boosting can include all of them. Data do not have to conveniently arrive in independent $(x_i,y_i)$ pairs and compliantly await an algorithm to train upon. 
 Second,  instead of a static algorithm, hierarchical stacking is now part of a statistical workflow \citep{gelman2020bayesian}. It then enjoys all the flexibility of Bayesian model building, fitting, and checking---we can incorporate other Bayesian shrinkage priors as add-on components without reinventing them; 
 we can run a posterior predictive check or  approximate leave-one-out cross validation \eqref{eq_st_loo} to assess the out-of-sample performance of the final stacking model; 
 we may even further select, stack, or hierarchically stack  a sequence of hierarchical stacking model with various priors and parametric forms. 
 Looking ahead,  the success of this work encourages more use of generative Bayesian modeling to improve
 other black box prediction algorithms.

 \vspace{0.2cm}
\subsection*{Acknowledgements}
 	The authors would like to thank the National Science Foundation, Institute of Education Sciences, Office of Naval Research, National Institutes of Health, Sloan Foundation, and Schmidt Futures for partial financial support. Gregor Pirš is supported by the Slovenian Research Agency Young researcher grant.

 \bibliographystyle{apalike}
 \bibliography{hstack}

\begin{thebibliography}{}

\bibitem[Abramowitz, 2008]{Abramowitz2008}
Abramowitz, A.~I. (2008).
\newblock Forecasting the 2008 presidential election with the time-for-change
  model.
\newblock {\em Political Science and Politics}, 41:691--695.

\bibitem[Bernardo and Smith, 1994]{bernardo1994bayesian}
Bernardo, J.~M. and Smith, A.~F. (1994).
\newblock {\em {B}ayesian Theory}.
\newblock Wiley, Chichester.

\bibitem[Bhatt et~al., 2017]{bhatt2017improved}
Bhatt, S., Cameron, E., Flaxman, S.~R., Weiss, D.~J., Smith, D.~L., and
  Gething, P.~W. (2017).
\newblock Improved prediction accuracy for disease risk mapping using
  {G}aussian process stacked generalization.
\newblock {\em Journal of The Royal Society Interface}, 14.

\bibitem[Breiman, 1996]{breiman1996stacked}
Breiman, L. (1996).
\newblock Stacked regressions.
\newblock {\em Machine Learning}, 24:49--64.

\bibitem[Breiman, 2001]{breiman2001statistical}
Breiman, L. (2001).
\newblock Statistical modeling: the two cultures.
\newblock {\em Statistical Science}, 16:199--231.

\bibitem[B{\"u}rkner et~al., 2020]{burkner2019approximate}
B{\"u}rkner, P.-C., Gabry, J., and Vehtari, A. (2020).
\newblock Approximate leave-future-out cross-validation for {Bayesian} time
  series models.
\newblock {\em Journal of Statistical Computation and Simulation},
  90:2499--2523.

\bibitem[Clarke, 2003]{clarke2003comparing}
Clarke, B. (2003).
\newblock Comparing {B}ayes model averaging and stacking when model
  approximation error cannot be ignored.
\newblock {\em Journal of Machine Learning Research}, 4:683--712.

\bibitem[Clyde and Iversen, 2013]{clyde2013bayesian}
Clyde, M. and Iversen, E.~S. (2013).
\newblock {B}ayesian model averaging in the {M}-open framework.
\newblock In {\em {B}ayesian Theory and Applications}, pages 483--498. Oxford
  University Press.

\bibitem[Fushiki, 2020]{fushiki2020selection}
Fushiki, T. (2020).
\newblock On the selection of the regularization parameter in stacking.
\newblock {\em Neural Processing Letters}, pages 1--12.

\bibitem[Gelman and Pardoe, 2006]{gelman2006bayesian}
Gelman, A. and Pardoe, I. (2006).
\newblock Bayesian measures of explained variance and pooling in multilevel
  (hierarchical) models.
\newblock {\em Technometrics}, 48:241--251.

\bibitem[Gelman et~al., 2020]{gelman2020bayesian}
Gelman, A., Vehtari, A., Simpson, D., Margossian, C.~C., Carpenter, B., Yao,
  Y., Kennedy, L., Gabry, J., B{\"u}rkner, P.-C., and Modr{\'a}k, M. (2020).
\newblock Bayesian workflow.
\newblock {\em arXiv:2011.01808}.

\bibitem[Ghasemian et~al., 2020]{ghasemian2020stacking}
Ghasemian, A., Hosseinmardi, H., Galstyan, A., Airoldi, E.~M., and Clauset, A.
  (2020).
\newblock Stacking models for nearly optimal link prediction in complex
  networks.
\newblock {\em Proceedings of the National Academy of Sciences},
  117:23393--23400.

\bibitem[Heidemanns et~al., 2020]{heidemanns2020updated}
Heidemanns, M., Gelman, A., and Morris, G.~E. (2020).
\newblock An updated dynamic {B}ayesian forecasting model for the {US}
  presidential election.
\newblock {\em Harvard Data Science Review}, 2.

\bibitem[Heiner et~al., 2019]{heiner2019structured}
Heiner, M., Kottas, A., and Munch, S. (2019).
\newblock Structured priors for sparse probability vectors with application to
  model selection in {M}arkov chains.
\newblock {\em Statistics and Computing}, 29:1077--1093.

\bibitem[Hoeting et~al., 1999]{hoeting1999bayesian}
Hoeting, J.~A., Madigan, D., Raftery, A.~E., and Volinsky, C.~T. (1999).
\newblock Bayesian model averaging: a tutorial.
\newblock {\em Statistical Science}, pages 382--401.

\bibitem[Jacobs et~al., 1991]{jacobs1991adaptive}
Jacobs, R.~A., Jordan, M.~I., Nowlan, S.~J., and Hinton, G.~E. (1991).
\newblock Adaptive mixtures of local experts.
\newblock {\em Neural Computation}, 3:79--87.

\bibitem[Jordan and Jacobs, 1994]{jordan1994hierarchical}
Jordan, M.~I. and Jacobs, R.~A. (1994).
\newblock Hierarchical mixtures of experts and the {EM} algorithm.
\newblock {\em Neural Computation}, 6:181--214.

\bibitem[Kamary et~al., 2019]{kamary2014testing}
Kamary, K., Mengersen, K., Robert, C.~P., and Rousseau, J. (2019).
\newblock Testing hypotheses via a mixture estimation model.
\newblock {\em arXiv:1412.2044}.

\bibitem[Le and Clarke, 2017]{le2017bayes}
Le, T. and Clarke, B. (2017).
\newblock A {B}ayes interpretation of stacking for $\mathcal{M}$-complete and
  $\mathcal{M}$-open settings.
\newblock {\em Bayesian Analysis}, 12:807--829.

\bibitem[LeBlanc and Tibshirani, 1996]{leblanc1996combining}
LeBlanc, M. and Tibshirani, R. (1996).
\newblock Combining estimates in regression and classification.
\newblock {\em Journal of the American Statistical Association}, 91:1641--1650.

\bibitem[Linzer, 2013]{linzer2013dynamic}
Linzer, D.~A. (2013).
\newblock Dynamic {B}ayesian forecasting of presidential elections in the
  states.
\newblock {\em Journal of the American Statistical Association}, 108:124--134.

\bibitem[Meng and van Dyk, 1999]{meng1999seeking}
Meng, X.-L. and van Dyk, D.~A. (1999).
\newblock Seeking efficient data augmentation schemes via conditional and
  marginal augmentation.
\newblock {\em Biometrika}, 86:301--320.

\bibitem[Neal, 1998]{neal1998regression}
Neal, R.~M. (1998).
\newblock Regression and classification using {G}aussian process priors.
\newblock In Bernardo, J., Berger, J.~O., Dawid, A.~P., and Smith, A. F.~M.,
  editors, {\em Bayesian Statistics}, volume~6, pages 475--501. Oxford
  University Press.

\bibitem[Piironen and Vehtari, 2017]{piironen2017comparison}
Piironen, J. and Vehtari, A. (2017).
\newblock Comparison of {B}ayesian predictive methods for model selection.
\newblock {\em Statistics and Computing}, 27(3):711--735.

\bibitem[Pirš and Štrumbelj, 2019]{pirs2019bayesian}
Pirš, G. and Štrumbelj, E. (2019).
\newblock {B}ayesian combination of probabilistic classifiers using
  multivariate normal mixtures.
\newblock {\em Journal of Machine Learning Reserach}, 20:1--18.

\bibitem[Polson and Scott, 2011]{polson2011data}
Polson, N.~G. and Scott, S.~L. (2011).
\newblock Data augmentation for support vector machines.
\newblock {\em Bayesian Analysis}, 6:1--23.

\bibitem[Reid and Grudic, 2009]{reid2009regularized}
Reid, S. and Grudic, G. (2009).
\newblock Regularized linear models in stacked generalization.
\newblock In {\em International Workshop on Multiple Classifier Systems}, pages
  112--121.

\bibitem[{\c{S}}en and Erdogan, 2013]{csen2013linear}
{\c{S}}en, M.~U. and Erdogan, H. (2013).
\newblock Linear classifier combination and selection using group sparse
  regularization and hinge loss.
\newblock {\em Pattern Recognition Letters}, 34:265--274.

\bibitem[Shimodaira, 2000]{shimodaira2000improving}
Shimodaira, H. (2000).
\newblock Improving predictive inference under covariate shift by weighting the
  log-likelihood function.
\newblock {\em Journal of Statistical Planning and Inference}, 90:227--244.

\bibitem[Sill et~al., 2009]{sill2009feature}
Sill, J., Tak{\'a}cs, G., Mackey, L., and Lin, D. (2009).
\newblock Feature-weighted linear stacking.
\newblock {\em arXiv:0911.0460}.

\bibitem[{Stan Development Team}, 2020]{stan2020}
{Stan Development Team} (2020).
\newblock {\em Stan Modeling Language Users Guide and Reference Manual}.
\newblock Version 2.25.0, \url{http://mc-stan.org}.

\bibitem[Sugiyama et~al., 2007]{sugiyama2007covariate}
Sugiyama, M., Krauledat, M., and M\"{u}ller, K.-R. (2007).
\newblock Covariate shift adaptation by importance weighted cross validation.
\newblock {\em Journal of Machine Learning Research}, 8:985--1005.

\bibitem[Sugiyama and M{\"u}ller, 2005]{sugiyama2005input}
Sugiyama, M. and M{\"u}ller, K.-R. (2005).
\newblock Input-dependent estimation of generalization error under covariate
  shift.
\newblock {\em Statistics and Decisions}, 23:249--280.

\bibitem[Svens\"{e}n and Bishop, 2003]{svenswn2003bayesian}
Svens\"{e}n, M. and Bishop, C.~M. (2003).
\newblock Bayesian hierarchical mixtures of experts.
\newblock In {\em Uncertainty in Artificial Intelligence}.

\bibitem[Tracey and Wolpert, 2016]{tracey2016reducing}
Tracey, B.~D. and Wolpert, D.~H. (2016).
\newblock Reducing the error of {M}onte {C}arlo algorithms by learning control
  variates.
\newblock In {\em Conference on Neural Information Processing Systems}.

\bibitem[Vehtari et~al., 2020]{vehtari2018loo}
Vehtari, A., Gabry, J., Magnusson, M., Yao, Y., Bürkner, P.-C., Paananen, T.,
  and Gelman, A. (2020).
\newblock loo: {E}fficient leave-one-out cross-validation and {WAIC} for
  bayesian models.
\newblock {R} package version 2.4.1, \url{{https://mc-stan.org/loo/}}.

\bibitem[Vehtari et~al., 2017]{vehtari2017practical}
Vehtari, A., Gelman, A., and Gabry, J. (2017).
\newblock Practical {B}ayesian model evaluation using leave-one-out
  cross-validation and {WAIC}.
\newblock {\em Statistics and Computing}, 27:1413--1432.

\bibitem[Vehtari and Ojanen, 2012]{vehtari2012survey}
Vehtari, A. and Ojanen, J. (2012).
\newblock A survey of {B}ayesian predictive methods for model assessment,
  selection and comparison.
\newblock {\em Statistics Surveys}, 6:142--228.

\bibitem[Vehtari et~al., 2019]{vehtari2015pareto}
Vehtari, A., Simpson, D., Gelman, A., Yao, Y., and Gabry, J. (2019).
\newblock {P}areto smoothed importance sampling.
\newblock {\em arXiv:1507.02646}.

\bibitem[Waterhouse et~al., 1996]{waterhouse1996bayesian}
Waterhouse, S., MacKay, D., and Robinson, T. (1996).
\newblock Bayesian methods for mixtures of experts.
\newblock In {\em Advances in Neural Information Processing Systems}.

\bibitem[Wolpert, 1992]{wolpert1992stacked}
Wolpert, D.~H. (1992).
\newblock Stacked generalization.
\newblock {\em Neural Networks}, 5:241--259.

\bibitem[Yang and Dunson, 2014]{yang2014minimax}
Yang, Y. and Dunson, D.~B. (2014).
\newblock Minimax optimal {B}ayesian aggregation.
\newblock {\em arXiv:1403.1345}.

\bibitem[Yao, 2019]{yao2019bayesian}
Yao, Y. (2019).
\newblock {B}ayesian aggregation.
\newblock {\em arXiv:1912.11218}.

\bibitem[Yao et~al., 2020]{yao2020stacking}
Yao, Y., Vehtari, A., and Gelman, A. (2020).
\newblock Stacking for non-mixing {B}ayesian computations: The curse and
  blessing of multimodal posteriors.
\newblock {\em arXiv:2006.12335}.

\bibitem[Yao et~al., 2018]{yao2018using}
Yao, Y., Vehtari, A., Simpson, D., and Gelman, A. (2018).
\newblock Using stacking to average {B}ayesian predictive distributions (with
  discussion).
\newblock {\em Bayesian Analysis}, 13:917--1007.

\bibitem[Zhang and Zhou, 2011]{zhang2011sparse}
Zhang, L. and Zhou, W.-D. (2011).
\newblock Sparse ensembles using weighted combination methods based on linear
  programming.
\newblock {\em Pattern Recognition}, 44:97--106.

\end{thebibliography}

\vspace{1cm}
 \appendix
 \section*{Appendix} 
 \section{A theoretical example}\label{theory_example}
 Before theorem proofs, we first consider a toy example. It can be solved with a closed form solution and illustrates  how  Theorems 1--4 apply. 
 
 As shown in Figure \ref{fig:toy},  the true data generating process (DG) of the outcome is  $y \sim \mathrm{uniform}(-3,1)$, and there are two given (pre-trained) models with  spike-and-slab predictive distributions 
 \begin{align*}
 	&M_1: y \sim .99 ~\mathrm{uniform}(-4,0) + .01 ~\mathrm{uniform}(0,2),\\
 	&M_2: y \sim .99 ~\mathrm{uniform}(0,2) + .01 ~\mathrm{uniform}(-4,0),   
 \end{align*}
 which yield piece-wise constant  predictive densities 
 \begin{align*}
 	&p_1(y)=0.99/4 \mathbbm{1}(y\in[-4,0])+ 0.01/2  \mathbbm{1}(y\in[0,2]),\\
 	&p_2(y)=0.99/2 \mathbbm{1}(y\in[0,2]) + 0.01 / 4 \mathbbm{1}(y\in[-4,0]).   
 \end{align*}
 Using our notation in Section \ref{sec_bound_2}, the region in which $M_1$  predominates is $\mathcal{J}_1=[-4,0]$, and $M_2$ outperforms on  $\mathcal{J}_2=(0,2]$ (the conventions send the tie $\{0\}$ to $M_1$). We count their masses with respect to the true DG: $\Pr(\mathcal{J}_1)=3/4$ and  $\Pr(\mathcal{J}_2)=1/4$. %Now consider the weighted density $   (w_1 0.99/4 +  w_2 0.01/4)  \mathbbm{1}(y\in[-4,0])+  (w_2  0.99/2+ w_1 0.01/2) \mathbbm{1}(y\in[0,2])$, 
 
 Complete-pooling stacking solves 
 \begin{align*}
 	\max_{\w\in \mathcal{S}_2} \int_{-3}^1 1/4  \log \Bigl(  &(w_1 0.99/4 +  w_2 0.01/4)  \mathbbm{1}(y\in[-4,0]) +\\
 	& (w_2  0.99/2+ w_1 0.01/2) \mathbbm{1}(y\in[0,2]) \Bigr) dy.
 \end{align*}
 The exact optimal weight is $w_1=0.755$, close to the mass $\Pr(\mathcal{J}_1)=0.75$ and is irrelevant to the height of each regions. For instance, if the right bump in $M_2$ shrinks to the interval $(0,1)$ (i.e., $y \sim 0.99 ~\mathrm{uniform}(0,1) + 0.01 ~\mathrm{uniform}(-4,0)$), then the winning margin therein is twice as big,  while the winning probability as well as the stacking weight remains nearly unchanged.

 \begin{figure}[!b]
 	\centering
 	\includegraphics[width=\textwidth]{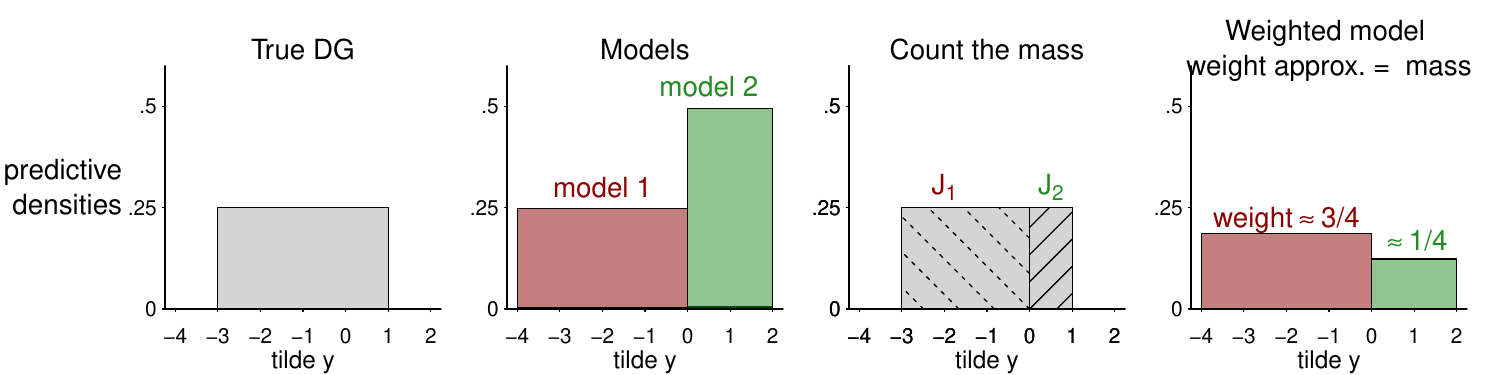}
 	\caption{\em The true data is generated from $\mbox{uniform}(-3,1)$ and there are two models with spikes and slabs on intervals $(-4,0)$ and $(0,2)$ respectively.  $\mathcal{J}_1$ and $\mathcal{J}_2$ in Theorem \ref{thm_stacking_seperate_y} are $[-4,0]$ and $(0,2]$, with DG probabilities 3/4 and 1/4. The stacking weights are approximately  these two probabilities, and irrelevant to how high the winning margins are.} \label{fig:toy}
 \end{figure}
 At the pointwise level, stacking behaves as  a   plurality voting system:  as long a model ``wins" a sub-region (subject to  a prefixed threshold  $L$ in condition \eqref{eq_sep_y}),  \emph{the winner take all} and its winning margin no longer matters.   
 
 By contrast, likelihood-based model averaging techniques such as Bayesian model averaging \citep[BMA,][]{hoeting1999bayesian} and pseudo-Bayesian model averaging \citep{yao2018using} are analogies of  \emph{proportional  representation}: every count of the winning margin matters. For illustration, we vary the slab probability $\delta$ in Model 1 and 2:   \begin{align*}
 	M_1 \mid \delta:~~ y \sim (1-\delta) \times \mathrm{uniform}(-4,0) + \delta \times\mathrm{uniform}(0,2),\\
 	M_2\mid \delta: ~~y \sim (1-\delta) \times\mathrm{uniform}(0,2) + \delta \times\mathrm{uniform}(-4,0).
 \end{align*}
 The left column  in Figure \ref{fig:winner} visualizes the predictive densities from these two models at $\delta=0.2$, $0.33$, and $0.45$.

 \begin{figure}
 	\centering
 	\includegraphics[width=0.52 \textwidth]{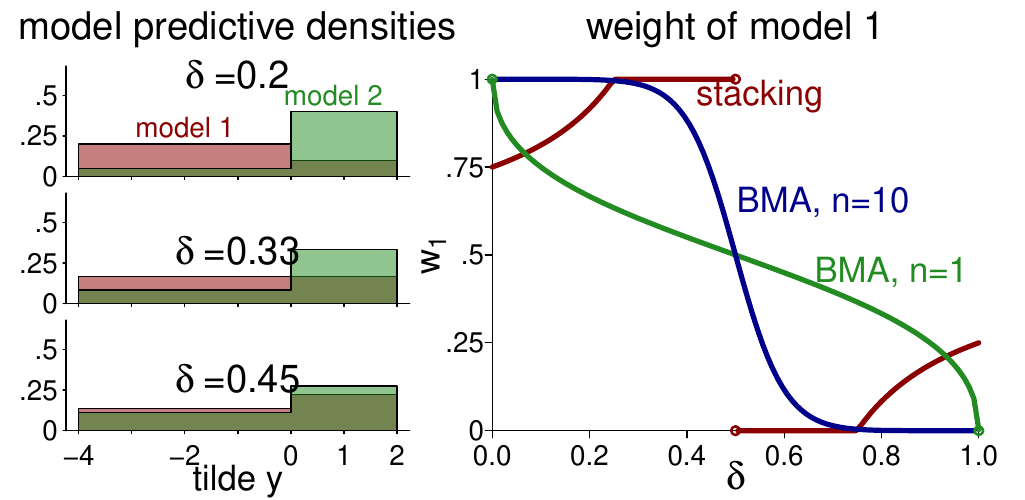}
 	\caption{\em  Left: Pointwise predictive density  $p(\tilde y|M_1 ~\mathrm{or}~ M_2)$ when the slab probability $\delta$ is chosen 0.2, 1/3 and 0.45. Right: Weight of model 1 in complete-pooling stacking (not defined at $\delta= 0.5$) and pseudo-BMA (sample size $n$=1 or 10, not defined at $\delta=0$ or 1) as a function of the slab probability $\delta$. They evolve in the opposite direction. Besides, stacking weights are more polarized when models are more similar.} \label{fig:winner}
 \end{figure}

 \begin{figure}
 	\centering
 	\includegraphics[width=\textwidth]{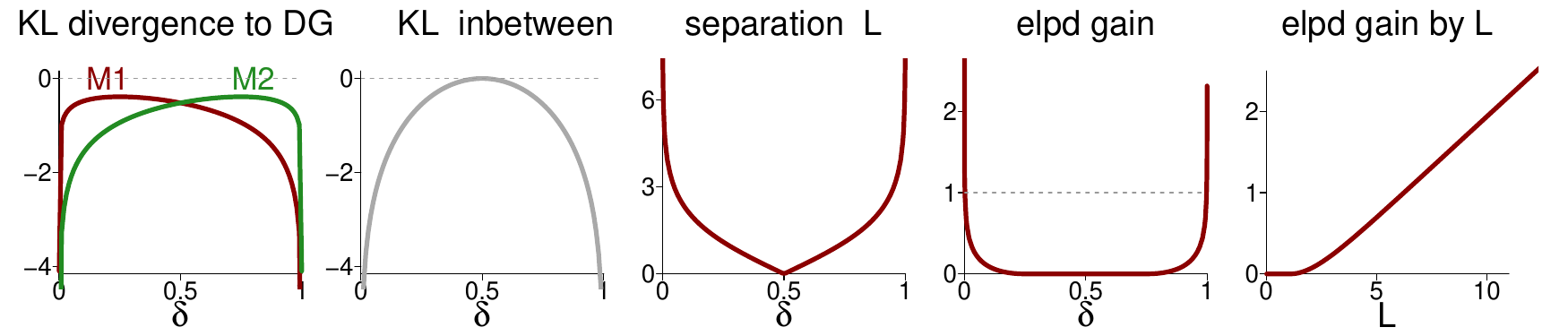}
 	\caption{\em From left: (1) KL divergence between model 1 or model 2 and data generating process. (2) KL divergence between model 1 and model 2.
 		(3) Separation constant $L$. (4) Stacking elpd gain compared with the best individual  model.
 		(5) Stacking elpd gain as a function of $L$. } \label{fig:winner2}
 \end{figure}

 When the slab probability $\delta$ increases from 0 to 0.5, these two models are closer and closer to each other,  measured by a smaller KL($M_1, M_2$). The $(0.5,1)$ counterpart is similar, though not exactly symmetric. We compute stacking weight and the expected pseudo-BMA  weight with sample size $n$: $w_1^{\mathrm{BMA}}(n, \delta)=   \left(1+\exp\left( n\E_{y|\delta}  \log p_2(y)- n\E_{y|\delta}  \log p_1(y)\right)\right)^{-1}$.

 Interestingly,  pseudo-BMA  weight $w_1^{\mathrm{BMA}}(n, \delta)$ is strictly decreasing as a function of $\delta \in  (0,1)$. This is because when $\delta\to 0^+$, log predictive density of model 2 in the left part $\log(\delta /4) \to \infty$ can be arbitrarily small, and the influence of this bad region dominates the overall performance of model 2. By contrast, stacking weight is monotonic non-increasing on $(0,0.5)$ (strictly decreasing on $(0,1/3)$, and remains flat afterwards)---the opposite direction of BMA. 
 Stacking simply recognizes model 1  winning the $[-3,0]$ interval and does not haggle over how much it wins. 
 
 In addition, when $\delta=1/3$, $M_1$ becomes a uniform density on $[-4,2]$. When $\delta \in (1/3,1/2)$, model 2 is not only strictly worse than model 1 but also provides no extra information for model averaging. Hence stacking assigns it weight zero.

 The first two panels in Figure \ref{fig:winner2} show the KL divergence from model 1 or from model 2 to the data generating process and the KL divergence between model 1 and model 2. The third panel is the largest separation constant $L$ for which the  separation condition \eqref{eq_sep_y} holds. The last two panels show the stacking epld gain (compared with the best individual model) as a function of $\delta$ and $L$. This constructive example reflects the worst case 
 for it matches the  theoretical lower bound
 $g^*(L,K, \rho, \epsilon)=\log(\rho)+ (1-\rho)(\log(1-\rho)- \log(K-1))$ (here $L=L, K=2, \rho=1/4, \epsilon=0$) in Theorem~\ref{theo_lower_stacking}.

 When  $\delta \in [1/3,1/2)$, Model 2 still wins 
 on the interval $\mathcal{J}_2= (0,2]$ with the separation constant  $\epsilon =0$ and  $L\leq \log 2$ (the winning margin is maximized at $\delta=1/3$).  Nevertheless, a zero stacking weight and a non-zero winning area do not contradict  Theorem  \ref{thm_stacking_seperate_y}. Indeed, Theorem \ref{thm_winning_prob}  precisely bounds the mass of the winning region when stacking weight is zero.  We provide self-contained theorem proofs  in the next section.

 Loosely speaking, BMA computes the probability of a model being \emph{true} (if one model \emph{has to} be true), while stacking (through the approximation  Pr$(\mathcal{J}_k)$) computes  the probability of a model being the \emph{best}.

 \section{Proofs of theorems}\label{sec_proof}

 For briefly, in later proofs we will use the abbreviation for the posterior pointwise conditional predictive density from the $k$-th model:  $$ p_k(\tilde y | \tilde x)\coloneqq p(\tilde y | \tilde x, M_k) = \int p(\tilde y | \tilde x, \theta_k) p(\theta_k |\mathcal{D})d\theta_k , \quad k=1, \dots, K.$$
 
 This subscript index $k$ should not be confused with the notation $t$ as in $p_t(\tilde y | \tilde x)$ or $p_t(\tilde y, \tilde x)$: the unknown conditional or joint density of the true data generating process. The subscript letter ${t}$ is always reserved for ``\emph{true}".

 Recall that in this section $\w^{\mathrm{stacking}}$ refers to the complete-pooling stacking in the  population:
 \begin{align*}\label{eq_stacking_population}
 	\w^{\mathrm{stacking}}&\coloneqq \arg\max_{\w\in \mathcal{S_K}} \mathrm{elpd} (\w),\nonumber\\ 
 	\mathrm{elpd} (\w)&=\int_{\mathcal{X}  \times \mathcal{Y}} \log\left( \sum_{k=1}^K  w_k p(\tilde y| M_k, \tilde x)\right) p_t(\tilde y,  \tilde x) d\tilde y d\tilde x. 
 \end{align*}
 
 \begin{repthm}
 	We call $K$ predictive densities $\{p(\tilde y =\cdot | \tilde x=\cdot, M_k)\}_{k=1}^K$ to be locally separable with a constant pair  $L>0$ and  $0<\epsilon<1$ with respect to the true data generating process  $p_t(\tilde y, \tilde x)$, if 
 	\begin{equation*} 
 		\sum_{k=1}^K \int_{(\tilde x, \tilde y)\in \mathcal{J}_k  } \mathbbm{1}\Big(  \log p(\tilde y|\tilde x, M_k)<   \log p(\tilde y|\tilde x, M_{k^{\prime}})+ L,  \; \forall    k^{\prime} \neq k \Big)  p_t(\tilde y,  \tilde x) d\tilde y d\tilde x \leq \epsilon.
 	\end{equation*}
 	For a small $\epsilon$ and a large $L$, the stacking weights that solve
 	\eqref{stacking_obj}  is  approximately  the proportion of the model being the locally best model:  $$\w^{\mathrm{stacking}}_k  \approx  \w^{\mathrm{approx}}_k \coloneqq \Pr(\mathcal{J}_k)= \int_{\mathcal{J}_k}   p_t(\tilde y | \tilde x ) p( \tilde x )d\tilde y d\tilde x.$$ in the sense  that the objective function is nearly optimal: 
 	\begin{equation*} 
 		|\mathrm{elpd}(\w^{\mathrm{approx}}) - \mathrm{elpd}(\w^{\mathrm{stacking}})| \leq   \mathcal{O}(\epsilon + \exp(-L)).
 	\end{equation*} 
 \end{repthm}

 \begin{proof}
 	The expected log predictive density of the weighted prediction $\sum_{k} w_k p_k(\cdot | x)$ (as a function of $\w$) is
 	\begin{align*}
 		\mathrm{elpd}(\w)&= \int_{\mathcal{X}\times\mathcal{Y} } \log \left(\sum_{l=1}^K w_l p_l(\tilde y |\tilde x)  \right)p_t(\tilde y|\tilde x) p(\tilde x) d\tilde x d\tilde y\\
 		&=\sum_{k=1}^K \int_{\mathcal{J}_k } \log \left(\sum_{l=1}^K w_l p_l(\tilde y |\tilde x)  \right)p_t(\tilde y|\tilde x) p(\tilde x) d\tilde x d\tilde y\\
 		&=\sum_{k=1}^K  \int_{\mathcal{J}_k } \log \left(w_k p_k(\tilde y |\tilde x) + \sum_{l\neq k} w_l p_l(\tilde y |\tilde x)  \right)
 		p_t(\tilde y|\tilde x) p(\tilde x) d\tilde x d\tilde y\\
 		&=\sum_{k=1}^K  \int_{\mathcal{J}_k } \left( \log \left(w_k p_k(\tilde y |\tilde x)\right)+\log   \left(1+\sum_{l\neq k}\frac{w_l p_l(\tilde y |\tilde x)}{w_k p_k(\tilde y |\tilde x)}
 		\right)\right)
 		p_t(\tilde y|\tilde x) p(\tilde x) d\tilde x d\tilde y.
 	\end{align*}
 	The expression is legit for any simplex vector $\w\in \mathcal{S}_K$  that  does not contain zeros. We will treat zeros later. For now we only consider a dense weight: $\{ \w\in \mathcal{S}_K: w_k>0, k=1, \dots K\}$.

 	Consider a surrogate objective function (the first term in the integral above): 
 	\begin{align*}
 		\mathrm{elpd}^{\mathrm{surrogate}}(\w)&=\sum_{k=1}^K  \int_{\mathcal{J}_k }   \log \left(w_k p_k(\tilde y |\tilde x)\right)
 		p_t(\tilde y|\tilde x) p(\tilde x) d\tilde x d\tilde y\\
 		&=\sum_{k=1}^K  \int_{\mathcal{J}_k }   \left( \log w_k + \log p_k(\tilde y |\tilde x)\right)
 		p_t(\tilde y|\tilde x) p(\tilde x) d\tilde x d\tilde y\\
 		&=\sum_{k=1}^K  \log w_k \int_{\mathcal{J}_k } p_t(\tilde y|\tilde x) p(\tilde x)d\tilde x d\tilde y + \sum_{k=1}^K  \int_{\mathcal{J}_k } \log p_k(\tilde y |\tilde x) p_t(\tilde y|\tilde x) p(\tilde x)d\tilde x d\tilde y\\
 		&=\sum_{k=1}^K \left( \Pr(\mathcal{J}_k)   \log w_k\right) + \mathrm{constant}.
 	\end{align*}

 	Ignoring the constant term above (the expected cross-entropy between each conditional prediction and the true DG), to maximize the  surrogate objective function is equivalent to maximizing  
 	$\sum_{k=1}^K \Pr(\mathcal{J}_k) \log w_k$, we call this function elbo$(\w)$, the \emph{evidence lower bound}. To optimize  $\mathrm{elpd}^{\mathrm{surrogate}}$ is equivalent to optimizing  elbo.  We show that this elbo function has a closed form optimum.
 	Using Jensen's inequality,
 	\begin{align*}
 		\mathrm{elbo}(\w)& =\sum_{k=1}^K \Pr(\mathcal{J}_k) \log w_k \\
 		& =  \sum_{k=1}^K \Pr(\mathcal{J}_k) \log \frac{w_k}{\Pr(\mathcal{J}_k)} +\sum_{k=1}^K \Pr(\mathcal{J}_k) \log {\Pr(\mathcal{J}_k)}\\
 		& \leq  \log \left( \sum_{k=1}^K \Pr(\mathcal{J}_k) \frac{w_k}{\Pr(\mathcal{J}_k)} \right)+\sum_{k=1}^K \Pr(\mathcal{J}_k) \log {\Pr(\mathcal{J}_k)}\\
 		&=\sum_{k=1}^K \Pr(\mathcal{J}_k) \log {\Pr(\mathcal{J}_k)}.
 	\end{align*}
 	The equality is attained at   $w_k= \Pr(\mathcal{J}_k),  ~k=1, \dots, K$,  which reaches our definition of $\w^{\mathrm{approx}}$ in Theorem~\ref{thm_stacking_seperate_y}. 
 	
 	What remains to be proved is that the  surrogate objective function is close to the actual objective.   
 	We  divide each set $\mathcal{J}_k$ into two disjoint subsets $\mathcal{J}_k =\mathcal{J}_\circ \cup \mathcal{J}_k^{\bullet}$, for 
 	\begin{align*}
 		&\mathcal{J}_k^\circ\coloneqq \left\{(\tilde x, \tilde y)\in  \mathcal{J}_k :\log p( \tilde y|\tilde x, M_k)<   \log p(\tilde y|\tilde x, M_{ k^\prime})+ L   \right\}; \\
 		&\mathcal{J}_k^\bullet\coloneqq \left\{(\tilde x, \tilde y)\in  \mathcal{J}_k :\log p( \tilde y|\tilde x, M_k)\geq   \log p(\tilde y|\tilde x, M_{ k^\prime})+ L   \right\}.
 	\end{align*}
 	The separation condition ensures  $\sum_{k=1}^K \Pr(\mathcal{J}_k^\circ)\leq \epsilon$.
 	
 	Let $\Delta (\w)=\mathrm{elpd}(\w)- \mathrm{elpd}^{\mathrm{surrogate}}(\w)$. For any fixed simplex vector $\w$,
 	this absolute difference of the objective function is bounded by 
 	\begin{align*}
 		|\Delta (\w)|&=  \left|\sum_{k=1}^K  \int_{\mathcal{J}_k }\left(\log   \left(1+\sum_{l\neq k}\frac{w_l p_l(\tilde y |\tilde x)}{w_k p_k(\tilde y |\tilde x)}
 		\right)\right)
 		p_t(\tilde y|\tilde x) p(\tilde x) d\tilde x d\tilde y  \right|\\
 		&\leq  \sum_{k=1}^K  \int_{\mathcal{J}_k }\left|\log   \left(1+\sum_{l\neq k}\frac{w_l p_l(\tilde y |\tilde x)}{w_k p_k(\tilde y |\tilde x)}
 		\right)\right|
 		p_t(\tilde y|\tilde x) p(\tilde x) d\tilde x d\tilde y   \\
 		&=  \sum_{k=1}^K \left( \int_{\mathcal{J}_k^\circ  }+ \int_{\mathcal{J}_k^\bullet } \right)\left|\log   \left(1+\sum_{l\neq k}\frac{w_l p_l(\tilde y |\tilde x)}{w_k p_k(\tilde y |\tilde x)}
 		\right)\right|
 		p_t(\tilde y|\tilde x) p(\tilde x) d\tilde x d\tilde y \\
 		&\leq \sum_{k=1}^K \int_{\mathcal{J}_k^\circ}  \log(1+ \sum_{l\neq k}\frac{w_l}{w_k} )p_t(\tilde y|\tilde x) p(\tilde x) d\tilde x d\tilde y  +  \sum_{k=1}^K \int_{\mathcal{J}_k^\bullet}  \sum_{l\neq k}\frac{w_l}{w_k} \frac{p_l(\tilde y |\tilde x)}{p_k(\tilde y |\tilde x)}
 		p_t(\tilde y|\tilde x) p(\tilde x) d\tilde x d\tilde y \\
 		&\leq  \left( \sum_{k=1}^K\sum_{l\neq k}\frac{w_l}{w_k}\right) \left(\sum_{k=1}^K \int_{\mathcal{J}_k^\circ}p_t(\tilde y|\tilde x)  p(\tilde x) d\tilde x d\tilde y +   \sum_{k=1}^K \int_{\mathcal{J}_k^\bullet}\frac{p_l(\tilde y |\tilde x)}{p_k(\tilde y |\tilde x)}
 		p_t(\tilde y|\tilde x) p(\tilde x) d\tilde x d\tilde y 
 		\right)\\
 		&\leq  \left( \sum_{k=1}^K \frac{1-w_k}{w_k}\right) (\epsilon+\exp(-L)).
 	\end{align*}
 	The second inequality used $\log(1+x)\leq x$ for $x\geq 0$.
 	
 	The exact optima of objective function is $\w^{\mathrm{stacking}}$. Using the inequality above twice,
 	\begin{align*}
 		0\leq\mathrm{elpd}(\w^{\mathrm{stacking}})- \mathrm{elpd}(\w^{\mathrm{approx}}) &\leq |\mathrm{elpd}^\mathrm{surrogate}(\w^{\mathrm{stacking}})-\mathrm{elpd}(\w^{\mathrm{stacking}})|\\
 		&\qquad + |\mathrm{elpd}^\mathrm{surrogate}(\w^{\mathrm{approx}})-\mathrm{elpd}(\w^{\mathrm{approx}})|\\ 
 		&\qquad + \mathrm{elpd}^\mathrm{surrogate}(\w^{\mathrm{stacking}})- \mathrm{elpd}^\mathrm{surrogate}(\w^{\mathrm{approx}})\\
 		&\leq |\Delta (\w^\mathrm{approx})|+|\Delta (\w^\mathrm{stacking})|\\
 		&\leq  \sum_{k=1}^K\left(  \frac{1-w_k^{\mathrm{approx}}}{w_k^{\mathrm{approx}}}+\frac{1-w_k^\mathrm{stacking}} {w_k^\mathrm{stacking}} \right)  (\epsilon+\exp(-L)).
 	\end{align*}
 	
 	It has almost finished the proof 
 	except for the simplex edge where $w_k^\mathrm{stacking}$ or $w_k^\mathrm{approx}$   attains zero. 
 	
 	Without loss of generality,  if $w_1^\mathrm{approx}=0, w_k^\mathrm{approx}\neq 0, \forall k\neq 1$, which means $p(\tilde y|M_k,  \tilde x)$ is always inferior to some other models. This will only happen if $p(\tilde y|M_k,  \tilde x)$ is almost sure zero (w.r.t $p_t(\tilde y|\tilde x)p(\tilde x)$) hence we can remove  model 1 from the model list, and the same  $\mathcal{O}(\epsilon+\exp(-L))$ bound applies to remaining model $2, \dots, K$. If there are more than one zeros,  repeat until all zeros have been removed. 
 	
 	Next, we deal with $w_1^\mathrm{stacking}=0, w_k^\mathrm{stacking}\neq 0, \forall k\neq 1$.  If $w_1^\mathrm{approx}=0$, too,  then we have solved in the previous paragraph.  If not,  Theorem \ref{thm_winning_prob}  shows that   $w_1^\mathrm{approx}$ has to be a small order term:
 	$$\Pr(\mathcal{J}_1) \leq  (1+ (\exp(L)-1) (1-\epsilon) +  \epsilon)^{-1} <   \exp(-L)+  \epsilon.$$
 	We leave the proof of this inequality in  Theorem \ref{thm_winning_prob}.   
 	
 	The contribution of the first model in the surrogate model is at most 
 	$\Pr(\mathcal{J}_1) \log  \Pr(\mathcal{J}_1)$. After we  remove the first model from the model list, with the surrogate model elpd changes by at most a small order term, not affecting the final bound.  Because the separation condition with constant $(\epsilon, L)$ applies to  model $1,\dots, K$,  and due to lack of a competition source, the  same separation condition  applies to model $2,\dots, K$ and the same bound applies.  
 	
 	% $\tilde w^{\mathrm{approx}}_{2:K} \coloneqq w^{\mathrm{approx}}_{2:K} / (1-\tilde w^{\mathrm{approx}}_{1})$ and 
 	% $\tilde w^{\mathrm{approx}}_{1} \coloneqq 0$

 	%Hypothetically we may still remove the first model from the model list. Because the separation condition with constant $(\epsilon, L)$ applies to  model $1,\dots, K$,  and due to lack of a competition source, the  same separation condition  applies to model $2,\dots, K$ by modifying
 	% $\tilde w^{\mathrm{approx}}_{2:K} \coloneqq w^{\mathrm{approx}}_{2:K} / (1-\tilde w^{\mathrm{approx}}_{1})$ and 
 	%$\tilde w^{\mathrm{approx}}_{1} \coloneqq 0$. We apply the previous bound to model $2,\dots, K$  and obtain $$\mathrm{elpd}(0, w^{\mathrm{stacking}}_{2:K})- \mathrm{elpd}(\tilde w^{\mathrm{approx}})\leq \mathcal{O}(\epsilon+\exp(-L)).$$
 	% Next by the optimality of $w^{\mathrm{approx}}$  we have $$\mathrm{elpd}^\mathrm{surrogate}(\tilde w^{\mathrm{approx}}) \leq \mathrm{elpd}^\mathrm{surrogate}( w^{\mathrm{approx}}).$$ Lastly $|\mathrm{elpd}^\mathrm{surrogate}( w^{\mathrm{approx}})|$. 
 	
 	%Hence $w^{\mathrm{approx}})$ can only be a better approximation and we
 \end{proof} 
 
 \begin{repthm} 
 	When the separation condition \eqref{eq_sep_y} holds,  and if the $k$-th model has zero weight in stacking, $w_k^\mathrm{stacking}=0$,
 	then the probability of its winning region is bounded by:
 	$$\Pr(\mathcal{J}_k) \leq  \left(1+ (\exp(L)-1) (1-\epsilon) +  \epsilon\right)^{-1}.$$ 
 	The right hand side can be further upper-bounded by  $\exp(-L)+  \epsilon$.
 \end{repthm}

 \begin{proof}
 	Without loss of generality, assume $w_1^\mathrm{stacking}=0$.  Let $p_0(\tilde y|\tilde x) = \sum_{k=2}^K \w^{\mathrm{stacking}} {p_k(\tilde y|\tilde x)}$.
 	Consider a constrained objective 
 	$\widetilde{\mathrm{elpd}}(w_1)=\E (\log (w_1 p_1(\tilde y | \tilde x) + (1-w_1)p_0(\tilde y|\tilde x)))$
 	where the expectation is over both $\tilde y$ and $\tilde x$ as before. Because the max is attained at $w_1=0$ and because   $\log(\cdot)$ is a concave function,
 	the derivative at any $w_1 \in [0,1]$ is
 	$$\frac{d}{d w_1} \widetilde{\mathrm{elpd}}(w_1)=
 	\E_{\tilde y, \tilde x}\left( \frac{p_1(\tilde y|\tilde x)-p_0(\tilde y|\tilde x) }{w_1 p_1(\tilde y|\tilde x) + (1-w_1) p_0(\tilde y|\tilde x)}\right)\leq 0. $$
 	That is
 	\begin{align*}
 		0\geq&  \E\left( \frac{p_1(\tilde y|\tilde x)- p_0(\tilde y|\tilde x)}{ p_0(\tilde y|\tilde x)}\right)\\
 		=&\Pr(\mathcal{J}_1)  \E[\frac{p_1(\tilde y|\tilde x)- p_0(\tilde y|\tilde x)}{ p_0(\tilde y|\tilde x)} |\mathcal{J}_1]+  (1-\Pr(\mathcal{J}_1))\E[\frac{p_1(\tilde y|\tilde x)- p_0(\tilde y|\tilde x)}{ p_0(\tilde y|\tilde x)} |\mathcal{J}_0]\\
 		\geq&\Pr(\mathcal{J}_1)  \E[\frac{p_1(\tilde y|\tilde x)- p_0(\tilde y|\tilde x)}{ p_0(\tilde y|\tilde x)} |\mathcal{J}_1]-(1-\Pr(\mathcal{J}_1) ). 
 	\end{align*}
 	Rearranging this inequality arrives at
 	\begin{align*}
 		1\geq& \Pr(\mathcal{J}_1) \left(1+\E[\frac{p_1(\tilde y|\tilde x)- p_0(\tilde y|\tilde x)}{ p_0(\tilde y|\tilde x)} |\mathcal{J}_1]\right)\\
 		\geq& \Pr(\mathcal{J}_1) (1+ (\exp(L)-1) (1-\epsilon) +  \epsilon).
 	\end{align*}
 	As a result, the model that has stacking weight zero cannot have a large probability to predominate all other models,
 	$$\Pr(\mathcal{J}_1) \leq  (1+ (\exp(L)-1) (1-\epsilon) +  \epsilon)^{-1} <   \exp(-L)+  \epsilon.$$
 \end{proof}

 \begin{repthm}
 	Let $\rho=\sup_{1\leq k \leq K} \Pr(\mathcal{J}_k)$,  and two  deterministic functions $g$ and $g^*$ by
 	\begin{align*}
 		&g(L,K, \rho, \epsilon)= L(1-\rho)(1-\epsilon)-\log K \\
 		\leq& g^*(L,K, \rho, \epsilon)= L(1-\rho)(1-\epsilon) + \rho \log(\rho)+ (1-\rho)(\log(1-\rho)- \log(K-1)).    
 	\end{align*}
 	Assuming the separation condition  \eqref{eq_sep_y} holds for all $k=1, \dots, K$, then the utility gain of stacking is further lower-bounded by
 	$$ \mathrm{elpd}_{\mathrm{stacking}} -  \mathrm {elpd}_k \geq   \max \left(g^*(L,K, \rho)+ \mathcal{O}(\exp(-L)+ \epsilon) , 0\right).$$
 \end{repthm}
 
 \begin{proof}
 	As before, we consider the approximate weights:
 	$w_k^{\mathrm{approx}}= \Pr(\mathcal{J}_k)$, and the surrogate elpd 
 	$ \mathrm{elpd}^{\mathrm{surrogate}}(\w)=\sum_{k=1}^K  \int_{\mathcal{J}_k }   \log \left(w_k p_k(\tilde y |\tilde x)\right)
 	p_t(\tilde y|\tilde x) p(\tilde x) d\tilde x d\tilde y.$ 
 	\begin{align*}
 		&\mathrm{elpd}^{\mathrm{surrogate}}(\w^{\mathrm{approx}})- \mathrm{elpd}_k\\
 		=&\sum_{l=1}^K  \int_{\mathcal{J}_l }   \log \left(\Pr({\mathcal{J}_l})  p_l(\tilde y |\tilde x)\right) p_t(\tilde y | \tilde x) p(\tilde x) d \tilde x d \tilde y - \sum_{l=1}^K  \int_{\mathcal{J}_l }   \log \left( p_k(\tilde y |\tilde x)\right) p_t(\tilde y | \tilde x) p(\tilde x) d \tilde x d \tilde y \\
 		=&\sum_{l=1}^K \int_{\mathcal{J}_l }  \left( \log \Pr({\mathcal{J}_l}) + \log p_l(\tilde y |\tilde x) -  \log p_k (\tilde y |\tilde x)  \right)  p_t(\tilde y | \tilde x) p(\tilde x) d \tilde x d \tilde y\\
 		=&\sum_{l=1}^K  \Pr({\mathcal{J}_l}) \log \Pr({\mathcal{J}_l}) + \sum_{l=1 }^K \mathbbm{1}(l\neq k)\int_{\mathcal{J}_l } \log (p_l(\tilde y |\tilde x) -  \log p_k (\tilde y |\tilde x) ) p_t(\tilde y | \tilde x) p(\tilde x) d \tilde x d \tilde y\\
 		=&\sum_{l=1}^K  \Pr({\mathcal{J}_l}) \log \Pr({\mathcal{J}_l}) + \sum_{l=1 }^K \mathbbm{1}(l\neq k) \left(\int_{\mathcal{J}_l^\circ } + \int_{\mathcal{J}_l^\bullet}\right)  \log (p_l(\tilde y |\tilde x)- \log p_k (\tilde y |\tilde x) )    p_t(\tilde y | \tilde x) p(\tilde x)  d \tilde x d \tilde y\\
 		\geq& \sum_{l=1}^K  \Pr({\mathcal{J}_l}) \log \Pr({\mathcal{J}_l})+ (1-\epsilon)(1-\rho)L -\epsilon\\
 		\geq& \rho\log \rho + (1-\rho) \log \frac{1-\rho}{K-1} + (1-\epsilon)(1-\rho)L -\epsilon\\
 		=& g^*(L, K, \rho, \epsilon) -\epsilon.
 	\end{align*}
 	The last inequality comes from the fact that, under the constraint of $\max_{k} \Pr({\mathcal{J}_k})=\rho$, the  entropy  $\sum_{k=1}^K  \Pr({\mathcal{J}_k}) \log \Pr({\mathcal{J}_k})$ 
 	attains its minimal when each of the $\Pr({\mathcal{J}_{k}})$ term equals   $(1-\rho)/(K-1)$ except for the largest term $\rho$.
 	This inequality is due to the convexity of $x\log x$. 
 	
 	Finally, using the proof of Theorem \ref{thm_stacking_seperate_y}, the error from the surrogate is bounded,  $$ |\mathrm{elpd}^{\mathrm{surrogate}}(\w^{\mathrm{approx}})- \mathrm{elpd}^{\mathrm{stacking}}(\w^{\mathrm{stacking}})| \leq \mathcal{O} (\exp(-L)+ \epsilon).$$
 	Hence the overall utility is bounded,
 	\begin{align*}
 		&\mathrm{elpd}_{\mathrm{stacking}} -  \mathrm {elpd}_k \\
 		&= \left(\mathrm{elpd}_{\mathrm{stacking}} -  \mathrm{elpd}^{\mathrm{surrogate}}(\w^{\mathrm{approx}})\right) + \left(\mathrm{elpd}^{\mathrm{surrogate}}(\w^{\mathrm{approx}})- \mathrm {elpd}_k \right)\\
 		&\geq  g^*(L,K, \rho)+ \mathcal{O}(\exp(-L)+ \epsilon).
 	\end{align*}
 	Because selection is always a specials case of averaging, the utility is further bounded below by 0.
 	
 	To replace $g^*(\cdot)$ with the looser bound $g(\cdot)$, we only need to ensure  $\rho\log \rho + (1-\rho) \log \frac {1-\rho}{K-1}\geq -\log K $, for the range $\rho \in [1/K,1), K\geq 2$. The proof is  elementary.
 	For any fixed $K\geq 2$,  let  $h(\rho) = \rho\log \rho + (1-\rho) \log \frac {1-\rho}{K-1} + \log K$. It is increasing on $\rho \in [1/K,1)$, for $\frac{d }{d \rho} h(\rho) = \log \frac{(K-1) \rho}{1-\rho} \geq 0$ .  Hence, $h(\rho)$ attains minimum at $\rho= 1/K$,  at which $h(1/K) =0$.
 \end{proof}
 
 From the constrictive example (Appendix \ref{theory_example}),   $g^* (\cdot)$  is a tight bound. We use the looser bound $g(\cdot)$ in the main paper for its simpler form.

 \begin{repthm}
 	Under the strong separation assumption 
 	$$
 	\sum_{k=1}^K \int_{\tilde x \in \mathcal{I}_k} \int_{\tilde y \in \mathcal{Y}}  \mathbbm{1}\Big(\log p(\tilde y|M_{k}, x )< \log p(\tilde y|M_{k^{\prime}}, x)+ L,  \; \forall    k^{\prime} \neq k^*(x)  \Big) p_t (\tilde y|x, D) d\tilde y d\tilde x \leq \epsilon, 
 	$$
 	and if the sets $\{\mathcal{I}_k\}$ are known exactly, then we can construct pointwise selection
 	$$p(\tilde y| x, \mathrm{pointwise~selection} ) = \sum_{k=1}^K\mathbbm{1} (x\in \mathcal{I}_k) p(\tilde y|x,  M_k).$$ Its utility gain is bounded from below by
 	%$$  \frac{\exp\left(  \mathrm{elpd}_{\mathrm{pointwise~ selection}}  \right) - \exp(\mathrm{elpd}_{\mathrm{stacking}} )}{\exp(\mathrm{elpd}_{\mathrm{stacking}} )}  \geq \frac{1-\rho}{\rho} $$
 	$$ \mathrm{elpd}_{\mathrm{pointwise~ selection}}    - \mathrm{elpd}_{\mathrm{stacking}}    \geq  -\log \rho_\mathcal{X} + \mathcal{O} (\exp(-L)+ \epsilon).$$
 \end{repthm} 
 \begin{proof}
 	\begin{align*}
 		&\mathrm{elpd}_{\mathrm{pointwise~ selection}}    - \mathrm{elpd}^{\mathrm{surrogate}}(\w^{\mathrm{approx}})\\
 		&=\sum_{l=1}^K 
 		\int_{\mathcal{I}_l } \left( \log  p_l(\tilde y |\tilde x)-
 		\log \left(\Pr({\mathcal{I}_l})  p_l(\tilde y |\tilde x)\right)  \right)p_t(\tilde y | \tilde x) p(\tilde x) d \tilde x d \tilde y\\
 		&=- \sum_{l=1}^K   \Pr({\mathcal{I}_l} ) \log \Pr({\mathcal{I}_l})\\
 		&\geq - \sum_{l=1}^K \Pr({\mathcal{I}_l} ) \log     \rho_x\\
 		&= -  \log \rho_x.
 	\end{align*}
 	Finally, from the proof of Theorem \ref{thm_stacking_seperate_y},$$ |\mathrm{elpd}^{\mathrm{surrogate}}(\w^{\mathrm{approx}})- \mathrm{elpd}^{\mathrm{stacking}}(\w^{\mathrm{stacking}})| \leq \mathcal{O} (\exp(-L)+ \epsilon).$$
 \end{proof}
 
 We close this section with two remarks. First, 
 \citet{yao2019bayesian} approximates the probabilistic stacking weights under the strong separation condition \eqref{eq_sep}. The result therein can be viewed as a special case of Theorem~\ref{theo_lower_selection} in the present paper as $\Pr(\mathcal{J}_k)\approx\Pr(\mathcal{I}_k)$ under assumption  \eqref{eq_sep}. 
 
 Second, most proofs only use the concavity of the log scoring rule.
 %: the function $g(w_1)=   \log(w_1p_1(y))$  is concave on $w_1$. 
 Therefore, some proprieties of stacking weights could be extended to other concave scoring rules, too. 
 
 \clearpage
 \section{Software implementation in Stan}\label{sec_stan}
 We summarize our formulation of hierarchical stacking by pseudo code \ref{eq_code}.

 \begin{algorithm}
 	\SetAlgoLined
 	\SetKwInOut{Input}{Input}
 	\KwData {	$y$: outcomes; $x$: input on which the stacking weights vary, $z$: other inputs; \\
 		$p_{k,-i}$: approximate leave-one-out predictive densities of the $k$-th model and $i$-th data. }
 	\KwResult{input-dependent stacking weight  $\overline \w(x): \mathcal{X} \to \mathcal{S}_K$ ; combined model.}
 	Sample  from the    joint densities   $p(\alpha, \mu, \sigma|\mathcal{D})$ in
 	hierarchical stacking model \eqref{stacking_obj_h2}\;
 	Compute posterior mean of $w_{k}(\tilde x)$ at any $\tilde x$,  and 
 	make predictions  $p(\tilde y| \tilde x, \tilde z )$   by \eqref{eq_point_stacking_final}.
 	\caption{Hierarchical stacking}\label{eq_code}
 \end{algorithm}

 To code the basic additive model, we prepare the input covariate $X= (X_\mathrm{discrete},   X_\mathrm{continuous})$,  where $X_\mathrm{discrete}$ is discrete dummy variable, and $X_\mathrm{continuous}$ are remaining features (already rectified as in \eqref{eq_relu}). The dimension of these two parts are $d_\mathrm{continuous}$ and  $d_\mathrm{discrete}$.
 
 Here we use the ``grouped hierarchical priors"   (Section \ref{sec_stacking_post}) with only two groups,  distinguishing between continuous and discrete variables.  We  discuss more on the hyper prior choice in the next section.
 \begin{align*}
 	&w_{1:K}(x)=  \mathrm{softmax}(w^*_{1:K}(x)), ~~ w^*_k(x)=   \sum_{m=1}^M \alpha_{mk} f_{m}(x)+  \mu_k,~~ k\leq K-1,  ~~w^*_{K}(x)=0,    \\
 	&  \alpha_{mk} \mid   \sigma_{k1} \sim \n (0, \sigma_{k1}),  ~ k=1, \dots, K-1, ~ m=1, \dots, d_\mathrm{discrete},\\
 	&  \alpha_{mk} \mid   \sigma_{k2} \sim \n (0, \sigma_{k2}),  ~ k=1, \dots, K-1, ~ m= d_\mathrm{discrete}+1, \dots, d_\mathrm{discrete}+d_\mathrm{continuous}, \\
 	&  \mu_k\sim \n(\mu_0, \tau_\mu),  \quad \sigma_{k1} \sim \n^+(0, \tau_{\sigma1}), \sigma_{k2} \sim \n^+(0, \tau_{\sigma2}), \quad  k=1, \dots, K-1.
 \end{align*}

 \paragraph{Stan code for hierarchical stacking.}
 Besides advantage listed in this paper,  another benefit of stacking  now being a  Bayesian model is the automated inference in generic computing programs, such as \texttt{Stan} \citep{stan2020}.  The following \texttt{Stan} program is one example of stacking with a linear additive form.
 \begin{lstlisting}[language=C++ ]
 	data {
 		int<lower=1> N; // number of observations
 		int<lower=1> d; //number of input variables
 		int<lower=1> d_discrete; // number of discrete dummy inputs
 		int<lower=2> K;  // number of models  
 		//when K=2, replace softmax by inverse-logit for higher efficiency
 		matrix[N,d] X;   // predictors 
 		//including continuous and discrete in dummy variables, no constant
 		matrix[N,K] lpd_point;  //the input pointwise  predictive density
 		real<lower=0> tau_mu;
 		real<lower=0> tau_discrete;//global regularization for discrete x
 		real<lower=0> tau_con;//overall regularization for continuous x
 	}
 	
 	transformed data {
 		matrix[N,K] exp_lpd_point = exp(lpd_point);
 	}
 	
 	parameters {
 		vector[K-1] mu; 
 		real mu_0;
 		vector<lower=0>[K-1] sigma;  
 		vector<lower=0>[K-1] sigma_con;
 		vector[d-d_discrete] beta_con[K-1];
 		vector[d_discrete] tau[K-1]; // using non-centered parameterization
 	}
 	
 	transformed parameters {
 		vector[d] beta[K-1];
 		simplex[K] w[N];
 		matrix[N,K] f;
 		for (k in 1:(K-1))
 		beta[k] = append_row(mu_0*tau_mu + mu[k]*tau_mu + sigma[k]*tau[k],
 		sigma_con[k]*beta_con[k]); 
 		for (k in 1:(K-1))
 		f[,k] = X * beta[k];
 		f[,K] = rep_vector(0, N);
 		for (n in 1:N)
 		w[n] = softmax(to_vector(f[n, 1:K]));
 	}
 	
 	model{
 		for (k in 1:(K-1)){
 			tau[k] ~ std_normal();
 			beta_con[k] ~ std_normal();
 		}
 		mu ~ std_normal();
 		mu_0 ~ std_normal();
 		sigma ~ normal(0, tau_discrete);
 		sigma_con ~ normal(0,tau_con);
 		for (i in 1:N) 
 		target += log(exp_lpd_point[i,] * w[i]); //log likelihood  
 	}
 	
 	//optional block: needed if an extra layer of LOO (eq.28) is called to evaluate the final stacked prediction.
 	generated quantities { 
 		vector[N] log_lik;
 		for (i in 1:N) 
 		log_lik[i] = log(exp_lpd_point[i,] * w[i]);
}
\end{lstlisting}
 To run this stacking program on model fits, we can fit all individual models in \texttt{Stan}, and extract their leave-one-out likelihoods $\{p_{k,-i}\}$. In \texttt{R}, we use the 
 efficient leave-one-out approximation package \texttt{loo} \citep{vehtari2018loo}:
 
 \begin{lstlisting} [style={R2},  language=R]  
 	library("loo")  # https://mc-stan.org/loo/
 	lpd_point <- matrix(NA, nrow(X), K)
 	for (k in 1:K) {
 		fit_stan <- stan(stan_model = model_k, data = ...)
 		# input x may differ in models
 		log_lik <- extract_log_lik(fit_stan, merge_chains = FALSE)
 		lpd_point[,k] <- loo(log_lik, 
 		r_eff = relative_eff(exp(log_lik)))$pointwise
}
\end{lstlisting}
 
 Finally, we run hierarchical stacking as a regular Bayesian model in \texttt{Stan}.
 \begin{lstlisting} [style={R2},  language=R]  
 	library("rstan")  # https://mc-stan.org/rstan/
 	# save the stan code above to a  file "stacking.stan".
 	stan_data <- list(X = X, N=nrow(X), d=ncol(X), d_discrete=d_discrete, 
 	lpd_point=lpd_point, K=ncol(lpd_point), tau_mu = 1,
 	tau_sigma = 1, tau_discrete= 0.5, tau_con = 1)
 	fit_stacking <- stan("stacking.stan",  data = stan_data)
 	w_fit <- extract(fit_stacking, pars = 'w')$w     # posterior simulation of pointwise stacking weights.
\end{lstlisting}

 \section{Prior recommendations}
 We believe the prior specification should follow the general principle of the weakly-informative prior\footnote{For example, see \url{https://github.com/stan-dev/stan/wiki/Prior-Choice-Recommendations}.}. In the context of the  additive model (Section \ref{sec_additive}), some weakly-informative prior heuristics imply
 \begin{itemize}
 	\item We  would like to use a half-normal instead of a too wide half-Cauchy or inverse-gamma for the model-wise scale parameter (i.e., $\sigma_k \sim \n^+(0, \tau_\sigma)$). This is not only because generally, we prefer half-normal for its lighter right tail in hierarchical models, but also because we know that the complete-pooling stacking  ($\sigma_k\equiv0$) is often a rational solution in many problems, to begin with. 
 	
 	On the contrary, a wide  $\sigma_k^2 \sim \mathrm{InvGamma}(10^{-2}, 10^{-4})$ seems a popular choice in the mixture of experts,  which we do not recommend.
 	\item When the number of features $M$ is large, it is sensible to first standardize  feature such that Var$(f_m(x))=1,  ~1\leq m\leq M$, and 
 	scale the  hyper-parameter to control 
 	Var$( \sum_{m=1}^M\alpha_{mk} f_{m}(x))$. With independent inputs, it leads to  $\tau_\sigma = \mathcal{O}(\sqrt{1/M})$.
 	\item  When there are a small number of features and no extra information to incorporate,  we often  first standardize all features and use a half-normal$(0,1)$ prior on model-wise scale $\sigma_k$ (i.e., $\tau_\sigma \coloneqq 1$). The  half-normal$(0,1)$ has been used as a default informative prior for group-level scale in some applied regression tasks.  
 	\item The structure of the prior matters more than the scale of the prior. Hierarchical stacking is typically not  sensitive to the difference between a half-normal$(0,1)$ or  half-normal$(0,2)$ hyper-prior on $\sigma_k$, although this sensitivity can be checked. But it would be sensitive to the structure of priors, such as feature-model decomposition, correlated priors,  and horseshoe priors, as we have discussed in Section 2.4.
 \end{itemize}
 
 Second, instead of  recommending a static default prior, we would rather adopt the attitude that the prior is part of the model and can be checked and improved. Because of our full-Bayesian formulation of hypercritical stacking, we do not have to reinvent model checking tools. When there are concerns on the prior specification, we would like to run prior predictive checks, sensitivity analysis by influence function or importance sampling, 
 and select, stacking, or hierarchically stack a sequence of priors based upon an extra layer of (approximate) leave-one-out cross validation \eqref{eq_st_loo}.

 \section{Experiment details} \label{app:polling}
 The replication code for experiments is available at\\
 \url{https://github.com/yao-yl/hierarchical-stacking-code}.
 
 \paragraph{Well-switch.}
 \citet{vehtari2017practical} and \citet{gelman2020bayesian}  used the same pointwise pattern (first panel in Figure \ref{fig:well_pattern}) in our   well-switch example  to demonstrate the heterogeneity of model fit. 
 The input contains both continuous $x_{\mathrm{con}}\in \R^D$ and categorical    $x_{\mathrm{cat}}\in\{ 1, \dots, 8\}$. As per previous discussion \eqref{eq_relu},  we convert all continuous  inputs $x_{\mathrm{con}}$ into two parts  $x_{\mathrm{con}, j}^+\coloneqq (x_{\mathrm{con},j}-  \mathrm{median}( x_{\mathrm{con}, j}))_{+}$  and $x_{\mathrm{con}, j}^-\coloneqq (x_{\mathrm{con}, j}- \mathrm{median}( x_{\mathrm{con}, j}))_{-}$.  We then model the unconstrained weight by a linear regression
 \begin{equation} 
 	\alpha_k(x) = \sum_{j=1}^D \left(\beta_{2j-1,k} x_{\mathrm{con}, j}^+  + \beta_{2j,k} x_{\mathrm{con}, j}^-\right) + z_k[x_{\mathrm{cat}}], ~k=1, \dots, 4;  \quad \alpha_5(x)=0.  
 \end{equation}
 And place a default prior on parameters and hyper-parameters.
 $$
 z_k[j]\sim \n(\mu_k, \sigma_k), ~~\beta_j,\mu_k\sim \n(0,1) , ~~ \sigma_k\sim \n^+(0,1). $$

 \paragraph{Gaussian process regression.}
 We use training data $\{x_i,y_i\}$ 
 from \citet{neal1998regression}  (file \texttt{odata.txt} in our repo).  \citet{yao2020stacking} use same setting to explain the benefit of complete-pooling stacking.   The training size is  $n=100$. We generate additional test data for model evaluation. The univariate input $x$ is distributed $\mbox{normal}(0,1)$, and the corresponding outcome $y$ is also Gaussian. The true but unknown conditional mean  is
 $$	
 \E_{\mathrm{true}}(y|x) =f_{\mathrm{true}}(x) = 0.3 + 0.4 x + 0.5 \sin(2.7x ) + 1.1 / (1+x^2).
 $$
 
 In the data generating process, with probability 0.95, $y$ is a realization from $y|f_{\mathrm{true}}= \n(\mathrm{mean}=f_{\mathrm{true}}, ~ \mathrm{sd}=0.1)$.  With probability 0.05, $y$ is considered an outlier and the  standard deviation is inflated to 1:  $y|f_{\mathrm{true}}= \n(f_{\mathrm{true}}, 1)$. This outlier probability is independent of location $x$, and the observational noises are mutually independent.

 To infer the parameter $\theta= (a, \rho, \sigma)$ in the first level GP model  
 $$	y_i = f(x_i)+\epsilon _i, ~  \epsilon _i \sim  \mbox{normal}(0, \sigma),  ~ f(x) \sim  \mathcal{GP} \left( 0, a^2 \exp\left( -\frac{(x-x')^2}{\rho^2} \right) \right). $$
 
 We integrate out all local  $f(x_i)$ and obtain the marginal posterior distribution
 $
 \log p(\theta | y) = - \frac{1}{2}  y^T \left( K(x, x ) + \sigma^2 I \right)  ^{-1} y  - \frac{1}{2} \log | K(x, x ) + \sigma^2I   | + \log p(\theta) + \mathrm{constant},
 $
 where $K$ is squared-exponential-kernel, and $p(\theta)$ is the prior for which we choose an elementwise  half-Cauchy$(0, 3)$. Using  initialization $(\log \rho, \log a, \log \sigma)= (1, 0.7, 0.1)$ and $(-1, -5, 2)$ respectively,  we find  two posterior modes of hyper-parameter $\theta= (a, \rho, \sigma)$.   
 
 The posterior multimodality  relies on the particular realization of $x$ and $y$.   We have tried other randomly generated training datasets, among which only  \citet{neal1998regression}'s original data realization can give rise to two distinct modes.
 We then consider three standard mode-based approximate inference:  
 \begin{itemize}[topsep=0pt,partopsep=1ex,parsep=1ex]
 	\item  {Type-II MAP}: The value $\hat \theta$ that maximizes the marginal posterior distribution. We further draw $f| \hat \theta, y$. 
 	\item  {Laplace approximation}.  First compute $\Sigma$: the inverse of the negative Hessian matrix of the log posterior density at the local mode $\hat \theta$,  draw $z$ from MVN$(0, I_{3})$, and use $\theta(z)= \hat \theta + \mathrm{V} \Lambda ^{1/2} z$ as the approximate posterior samples around the mode $\hat \theta$,  where the matrices $\mathrm{V},  \Lambda$ are from the eigen-decomposition $\Sigma= \mathrm{V} \Lambda ^{1/2} \mathrm{V}^T$.
 	
 	\item {Importance resampling}.  First draw $z$ from uniform$(-4,4)$,  resample $z$ without replacement with probability proportional to  $p\left( \theta (z) | y\right) $,  and use the kept samples of $\theta(z)$ as an approximation of $p(\theta|y)$.
 \end{itemize}
 
 With two local modes $\hat \theta_1, \hat \theta_2$, we either obtain two MAPs, or two  nonoverlapped draws, $(\theta_{1 s})_{s=1}^S,  (\theta_{2s})_{s=1}^S $. We  evaluate the  predictive distribution of $f$, $p_k(f|y, \theta)= \int\! p(f|y, \theta) q(\theta| \hat \theta_k) d \theta,~  k=1,2,$
 where $q(\theta| \hat \theta_k)$ is a delta function at the mode $\hat \theta_k$, or the draws from the Laplace approximation and importance resampling expanded at $\hat \theta_k$.

 In the model averaging phase, we form the model  weight in  GP prior stacking   by 
 $$w_{1}(x) = \mathrm{invlogit}(\alpha(x)),~
 \alpha(x) \sim \mathcal{GP}(0, \mathcal{K}(x)), ~ \mathcal{K} (x_i, x_j)=  a \exp(-  \left((x_i - x_j) / \rho \right)^2 ).$$
 Because input $x$ is distributed $\n(0,1)$, the length scale  $\rho$ should be constrained on a similar scale. 
 We use the following hyperprior for GP prior stacking:
 $$ \rho \sim  \mathrm{Inv\!\!-\!\!Gamma}(4,1),~~
 a \sim \n(0,1).   $$
 The $\mathrm{Inv\!\!-\!\!Gamma}(4,1) $ prior puts 98\% of mass on the interval $0.1 < \rho < 1.2$. 
 
 \paragraph{Election polling.}	
 In the  election example, we conduct a back-test for  one-week-ahead forecasts.  For example, if there are 20 polls between Aug 1 and Aug 7, we first fit each model on the data prior to Aug 1 and forecast for each of the 20 polls in this week. Next, we move on to forecasting for the week between Aug 8 and Aug 14. We use this step-wise approach for both, fitting the candidate models and stacking.
 
 We use two variants of hierarchical stacking with discrete inputs---first with independent priors from Eq. \eqref{eq_prior} and second with correlated priors from Eq. \eqref{eq_corr}. We place default priors on the hyperparameters in both variants:
 $$ \mu_k \sim  \mathrm{normal}(0,1),~~
 \sigma_k \sim \mathrm{normal}^+(0,1).   $$
 
 We are evaluating all combining methods on the same data, therefore we can compare them pointwisely by selecting a reference model---in our case this is the proposed hierarchical stacking and set it to be zero in all visualisations. For each combination method and each poll $i$, we compute the pointwise difference in elpds: $\text{elpd\_diff}_i^{\text{M}_j} = \text{elpd}_i^{\text{M}_j} - \text{elpd}_i^{\text{M}_\text{ref}}$, where $\text{M}_j$ is the $j$-th model and $\text{M}_\text{ref}$ is the reference model. Then we report the mean of this differences over all polls in the test data, $\text{elpd\_diff}^{\text{M}_j} = \frac{1}{N} \sum_{i} \text{elpd\_diff}_i^{\text{M}_j}$, where $N$ is the number of all polls. 
 
 In the main text, to account for non-stationarity discussed in Section \ref{sec:timeseries}, we only use the last four weeks prior to prediction day for training model averaging. In the end  we obtain a trajectory of  this back-testing performance of  hierarchical stacking,  complete-pooling, and no-pooling stacking and single model selection. The time window of four-week is a relatively ad-hoc choice and we did not tune it. 
 Figure \ref{fig:sixtyday} displays the result when trained with the previous 60 days rather than four weeks on each backtesting day. The pattern remains similar.
 
 Additionally, we are interested in how models perform depending on time, as there are few polls available in the early days of the election year, and then their number continuously increases toward election day. This results in noisier observations in the beginning. To suitably evaluate the combining methods, we compute the cumulative mean elpd at each day $d$, elpd$_d^{\text{*,M}_j} = \frac{1}{N_d}\sum_{j \leq d} \text{elpd}_j^{\text{M}_j}$, where $N_d$ is the number of conducted polls prior to or on day $d$. Then we compute the pointwise differences between these cumulative mean elpds of each method and the reference method: $\text{elpd\_diff}_d^{\text{*,M}_j} = \text{elpd}_d^{\text{*,M}_j} - \text{elpd}_d^{\text{*,M}_\text{ref}}$. 
 
 To get the elpd of a state, we take the average of all elpds in that state, for example $\text{elpd}_\text{NY} = \frac{1}{N_\text{NY}} \sum_{i \in A_\text{NY}} \text{elpd}_i$, where $N_\text{NY}$ is the number of polls conducted in New York, and $A_\text{NY}$ is the set of indexes of polls in New York. Figure \ref{fig:polling_cumulative_by_state} displays the state-level  log predictive  density of the combined model in six representative states.
 
 %\begin{figure}
 %	\centering
 %\includegraphics[width=0.8\linewidth]{acc_error.pdf}	\caption{\em Illustration of the backtest  evaluation of the election model. To give a sense of how the model performance change dynamically, for week $t=T_0, \dots, T$, we fit all models and run stacking using training data up to week $(t-1)$, and test the model performance of all individual polls on week $t$ in each state.}\label{fig_acc_error}
 % \end{figure}
 \begin{figure}  [!ht] 
 	\centering	\includegraphics[width=\textwidth]{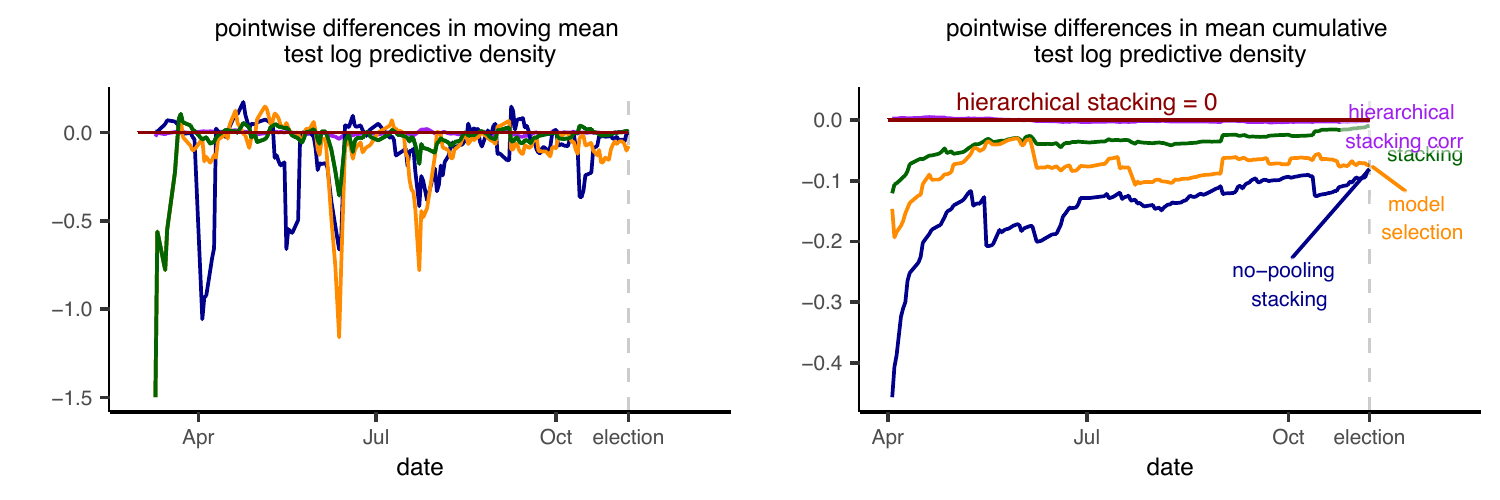}
 	\caption{\em Same pointwise model comparisons as in Figure \ref{fig:polling_results}, except this time all model averaging and selection methods are trained using the previous 60 days rather than 4 weeks on each backtesting day. The pattern remains similar.}\label{fig:sixtyday}
 \end{figure}
 \begin{figure} [!ht] 
 	\centering
 	\includegraphics[width=1\textwidth]{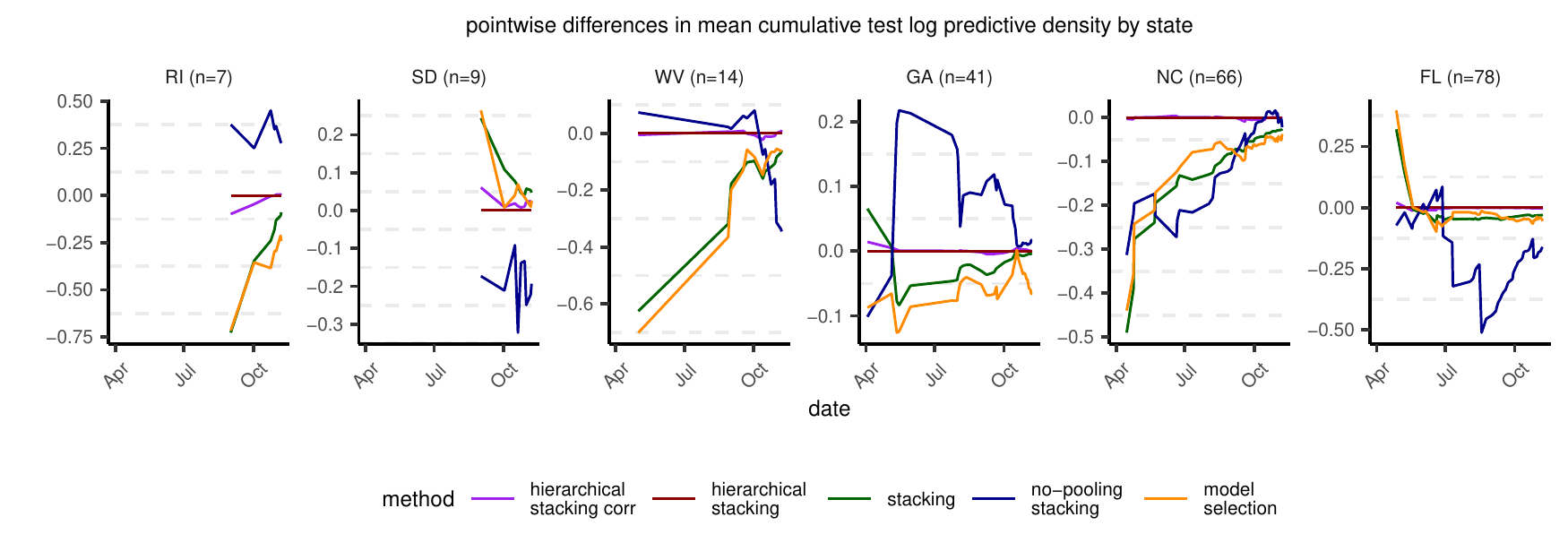}
 	\caption{\em  Log predictive  density of the combined model in three small  states (in the number of available state polls $n$, RI, SD, WV) and three swing states (GA, NC, FL). We fix the uncorrelated hierarchical stacking to be constant zero as a reference.
 		The number in the bracket is the total number of polls in that state.  With a large number of state polls available, for example, close to election day in Florida and North Carolina, no-pooling stacking performs well.   With fewer polls,  no-pooling stacking is unstable, as can be seen in Rhode Island, South Dakota, West Virginia, and the early part of Georgia plots.   Hierarchical stacking  alleviates this instability, while retaining enough flexibility for a good performance with large data come in.} \label{fig:polling_cumulative_by_state}
 \end{figure}

\end{document}